\documentclass[twocolumn,nofootinbib, superscriptaddress,10pt,aps,prx]{revtex4-2}
\usepackage[dvips]{graphicx} 
\usepackage{amsfonts}
\usepackage{csquotes} 
\usepackage{amssymb}
\usepackage{amscd}
\usepackage{amsmath}    
\usepackage{amsthm}
\usepackage{bm} 
\usepackage{bbm}
\usepackage{booktabs}
\usepackage{array}
\usepackage{enumerate}
\usepackage{enumitem}
\usepackage{epsfig}
\usepackage{subfigure}
\usepackage{subfloat}
\usepackage{xcolor}	
\usepackage{physics}
\usepackage[most]{tcolorbox}

\usepackage{tikz}
\usepackage{quantikz}
\usepackage{multirow}

\usepackage{tabularx}
\usepackage{float}
\usepackage{appendix}
\usepackage{makecell}
\usepackage{algpseudocode}  
\usepackage{dsfont}
\usepackage[colorlinks, linkcolor=blue, anchorcolor=blue, citecolor=blue]{hyperref}
\usepackage{MnSymbol}
\usepackage[linesnumbered,ruled,vlined]{algorithm2e}
\usepackage{natbib}
\usepackage{times}
\usepackage{tabularray}

\makeatletter
\newcommand{\printappendixtoc}{%
	\section*{Contents}    
	\@starttoc{atoc}       
}

\newcommand{\l@atocsection}{\@dottedtocline{1}{0em}{2.3em}}
\newcommand{\l@atocsubsection}{\@dottedtocline{2}{1.5em}{3em}}
\makeatother

\newtheorem{theorem}{Theorem}
\newtheorem{fact}{Fact}
\newtheorem{lemma}{Lemma}

\newtheorem{definition}{Definition}

\newtheorem{remark}{Remark}

\newcommand{\bE}{\mathbb{E}}
\newcommand{\bF}{\mathbb{F}}

\newcommand{\cD}{\mathcal{D}}

\newcommand{\cO}{\mathcal{O}}
\newcommand{\tilO}{\widetilde{\mathcal{O}}}

\definecolor{dred}{rgb}{.8,0.2,.2}
\definecolor{dblue}{rgb}{.2,0.2,.8}

\newtcolorbox[auto counter]{mybox}[2][]{
	enhanced,
	breakable,
	colback=blue!5!white,
	colframe=blue!75!black,
	fonttitle=\bfseries,
	title=Box \thetcbcounter: #2,#1
}

\begin{document}
	\title{Low-Depth Random Unitaries without Ancillae}
	\author{Zhenyu Du}
	\thanks{These authors contributed equally to this work.}
	\affiliation{Center for Quantum Information, Institute for Interdisciplinary Information Sciences, Tsinghua University, Beijing 100084, China}
	
	\author{Siyuan Cheng}
	\thanks{These authors contributed equally to this work.}
	\affiliation{Center for Quantum Information, Institute for Interdisciplinary Information Sciences, Tsinghua University, Beijing 100084, China}
	
	\author{Xiongfeng Ma}
	\email{xma@tsinghua.edu.cn}
	\affiliation{Center for Quantum Information, Institute for Interdisciplinary Information Sciences, Tsinghua University, Beijing 100084, China}
	
	\begin{abstract}
		Random unitaries are fundamental to quantum information and many-body physics, with widespread applications ranging from quantum learning and metrology to device benchmarking. 
		A central pursuit is to minimize the space and circuit depth required to generate them. 
		However, existing methods for generating low-depth random unitaries rely heavily on an extensive number of ancillary qubits, imposing severe spatial overhead. 
		In this work, we prove that random unitaries can be generated in optimal depth without ancillae. 
		For multiplicative-error approximate $k$-designs on $n$ qubits, our circuits achieve a depth of $\tilO(k) (\log n)^{1/\delta}$ on $\delta$-dimensional architectures and $\tilO(k) \log \log n$ with all-to-all connectivity. 
		Furthermore, by introducing a general exactification lemma, we lift our construction to optimal-depth exact $k$-designs, yielding an exponential resource reduction over state-of-the-art exact constructions. 
		Our results minimize the space-time costs for a wide range of quantum protocols.
	\end{abstract}
	
	\maketitle
	
	Random unitaries play a foundational role in modern quantum technologies and our understanding of quantum many-body physics. 
	They serve as the backbone for a wide range of quantum information protocols, including quantum learning~\cite{Brydges2019Probing, Huang2020Predicting, Huang2022QuantumAdvantage, Elben2020MixedStateEntanglement, Elben2023randomized, Li2025QueryOptimal, Zhou2020SingleCopyNegativity}, metrology~\cite{Zhu2018FisherSymmetric, Zhou2026RandomizedMetrology, Lu2026FisherRandomMeasurements, Du2026ComplexityDrivenTransition}, cryptography~\cite{Hayden2004RandomizingQuantumStates, Ji2018Pseudorandom, Bhattacharyya2026UncloneableEncryption}, random circuit sampling~\cite{Arute2019Supremacy, Zhong2020AdvantagePhotons}, and device benchmarking~\cite{Emerson2005NoiseEstimationRandomUnitary, Magesan2011ScalableRandomizedBenchmarking, Helsen2023GateSetProperty}. 
	Moreover, they provide essential models for investigating fundamental physical dynamics, such as quantum gravity~\cite{Patrick2007BlackHoles, Yasuhiro2008FastScramblers}, quantum chaos~\cite{Roberts2017ChaosDesign, Chan2018MinimalModel}, and information scrambling~\cite{Nahum2017EntanglementGrowth, Nahum2018OperatorSpreading, Mi2021InformationScrambling, Google2026ConstructiveInterference}.
	
	A central pursuit is therefore to understand the minimal resources required for physical quantum circuits to behave like genuine random unitaries~\cite{Harrow2009Random2Design, Brandao2016LocalRandomCircuits, Cleve2016ExactDesign, Nakata2021QuantumCircuitsExactDesigns, Haferkamp2022RandomQuantumCircuits, Harrow2023DesignLongRange, Chen2025Incompressibility, Schuster2024RandomUnitaries, Cui2025UnitaryDesignsOptimal, LaRacuente2026LowCommunication, Du2026ComplexityDrivenTransition, Liu2026QuadruplyOptimalUnitary, baer2026random, Haferkamp2023IndependentNonClifford, Zhang2026MagicAugmented, Leone2026NonCliffordCost, Bittle2026AdaptivelySecure}. 
	Achieving this not only minimizes the operational overhead for practical quantum information tasks but also deepens our understanding of complex many-body dynamics. 
	Yet implementing typical Haar-random unitaries is physically prohibitive, as it requires local gates that scale exponentially with system size $n$~\cite{Brandao2021ModelComplexity}. 
	Fortunately, many practical scenarios probe only low-order moments of the Haar measure. 
	Consequently, unitary $k$-designs, which replicate the statistical properties of the Haar measure up to the $k$-th moment, suffice for a wide range of applications.

	Extensive research has therefore focused on minimizing the circuit depth overhead required to generate such designs~\cite{Haferkamp2022RandomQuantumCircuits, Harrow2023DesignLongRange, Schuster2024RandomUnitaries, Chen2025Incompressibility, LaRacuente2026LowCommunication, Cui2025UnitaryDesignsOptimal, Zhang2026MagicAugmented, Du2026ComplexityDrivenTransition, Liu2026QuadruplyOptimalUnitary}. 
	A primary focus has been multiplicative-error unitary designs, which provide efficiency guarantees for diverse quantum applications while admitting relatively simple circuit implementations~\cite{Huang2020Predicting, Schuster2024RandomUnitaries, Brydges2019Probing, Du2026OptimalRandomized, Zhou2026RandomizedMetrology, Du2026ComplexityDrivenTransition, Emerson2005NoiseEstimationRandomUnitary, Dankert2009DesignFidelityEstimation, Magesan2011ScalableRandomizedBenchmarking, Chen2025Incompressibility}. 
	For these designs, seminal works have achieved poly-logarithmic depth scaling, $\cO((\log n)^{1/\delta})$, in $\delta$-dimensional architectures~\cite{Schuster2024RandomUnitaries, Du2026ComplexityDrivenTransition}, and doubly-logarithmic depth scaling, $\cO(\log \log n)$, in all-to-all circuits~\cite{Schuster2024RandomUnitaries, Cui2025UnitaryDesignsOptimal, Zhang2026MagicAugmented}.
	However, with the exception of specific one-dimensional constructions~\cite{Schuster2024RandomUnitaries, Cui2025UnitaryDesignsOptimal, Zhang2026MagicAugmented, Liu2026QuadruplyOptimalUnitary}, these protocols require an extensive number of ancillary qubits, often vastly exceeding the size of the original system. 
	Because spatial overhead is both a fundamental theoretical metric and a severe constraint in near-term quantum devices, a central open question emerges:
	\begin{center}
		\emph{Can we achieve depth-optimal random unitaries in all finite-dimensional and all-to-all architectures without using ancillary qubits?}
	\end{center}

	In this work, we answer this question in the affirmative (Table~\ref{tab:main-comparison}). We construct ancilla-free, multiplicative-error approximate unitary $k$-designs achieving a circuit depth of $\tilO(k) (\log n)^{1/\delta}$ for $\delta$-dimensional architectures, and $\tilO(k) \log \log n$ for all-to-all connectivity. 
	Crucially, these constructions achieve optimal depth scaling with system size $n$ and near-optimal scaling (up to logarithmic factors) with design order $k$ across all considered architectures, without introducing additional spatial overhead.

	Furthermore, we provide ancilla-free constructions for exact unitary $k$-designs, the strongest unitary designs that replicate the low-order moments of Haar measure with zero error. 
	Our circuits implement these exact designs in depth $\tilO(k) n^{1/\delta}$ on $\delta$-dimensional architectures and $\tilO(k) \log n$ for all-to-all connectivity. 
	This yields an exponential improvement over existing state-of-the-art exact constructions~\cite{Nakata2021QuantumCircuitsExactDesigns} while exhibiting optimal scaling in $n$ and near-optimal scaling in $k$. 
	
	
	Our constructions have immediate applications across various quantum information tasks, minimizing the resource costs for protocols ranging from classical shadow estimation~\cite{Huang2020Predicting, Schuster2024RandomUnitaries, Li2025QueryOptimal, Helsen2023GateSetProperty} and metrological measurements~\cite{Zhou2026RandomizedMetrology, Du2026ComplexityDrivenTransition} to broader randomized measurement schemes~\cite{Elben2019StatisticalCorrelations, Brydges2019Probing, Du2026OptimalRandomized, Li2026SingleMeasurement} and higher-order randomized benchmarking~\cite{Nakata2021QuantumCircuitsExactDesigns}.
	
	We first briefly introduce the concept of unitary designs. For any ensemble $\mathcal{E}$ of unitaries, its $k$-th moment channel is defined as
	\begin{equation}
		\Phi^{(k)}_{\mathcal E}(A)
		:=
		\mathbb E_{U\sim\mathcal E}
		\left[
		U^{\otimes k}A\,U^{\dagger\otimes k}
		\right].
		\label{eq:moment-channel}
	\end{equation}
	The ensemble $\mathcal{E}$ forms an exact unitary $k$-design if this channel perfectly matches the corresponding Haar average, namely, $\Phi^{(k)}_{\mathcal E} = \Phi^{(k)}_{\mathrm H}$. Operationally, this implies that no experiment making at most $k$ queries can distinguish $U \sim \mathcal{E}$ from a Haar-random unitary. 
	Relaxing this exact requirement, $\mathcal{E}$ forms an approximate unitary $k$-design if its moments merely approximate those of the Haar ensemble. Here, we adopt the notion of multiplicative $\epsilon$-approximate $k$-designs~\cite{Brandao2016LocalRandomCircuits}:
	\begin{equation}
		(1-\epsilon)\Phi^{(k)}_{\mathrm H}
		\preceq
		\Phi_{\mathcal E}^{(k)}
		\preceq
		(1+\epsilon)\Phi^{(k)}_{\mathrm H},
		\label{eq:multiplicative-design}
	\end{equation}
	where $\Phi\preceq\Phi'$ indicates that $\Phi'-\Phi$ is a completely positive map. Eq.~\eqref{eq:multiplicative-design} guarantees that any such $k$-query experiment behaves approximately identically to one using Haar-random unitaries~\cite{Schuster2024RandomUnitaries}.
	
	
	\begin{table}[!b]
		\centering
		\footnotesize
		\setlength{\tabcolsep}{3pt}
		\renewcommand{\arraystretch}{1.2}
		\begin{tabular}{@{}lll@{}}
			\toprule
			Construction & Circuit Depth & Ancillas \\
			\midrule
			\multicolumn{3}{@{}l}{Multiplicative $\epsilon$-approximate $k$-designs} \\
			\multicolumn{3}{@{}l}{\textit{$\delta$-dimensional architectures}} \\
			Ref.~\cite{Liu2026QuadruplyOptimalUnitary} (1D)
			& $\cO\left(k\log k + \log(n\epsilon^{-1})\right)$
			& $0$ \\
			Ref.~\cite{Du2026ComplexityDrivenTransition}
			& $\tilO(k^{1+1/\delta})[\log(n\epsilon^{-1})]^{1/\delta}$
			& $\cO(kn)$ \\
			\textbf{This work}
			& $\tilO(k)[\log(n\epsilon^{-1})]^{1/\delta}$
			& $\mathbf{0}$ \\
			\addlinespace
			\multicolumn{3}{@{}l}{\textit{All-to-all connectivity}} \\
			Ref.~\cite{Cui2025UnitaryDesignsOptimal}
			& $\tilO(k)\log\log(n\epsilon^{-1})$
			& $kn\cdot\tilO\!\left(\log\log(n\epsilon^{-1})\right)$ \\
			Ref.~\cite{Zhang2026MagicAugmented}
			& $2^{\tilO(k)} + \cO(\log\log(n\epsilon^{-1}))$
			& $\cO\left(n k^2 + n \log(n\epsilon^{-1})\right)$ \\
			\textbf{This work}
			& $\tilO(k) \log\log(n\epsilon^{-1})$
			& $\mathbf{0}$ \\
			\midrule
			\multicolumn{3}{@{}l}{Exact $k$-designs} \\
			\multicolumn{3}{@{}l}{\textit{$\delta$-dimensional architectures}} \\
			Ref.~\cite{Du2026ComplexityDrivenTransition} ($k=2$)
			& $\cO(n^{1/\delta})$
			& $\cO(n)$ \\
			Ref.~\cite{Maslov2007LinearDepthStabilizer} ($k=3$)
			& $\cO(n)$
			& $0$ \\
			\textbf{This work}
			& $\cO(n^{1/\delta}k\log^2 k)$
			& $\mathbf{0}$ \\
			\addlinespace
			\multicolumn{3}{@{}l}{\textit{All-to-all connectivity}} \\
			Ref.~\cite{Jiang2020SpaceDepthTradeoffCNOT} ($k=3$)
			& $\cO(\log n)$
			& $\cO(n^2/\log^2 n)$ \\
			Ref.~\cite{Nakata2021QuantumCircuitsExactDesigns}
			& Gate count $\exp[\cO(\sqrt{k}n)]$
			& Required \\
			\textbf{This work}
			& $\cO(k\log k(\log n+\log k))$
			& $\mathbf{0}$ \\
			\bottomrule
		\end{tabular}
		\caption{Comparison of space-time resource costs for $n$-qubit approximate and exact unitary designs against current state-of-the-art constructions. For a more comprehensive comparison including additional historical baselines, see Appendix~\ref{app:comparison}.}
		\label{tab:main-comparison}
	\end{table}
	
	Our first result provides low-depth constructions for approximate unitary designs on both finite-dimensional and all-to-all architectures, without using any ancillary qubits. Formal statements are provided in Theorems~\ref{thm:fd-large-design} and \ref{thm:a2a-large-design}.
	\begin{theorem}[Low-depth approximate designs without ancillae, informal]
		\label{thm:approximate-design-informal}
		An $n$-qubit multiplicative $\epsilon$-approximate unitary $k$-design can be implemented with a circuit depth of $\tilO(k) \left(\log (n\epsilon^{-1})\right)^{1/\delta}$ on $\delta$-dimensional architectures, and a depth of $\tilO(k) \log \log (n\epsilon^{-1})$ on all-to-all architectures, where $\tilO(k)$ hides logarithmic factors. These implementations require no additional ancillary qubits.  
	\end{theorem}
	
	Generating such approximate designs without ancillae requires depth $\Omega\left(k + (\log n)^{1/\delta}\right)$ on $\delta$-dimensional architectures~\cite{Du2026ComplexityDrivenTransition} and $\Omega\left(k + \log \log n\right)$ on all-to-all architectures~\cite{Cui2025UnitaryDesignsOptimal}. 
	Thus, our construction achieves the optimal dependence on the system size $n$, while its $k$-dependence is optimal up to logarithmic factors. 
	Previously, this ancilla-free optimal scaling in $n$ was known only in 1D architectures~\cite{Schuster2024RandomUnitaries}. 
	For finite-dimensional~\cite{Du2026ComplexityDrivenTransition} and all-to-all~\cite{Schuster2024RandomUnitaries, Cui2025UnitaryDesignsOptimal} architectures, existing implementations require $\omega(n)$ ancillary qubits, an overhead that can significantly exceed the system size itself (Table~\ref{tab:main-comparison}). 
	Our result, for the first time, achieves optimal depth with zero spatial overhead simultaneously across all these architectures.
	
	Our second result shows that exact unitary designs can also be implemented using low-depth, ancilla-free circuits. Formal result is given in Theorem~\ref{thm:exact-design-formal}.
	\begin{theorem}[Low-depth exact designs without ancillae, informal]
		\label{thm:exact-design-informal}
		An $n$-qubit exact unitary $k$-design can be implemented with circuit depth $\tilO(k)n^{1/\delta}$ on $\delta$-dimensional architectures, and depth $\tilO(k)\log n$ on all-to-all architectures. These implementations require no additional ancillary qubits.
	\end{theorem}
	
	Except for specific low-order exact designs~\cite{Cleve2016ExactDesign, Zhu2017MultiqubitClifford, Zhu2016FailsGracefully, Du2026ComplexityDrivenTransition}, prior state-of-the-art constructions of general exact unitary $k$-designs incur a two-qubit gate count that scales exponentially with both $k$ and $n$~\cite{Nakata2021QuantumCircuitsExactDesigns}. 
	In contrast, Theorem~\ref{thm:exact-design-informal} achieves exact designs using ancilla-free circuits of polynomial depth and size, yielding an exponential reduction in total gate count.
	This is achieved via a general exactification lemma (Lemma~\ref{lem:exactification-informal}), which lifts any sufficiently accurate approximate design to an exact one by suitably reweighting its classical probability distribution. 
	Moreover, this depth scaling is optimal in the system size $n$ (via standard light-cone arguments~\cite{Cleve2016ExactDesign}) and near-optimal in the design order $k$~\cite{Cui2025UnitaryDesignsOptimal}.
	We emphasize, however, that evaluating or sampling from this reweighted classical distribution currently demands exponential classical computation time.
	
	Because an exact unitary $k$-design perfectly replicates the Haar average for any operator polynomial in $U$ and $U^{\dagger}$ up to degree $k$, it also forms an exact strong unitary $k$-design~\cite{Schuster2025StrongUnitaries} that remains indistinguishable from Haar-random unitaries under queries to $U$, $U^{\dagger}$, $U^{*}$, and $U^T$. 
	Consequently, Theorem~\ref{thm:exact-design-informal} yields ancilla-free strong unitary designs that match state-of-the-art constructions across both finite-dimensional~\cite{Folkertsma2026ArtsCrafts} and all-to-all~\cite{Parelladilme2026StrongUnitaryOptimalSpace} architectures, while simultaneously achieving zero approximation error.
	
	\begin{figure*}[t]
		\centering
		\includegraphics[width=0.85\linewidth]{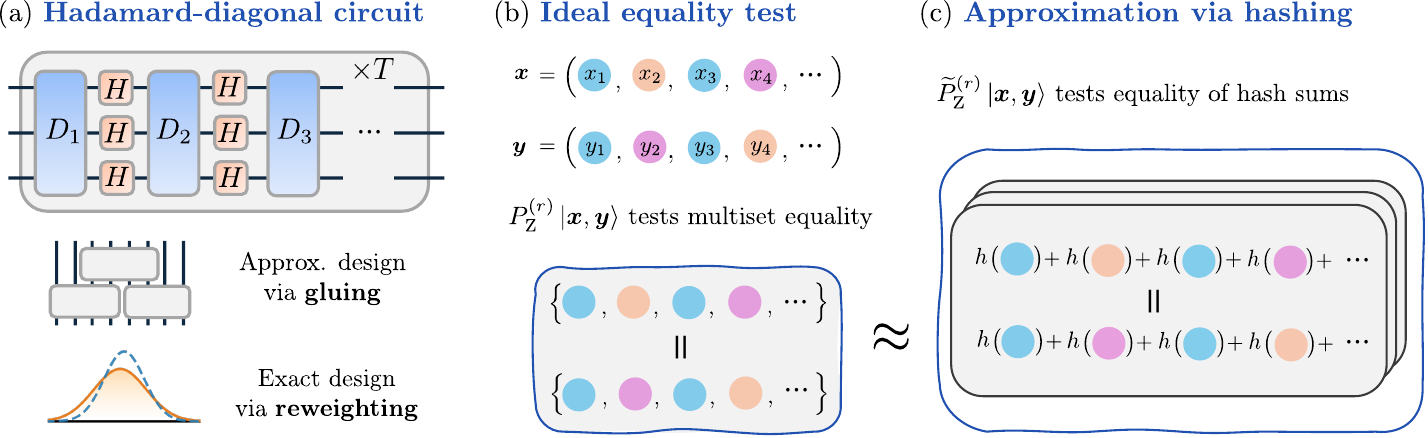}
		
		\caption{
			Overview of the constructions.
			(a) A Hadamard–diagonal block consists of three independently sampled diagonal unitaries $D_1,D_2,D_3$ interleaved with two Hadamard layers. 
			Repeated applications of this block further reduce its approximation error. 
			Gluing such local ensembles yields low-depth approximate designs, whereas reweighting the sampling distribution of a sufficiently accurate approximate design lifts it to an exact design without modifying the underlying circuits.
			(b) The $r$-th moment $P_{\rm Z}^{(r)}$ of the ideal diagonal ensemble performs an exact multiset-equality test, retaining $\ket{\bm{x},\bm{y}}$ if and only if the two tuples contain the same multiset of labels.
			(c) We approximate this ideal test using multiple randomized hash functions $h_j$. Equal multisets pass every hash-sum test, whereas distinct multisets pass all tests only with small probability.
		}
		\label{fig:summary}
	\end{figure*}
	
	At a high level, our construction starts from the Hadamard--diagonal ensemble (Fig.~\ref{fig:summary}a)~\cite{Nakata2017TimeIndependentHamiltonian}. 
	Its ideal diagonal unitaries are exponentially costly, but their moments admit a simple interpretation as multiset-equality tests (Fig.~\ref{fig:summary}b).
	We approximate these tests using randomized hash-sum tests (Fig.~\ref{fig:summary}c), which can be implemented without ancillae through \emph{catalytic computation}. The underlying linear hashing primitives are tailored to the circuit architecture using either finite-field arithmetic or linear codes.
	Repeating and gluing these ensembles then produces the low-depth approximate designs (Fig.~\ref{fig:summary}a). 
	Finally, a classical reweighting of the sampling distribution lifts this construction to an exact design without modifying the underlying quantum circuits (Fig.~\ref{fig:summary}a).

	\emph{Hadamard--diagonal ensemble}---For an $m$-qubit system, the Hadamard--diagonal ensemble $\nu_{\mathrm{HD}}$ interleaves random diagonal unitaries with layers of Hadamard gates:
	\begin{equation}
		U_{\mathrm{HD}} = D_3H^{\otimes m}D_2H^{\otimes m}D_1.
		\label{eq:HD-block}
	\end{equation}
	Here, the operators $D_1, D_2,$ and $D_3$ are drawn independently from the random diagonal ensemble $\cD_{\mathrm Z}$, where each sample takes the form $D = \sum_{\xi\in\{0,1\}^m} e^{i\theta_\xi} \ketbra{\xi}$, with the phases $\theta_\xi$ drawn i.i.d.\ uniformly from $[0,2\pi)$.

	We first establish a substantially sharper moment bound for this ensemble than previously known. To quantify this, we analyze its tensor-product-expander (TPE) error $\lambda_r(\nu_{\mathrm{HD}})$, which measures the gap $\lambda_r(\nu) \coloneqq \| M_r(\nu)-P_{\mathrm H}^{(r)} \|_{\mathrm{op}}$ between the ensemble's $r$-th moment operator $M_r(\nu) \coloneqq \mathbb E_{U\sim\nu} [U^{\otimes r}\otimes\overline U^{\otimes r}]$ and the ideal Haar moment $P_{\mathrm H}^{(r)}$.
	This quantity directly controls the multiplicative design error, up to a dimension-dependent factor.
	We prove that
	\begin{equation}\label{eq:ideal-fourier-gap}
		\lambda_r(\nu_{\mathrm{HD}}) \le \frac{9r^4}{2^m}
	\end{equation}
	for any moment order up to $r = 2^{\cO(m)}$. 
	This bound exponentially improves the previous factorial dependence on $r$~\cite{Nakata2017TimeIndependentHamiltonian}, a reduction that is crucial for minimizing the resource requirements of high-order designs.
	
	However, because the ideal ensemble $\cD_{\mathrm Z}$ assigns an independent phase to each of the $2^m$ computational-basis states $\ket{\xi}$, implementing it generally incurs an exponential gate cost. To overcome this complexity barrier, we replace the ideal ensemble $\cD_{\mathrm Z}$ with a low-depth approximation $\cD_m$ (introduced below). 
	Fortunately, the moment gap remains robust under such approximation. Let $P_{\mathrm Z}^{(r)} := M_r(\cD_{\mathrm Z})$ and $\widetilde P_{\mathrm Z}^{(r)} := M_r(\cD_m)$ denote the moment operators of the ideal and approximate diagonal ensembles, respectively. Define their approximation error as $\epsilon_{\mathrm{diag}}(r) := \| \widetilde P_{\mathrm Z}^{(r)} - P_{\mathrm Z}^{(r)} \|_{\mathrm{op}}$.
	Let $\nu_{\mathrm{HD}}^{(m)}$ denote the ensemble obtained from Eq.~\eqref{eq:HD-block} by drawing $D_1,D_2,D_3$ independently from $\cD_m$, it then follows that 
	\begin{equation}
		\lambda_r\!\left(\nu_{\mathrm{HD}}^{(m)}\right)
		\leq
		\frac{9r^4}{2^m}
		+
		3\epsilon_{\mathrm{diag}}(r).
		\label{eq:HD-master-gap}
	\end{equation}
	The detailed proofs are given in Appendix~\ref{app:spectral-Hadamard-diagonal}.
	We now proceed to construct low-depth realizations of the approximate ensemble $\cD_m$.

	\emph{Approximation via feature extraction}---The operator $P_{\mathrm Z}^{(r)}$ acts as a equality test: for any $r$-fold basis state $\ket{\bm{x}, \bm{y}}$, where $\bm{x}=(x_1,\ldots,x_r)$ and $\bm{y}=(y_1,\ldots,y_r)$ with $x_i, y_i \in \{0,1\}^m$, it preserves the state if and only if $\bm{x}$ and $\bm{y}$ comprise the exact same multiset of basis labels (Fig.~\ref{fig:summary}b).
	Our key insight is to approximate this exact test using a much smaller set of $s = \Theta(m)$ Boolean features.
	Specifically, we construct diagonal unitaries of the form
	\begin{equation}
		D|\xi\rangle
		=
		\exp\left(
		\mathrm{i}\sum_{j=1}^{s}\theta_jh_j(\xi)
		\right)|\xi\rangle,
		\label{eq:hashed-diagonal}
	\end{equation}
	where the phases $\theta_j$ are drawn independently and uniformly from $[0,2\pi)$, and $h_j:\{0,1\}^m\rightarrow\{0,1\}$ are carefully engineered hash functions. 
	Averaging over these random angles retains only the states where the aggregate sum of each feature matches perfectly between $\bm{x}$ and $\bm{y}$, meaning $\sum_{i=1}^r (h_j(x_i) - h_j(y_i)) = 0$ for all $j$ (Fig.~\ref{fig:summary}c). 
	Consequently, the approximation problem reduces to constructing efficiently implementable functions $h_j$ such that the mismatch between the multisets of $\bm{x}$ and $\bm{y}$ is detected (i.e., yields a nonzero sum for at least one feature) with high probability.
	
	Evaluating these randomized hash functions requires temporary memory for intermediate results. To maintain an ancilla-free implementation, we employ \emph{catalytic computation}. 
	Rather than allocating fresh ancillae, we borrow inactive system qubits as dirty workspace and restore them exactly after the computation, irrespective of their initial states.
	Specifically, we partition the system qubits into a constant number of registers, $Z_1, \ldots, Z_{r_0}$. 
	The first two, $X := Z_1$ and $Y := Z_2$, each of size $s = \Theta(m)$, serve as the active registers for the hashing computation.
	The remaining registers $Z_a$ (for $a \ge 3$), each of size at most $s$, function primarily as dirty ancillas. 
	To systematically manage these operations, we develop a dedicated compiler for the catalytic execution and composition of the necessary arithmetic (detailed in Appendix~\ref{app:finite-dimensional-depth}).
	The feature extraction then proceeds in three reversible steps.
	
	First, a concentration step accumulates the data from all $Z_a$ registers into $Y$ via randomized linear maps $L_a$, 
	\begin{equation}\label{eq:main-first-step}
		Y \gets Y \oplus \bigoplus_{a\neq 2} L_a(Z_a).
	\end{equation}
	Second, a hashing step maps this accumulated data into the $X$ register via another random linear map $M$ and a shift $\beta$, 
	\begin{equation}\label{eq:main-second-step}
		X \gets X \oplus M(Y) \oplus \beta.
	\end{equation}
	Together, these two steps map a randomized hash of the full $m$-bit label into the $s$-qubit register $X$.
	
	To better distinguish different multisets, we augment the linear hash functions with a quadratic layer.
	Specifically, we update each bit of $Y$ by adding $w=\Theta(\log k)$ random pairwise products of bits from $X$
	\begin{equation}\label{eq:main-third-step}
		Y_j \gets Y_j \oplus \bigoplus_{l=1}^w \gamma_{j,l} X_{\pi_l(j)} X_{\sigma_l(j)}
	\end{equation}
	using random permutations $\pi_l, \sigma_l$ and random coefficients $\gamma_{j,l}$. 
	As shown in the Appendices, these nonlinear features successfully detect any non-identical multiset of order up to $2k$ with high probability. 
	Finally, we apply random phases $\theta_j$ to the $s$ qubits of $Y$ and reverse steps~\eqref{eq:main-first-step}--\eqref{eq:main-third-step} to uncompute the hash and restore the computational basis.
	
	\emph{Architecture-tailored implementations}---The specific linear primitives ($L_a$ and $M$) are tailored to the underlying architecture to minimize circuit depth. 
	On finite-dimensional architectures, those hashing maps are implemented by random polynomials over $\mathbb F_{2^s}$,
	\begin{equation}
		L_a(z) = \lambda_a z, \qquad M(z) = \mu_0 z + \mu_1 z^2 + \mu_2 z^4,
	\end{equation}
	which are $\mathbb{F}_2$-linear functions because the field characteristic of $\mathbb F_{2^s}$ is $2$. 
	While low-depth implementations of such field arithmetic were previously introduced in Ref.~\cite{Du2026ComplexityDrivenTransition}, they required $\cO(m)$ ancillary qubits. Here, our catalytic compiler eliminates this spatial overhead while achieving the same depth scaling of $\cO(m^{1/\delta})$.
	
	On all-to-all architectures, we replace the finite-field arithmetic with a constant-rate, constant-distance binary linear error-correcting code~\cite{Spielman1996LinearTimeCodes}, denoted by $C$. 
	The linear maps $L_a$ and $M$ are constructed coordinate-wise by randomly sampling the encoded input. 
	Specifically, for a map $H \in \{L_a, M\}$ acting on an input $z$, its $j$-th output bit is given by
	\begin{equation}
		[H(z)]_j = \epsilon_j [C(z)]_{I_j},
	\end{equation}
	where $I_j$ is a uniformly sampled codeword index and $\epsilon_j \in \mathbb{F}_2$ is an independent random bit. 
	The code distance ensures that any nonzero input difference flips a constant fraction of the codeword, enabling the subsequent quadratic layer to detect the difference with high probability.
	The classical encoder for $C$ can be realized as an $\cO(\log m)$-depth, $\cO(m)$-size XOR circuit~\cite{Spielman1996LinearTimeCodes}. 
	We develop a catalytic compiler to implement such XOR networks using solely dirty ancillae while preserving the depth scaling (Lemma~\ref{lem:a2a-catalytic}). 
	Consequently, each such hash can be implemented in $\cO(\log m)$ depth without introducing additional spatial overhead.
	
	In both architectures, these hash families suppress the approximation error to
	\begin{equation}
		\epsilon_{\mathrm{diag}}(r) \leq \exp\left(-\Omega\left(\frac{m}{\log k}\right)\right)
	\end{equation}
	for any $m = \Omega(\log k)$. 
	Substituting this into Eq.~\eqref{eq:HD-master-gap} guarantees a small TPE error for the resulting Hadamard--diagonal block. 
	Repeating this block $T = \tilO(k)$ suppresses this error to the required multiplicative accuracy, while gluing overlapping blocks of size $m=\Theta(\log(nk))$ yields the global designs of Theorem~\ref{thm:approximate-design-informal}. 
	The complementary regime $m=\cO(\log k)$ is handled by random Pauli rotations~\cite{baer2026random}.
	Detailed proofs are provided in Appendices~\ref{app:finite-dimensional} and \ref{app:a2a}.
	
	\emph{Design exactification}---We now show how to construct low-depth exact designs. 
	Our key result shows that any sufficiently accurate approximate design can be lifted to an exact one solely by reweighting its classical sampling distribution, without modifying the underlying quantum circuits.
	
	\begin{lemma}[Design exactification, informal]
		\label{lem:exactification-informal}
		Let $\nu$ be an $n$-qubit unitary ensemble, and let $d = 2^n$. If its TPE errors at both orders $k$ and $2k$ satisfy
		\begin{equation}\label{eq:exactification-threshold}
			\max\{\lambda_k(\nu), \lambda_{2k}(\nu)\} \leq \frac{1}{32d^{4k}},
		\end{equation}
		then there exists a positive weight function $\omega(U) \in [3/4, 5/4]$ such that the reweighted probability measure $d\nu_{\mathrm{ex}}(U) := \omega(U)d\nu(U)$ forms an exact unitary $k$-design. 
	\end{lemma}
	
	We apply this exactification protocol to our previous constructions. 
	Unlike the approximate-design construction, which glues local patches, we apply the Hadamard--diagonal blocks globally across the entire $n$-qubit system. 
	Repeating $T = \tilO(k)$ independent blocks suffices to reach the exactification threshold in Eq.~\eqref{eq:exactification-threshold}, thereby establishing Theorem~\ref{thm:exact-design-informal}. 
	Detailed proofs are provided in Appendix~\ref{app:exactification}.
	
	\emph{Applications}---We now discuss several applications of our results. 
	The approximate $3$-designs in Theorem~\ref{thm:approximate-design-informal} can be directly applied to classical shadow tomography of quantum states~\cite{Huang2020Predicting} and quantum processes~\cite{Helsen2023GateSetProperty, Li2025QueryOptimal}. 
	The resulting protocols retain the prediction guarantees of global classical shadows~\cite{Schuster2024RandomUnitaries} while simultaneously achieving optimal circuit depth and zero spatial overhead. 
	Moreover, our constructions of higher-order designs enable thrifty shadow estimation, which significantly reduces the number of required experimental measurement settings~\cite{Helsen2023Thrifty}.
	
	Furthermore, approximate $3$-designs serve as universal readouts in quantum metrology, capable of converting a constant fraction of the pure-state quantum Fisher information to classical Fisher information~\cite{Zhou2026RandomizedMetrology, Du2026ComplexityDrivenTransition}.
	Ref.~\cite{Du2026ComplexityDrivenTransition} recently identified a sharp complexity-driven transition in this process: as the measurement circuit depth reaches the threshold required to implement an approximate $3$-design, the extractable classical Fisher information shifts abruptly from an exponentially small fraction to a constant fraction of the quantum Fisher information.
	While this metrological transition was previously established only for measurement circuits that rely on clean ancillary qubits, our ancilla-free constructions demonstrate that this transition persists in closed quantum systems across both finite-dimensional and all-to-all architectures.
	
	Randomized measurement protocols that target nonlinear properties often require both high design orders and many repetitions, so even minute design errors can accumulate across rounds.
	Exact designs eliminate this source of systematic error, making our low-depth constructions natural primitives for estimating nonlinear properties, ranging from the state purity $\tr(\rho^2)$~\cite{Brydges2019Probing} to generalized higher-order expectation values $\tr(O\rho^t)$~\cite{Zhou2024Hybrid, Liu2026AuxiliaryFree, Du2026OptimalRandomized, Li2026SingleMeasurement} and partial-transpose moments~\cite{Zhou2020SingleCopyNegativity}, essential for applications such as virtual cooling~\cite{Cotler2019VirtualCooling, Huggins2021VirtualDistillation, Koczor2021ExponentialSuppression} and entanglement detection and quantification~\cite{Elben2020MixedStateEntanglement, Liu2022PermutationMoments, Zhang2024PurityDetection}. 
	Furthermore, exact unitary $2t$-designs enable $t$-th-order randomized benchmarking, which probes higher-order properties of quantum noise beyond the average gate fidelity~\cite{Nakata2021QuantumCircuitsExactDesigns}.

	\emph{Discussion}---We have constructed ancilla-free approximate and exact unitary designs that exhibit optimal depth scaling with system size $n$ and near-optimal scaling with design order $k$. An important next step is to investigate whether high-order exact designs can be simultaneously quantumly efficient (low circuit depth) and classically efficient (easy to sample). 
	Furthermore, it remains an open question whether circuit depth can be further improved to decouple the scaling between $n$ and $k$. 
	Recently, a construction achieving a depth of $\cO(\log (n\epsilon^{-1}) + k \log k)$ for multiplicative $\epsilon$-approximate $k$-designs in the regime $k=\cO(n)$ was introduced in Ref.~\cite{Liu2026QuadruplyOptimalUnitary}. Generalizing this decoupled scaling to higher-dimensional architectures, all-to-all connectivity, and broader parameter regimes remains a promising avenue for future research.
	
	Beyond optimizing circuit depth, an important future direction is to address the practical resource overhead of implementing these random unitaries, for example, by minimizing the total count of non-Clifford magic gates~\cite{Haferkamp2023IndependentNonClifford, Zhang2026MagicAugmented, Leone2026NonCliffordCost, Bittle2026AdaptivelySecure, Liu2026QuadruplyOptimalUnitary}. 
	Another critical direction is determining the optimal depth at which more physically natural architectures, such as the canonical random brickwork circuits~\cite{Brandao2016LocalRandomCircuits, Nahum2017EntanglementGrowth, HunterJones2019StatisticalMechanics, Haferkamp2022RandomQuantumCircuits, Heinrich2026Anticoncentration}, form unitary designs, which is of fundamental importance to both many-body quantum physics and quantum information theory.
	
	\section*{Acknowledgments}
	We thank Guoding Liu for helpful discussions.
	This work was supported by the National Natural Science Foundation of China Grants No.~12575023,  the Innovation Program for Quantum Science and Technology Grant No.~2021ZD0300804,  No.~2021ZD0300702, the CCF-QuantumCtek Superconducting Quantum Computing Special Cooperation Program (Grant No.~CCF-QC2025005), and the Turing AI Institute of Nanjing.
	
	\bibliographystyle{ref_style}
	\bibliography{ref}
	
	\clearpage
	\onecolumngrid
	
	\clearpage
	\appendix
	
	\begin{center}
		{\large \textbf{Supplementary Material}}
	\end{center}
	
	\section*{Contents} 
	
	\makeatletter
	\let\oldaddcontentsline\addcontentsline
	\renewcommand{\addcontentsline}[3]{%
		\def\target{#1}%
		\def\toclabel{toc}%
		\ifx\target\toclabel
		\oldaddcontentsline{atoc}{atoc#2}{#3}%
		\else
		\oldaddcontentsline{#1}{#2}{#3}%
		\fi
	}
	
	\@starttoc{atoc}
	\makeatother
	
	\vspace{1cm}
	
	\renewcommand{\thetheorem}{S\arabic{theorem}}
	\renewcommand{\thefact}{S\arabic{fact}}
	\renewcommand{\thelemma}{S\arabic{lemma}}
	\renewcommand{\thedefinition}{S\arabic{definition}}
	\renewcommand{\theproposition}{S\arabic{proposition}}
	\renewcommand{\thecorollary}{S\arabic{corollary}}
	\renewcommand{\theclaim}{S\arabic{claim}}
	\renewcommand{\thepage}{S\arabic{page}}
	\renewcommand{\thefigure}{S\arabic{figure}}
	
	\renewcommand{\theHtheorem}{S\arabic{theorem}}
	\renewcommand{\theHfact}{S\arabic{fact}}
	\renewcommand{\theHlemma}{S\arabic{lemma}}
	\renewcommand{\theHdefinition}{S\arabic{definition}}
	\renewcommand{\theHproposition}{S\arabic{proposition}}
	\renewcommand{\theHcorollary}{S\arabic{corollary}}
	\renewcommand{\theHclaim}{S\arabic{claim}}
	\renewcommand{\theHfigure}{S\arabic{figure}}
	
	\setcounter{theorem}{0}
	\setcounter{fact}{0}
	\setcounter{lemma}{0}
	\setcounter{equation}{0}
	\setcounter{definition}{0}
	\setcounter{proposition}{0}
	\setcounter{claim}{0}
	\setcounter{corollary}{0}
	\setcounter{figure}{0}
	\setcounter{page}{1}
	\setcounter{section}{0}
	\setcounter{equation}{0}
	
	\section{Comparison with previous work}~\label{app:comparison}
	Tables~\ref{tab:multiplicative-design-comparison} and \ref{tab:exact-design-comparison} provide a detailed comparison of our multiplicative-error approximate and exact designs against state-of-the-art constructions.
	
	\begin{table}[htbp!]
		\centering
		\small
		\setlength{\tabcolsep}{6pt}
		\renewcommand{\arraystretch}{1.12}
		\begingroup
		\newcommand{\designname}[1]{\parbox[t]{6.75cm}{\raggedright #1\strut}}
		\newcommand{\designdepth}[1]{\parbox[t]{4.55cm}{\raggedright #1\strut}}
		\newcommand{\designancilla}[1]{\parbox[t]{5.05cm}{\raggedright #1\strut}}
		\begin{tabular}{@{}lll@{}}
			\toprule
			\designname{Construction} &
			\designdepth{Circuit depth} &
			\designancilla{Extra ancillas} \\
			\midrule
			\multicolumn{3}{@{}l}{\textit{$\delta$-dimensional architecture}} \\
			\designname{1D brickwork random circuits~\cite{Brandao2016LocalRandomCircuits,Chen2025Incompressibility}}
			& \designdepth{$\tilO\!\left(nk+\log\epsilon^{-1}\right)$}
			& \designancilla{$0$} \\
			\designname{Blocked 1D circuits~\cite{Schuster2024RandomUnitaries}}
			& \designdepth{$\tilO(k)\,\log(n\epsilon^{-1})$}
			& \designancilla{$0$} \\
			\designname{Almost quadruply optimal 1D circuits~\cite{Liu2026QuadruplyOptimalUnitary}}
			& \designdepth{$\cO\left(k\log k + \log(n\epsilon^{-1})\right)$}
			& \designancilla{$0$} \\
			\designname{Finite-dimensional LRFC circuits~\cite{Du2026ComplexityDrivenTransition}}
			& \designdepth{$\tilO(
				k^{1+1/\delta})[\log(n\epsilon^{-1})]^{1/\delta}$}
			& \designancilla{$\cO(kn)$} \\
			\designname{\textbf{This work}}
			& \designdepth{$\tilO(
				k)[\log(n\epsilon^{-1})]^{1/\delta}$}
			& \designancilla{$\mathbf{0}$} \\
			\midrule
			\multicolumn{3}{@{}l}{\textit{All-to-all architecture}} \\
			\designname{Blocked Clifford circuits ($k\leq 3$)~\cite{Schuster2024RandomUnitaries}}
			& \designdepth{$\cO\!\left(\log\log(n\epsilon^{-1})\right)$}
			& \designancilla{$\cO\!\left(n\log(n\epsilon^{-1})\right)$} \\
			\designname{Blocked LRFC circuits (low depth)~\cite{Cui2025UnitaryDesignsOptimal}}
			& \designdepth{$\tilO(k)\log\log(n\epsilon^{-1})$}
			& \designancilla{$kn\cdot\tilO\left(\log\log(n\epsilon^{-1})\right)$} \\
			\designname{Blocked LRFC circuits (low space)~\cite{Cui2025UnitaryDesignsOptimal}}
			& \designdepth{$\tilO(k^2) \log\log(n\epsilon^{-1})$}
			& \designancilla{$n\cdot \tilO(\log\log(nk\epsilon^{-1}))$} \\
			\designname{Magic-augmented Clifford circuits~\cite{Zhang2026MagicAugmented}}
			& \designdepth{$2^{\tilO(k)} + \cO(\log\log(n\epsilon^{-1}))$}
			& \designancilla{$\cO\left(n k^2 + n \log(n\epsilon^{-1})\right)$} \\
			\designname{\textbf{This work}}
			& \designdepth{$\tilO(k) \log\log(n\epsilon^{-1})$}
			& \designancilla{$\mathbf{0}$} \\
			\bottomrule
		\end{tabular}
		\caption{Comparison of circuit depth and ancilla count for $n$-qubit multiplicative $\epsilon$-approximate unitary $k$-designs. The 1D construction in Ref.~\cite{Liu2026QuadruplyOptimalUnitary} assumes $k=\cO(n)$, whereas the magic-augmented all-to-all construction assumes $k=o(\sqrt{n})$~\cite{Zhang2026MagicAugmented}. Other prior constructions assume $k = 2^{\cO(n)}$.}
		\label{tab:multiplicative-design-comparison}
		\endgroup
	\end{table}

	\begin{table}[htbp!]
		\centering
		\small
		\setlength{\tabcolsep}{3pt}
		\renewcommand{\arraystretch}{1.16}
		\begingroup
		\newcommand{\exactname}[1]{\parbox[t]{5.65cm}{\raggedright #1\strut}}
		\newcommand{\exactorder}[1]{\parbox[t]{2.50cm}{\raggedright #1\strut}}
		\newcommand{\exactdepth}[1]{\parbox[t]{5.20cm}{\raggedright #1\strut}}
		\newcommand{\exactancilla}[1]{\parbox[t]{3.20cm}{\raggedright #1\strut}}
		\begin{tabular}{@{}llll@{}}
			\toprule
			\exactname{Construction} &
			\exactorder{Design order} &
			\exactdepth{Circuit depth} &
			\exactancilla{Extra ancillas} \\
			\midrule
			\multicolumn{4}{@{}l}{\textit{$\delta$-dimensional architecture}} \\
			\exactname{Restricted Clifford circuits~\cite{Du2026ComplexityDrivenTransition}}
			& \exactorder{$k=2$}
			& \exactdepth{$\cO(n^{1/\delta})$}
			& \exactancilla{$\cO(n)$} \\
			\exactname{Clifford circuits~\cite{Maslov2007LinearDepthStabilizer}}
			& \exactorder{$k = 3$}
			& \exactdepth{$\cO(n)$}
			& \exactancilla{$0$} \\
			\exactname{\textbf{This work}}
			& \exactorder{arbitrary $k$}
			& \exactdepth{$\cO(n^{1/\delta}k\log^2 k)$}
			& \exactancilla{$\mathbf{0}$} \\
			\midrule
			\multicolumn{4}{@{}l}{\textit{All-to-all architecture}} \\
			\exactname{Near-linear exact-design circuits~\cite{Cleve2016ExactDesign}}
			& \exactorder{$k=2$}
			& \exactdepth{$\cO(\log n)$}
			& \exactancilla{$\widetilde{\cO}(n)$} \\
			\exactname{Clifford circuits (low depth)~\cite{Jiang2020SpaceDepthTradeoffCNOT}}
			& \exactorder{$k = 3$}
			& \exactdepth{$\cO(\log n)$}
			& \exactancilla{$\cO(n^2/\log^2 n)$} \\
			\exactname{Clifford circuits (no ancillas)~\cite{Jiang2020SpaceDepthTradeoffCNOT}}
			& \exactorder{$k = 3$}
			& \exactdepth{$\cO(n/\log n)$}
			& \exactancilla{$0$} \\
			\exactname{Inductive exact-design circuits~\cite{Nakata2021QuantumCircuitsExactDesigns}}
			& \exactorder{arbitrary $k$}
			& \exactdepth{Not reported; two-qubit gate count
				$\exp[\cO(\sqrt{k}n)]$}
			& \exactancilla{Required} \\
			\exactname{\textbf{This work}}
			& \exactorder{arbitrary $k$}
			& \exactdepth{$\cO(k\log k(\log n+\log k))$}
			& \exactancilla{$\mathbf{0}$} \\
			\bottomrule
		\end{tabular}
		\caption{Comparison of circuit depth and ancilla count for $n$-qubit exact designs.}
		\label{tab:exact-design-comparison}
		\endgroup
	\end{table}

	\section{Preliminaries on random unitaries}
	\label{app:moments_random_unitaries}
	
	In this section, we review essential properties of random unitaries. 
	For an $n$-qubit system, let $d=2^n$ denote the dimension of the Hilbert space. 
	For an integer $r\geq1$, we define the representation $R_r(U)\coloneqq  U^{\otimes r}\otimes\overline U^{\otimes r}$, where $\overline U$ denotes entrywise complex conjugation in the computational basis. 
	By definition, $R_r(UV)=R_r(U)R_r(V)$.

	For a probability measure $\nu$ on $\mathrm{U}(d)$, its $r$-th moment operator is given by $M_r(\nu):= \mathbb E_{U\sim\nu}[R_r(U)]$. 
	By the triangle inequality, 
	\begin{equation}
		\left\|M_r(\nu)\right\|_{\mathrm{op}}
		\leq
		\mathbb E_{U\sim\nu}
		\left\|R_r(U)\right\|_{\mathrm{op}}
		=1.
		\label{eq:app-moment-contraction}
	\end{equation}
	The corresponding Haar moment is $P_{\mathrm H}^{(r)} := \bE_{U\sim \rm H}[R_r(U)]$, where $\rm H$ denotes the Haar measure.
	We can now define tensor-product expanders (TPEs) based on their distance to the Haar moment~\cite{Hastings2008QuantumTensorProduct, Brandao2016LocalRandomCircuits}:
	\begin{definition}[TPE $\lambda$-approximate unitary $r$-design]
		A unitary ensemble $\nu$ forms a TPE $\lambda$-approximate unitary $r$-design if 
		\begin{equation}
			\lambda_r(\nu) \coloneqq \left\|
			M_r(\nu)-P_{\mathrm H}^{(r)} \right\|_{\mathrm{op}}
			\le \lambda.
			\label{eq:app-nonhaar-contraction}
		\end{equation}
	\end{definition}
	
	A TPE can be directly converted into a multiplicative-error approximate design via the following bound:
	
	\begin{fact}[Comparison between TPE and multiplicative errors,
		{\cite[Lemma~4]{Brandao2016LocalRandomCircuits}}]
		\label{fact:TPE-multiplicative-comparison}
		Let $\nu$ be a unitary ensemble on a $d$-dimensional Hilbert space.
		If $\nu$ is a TPE $\lambda$-approximate unitary $r$-design, then it is a multiplicative $\epsilon$-approximate unitary $r$-design with
		\begin{equation}
			\epsilon\leq d^{2r}\lambda.
		\end{equation}
		Conversely, if $\nu$ is a multiplicative $\epsilon$-approximate unitary $r$-design, then it is a TPE $\lambda$-approximate unitary $r$-design with
		\begin{equation}
			\lambda\leq 2d^{r/2}\epsilon.
		\end{equation}
	\end{fact}
	
	\noindent Consequently, if $\lambda_r(\nu) = 0$, the ensemble $\nu$ forms an exact unitary $r$-design.
	
	For two probability measures $\mu$ and $\nu$ on $\mathrm{U}(d)$, let $\mu*\nu$ denote their convolution, the distribution of the product $UV$ for independent samples $U\sim\mu$ and $V\sim\nu$. 
	Let $\nu^{*T}$ denote the $T$-fold convolution, which represents the distribution of a product of $T$ independent samples drawn from $\nu$. 
	This repetition exponentially suppresses the TPE error~\cite[Eq.~(3)]{Brandao2016LocalRandomCircuits}:
	\begin{equation}
		\lambda_r(\nu^{*T})
		\leq
		\lambda_r(\nu)^T.
		\label{eq:app-gap-amplification}
	\end{equation}
	
	\section{Hadamard-diagonal ensembles}\label{app:spectral-Hadamard-diagonal}
	
	In this section, we prove the moment bound for the Hadamard--diagonal ensemble presented in Eqs.~\eqref{eq:ideal-fourier-gap} and \eqref{eq:HD-master-gap}. 
	Throughout this section, we index the computational basis by $G:=\mathbb{F}_2^m$ and denote the Hilbert space dimension by $d:=|G|=2^m$.
	
	\subsection{Ideal diagonal moments}
	
	Recall that a random unitary drawn from the ideal diagonal ensemble $\mathcal{D}_{\mathrm Z}$ takes the form
	\begin{equation}
		D_{\mathrm Z}
		=
		\sum_{x\in G}e^{i\theta_x}|x\rangle\langle x|,
		\qquad
		\theta_x\overset{\mathrm{iid}}{\sim}
		\mathrm{Unif}([0,2\pi)).
	\end{equation}
	Its $r$-th moment operator, $P_{\mathrm Z}^{(r)}:=M_r(\mathcal{D}_{\mathrm Z})$, acts as an orthogonal projector. Concretely, consider basis states labeled by $\bm x=(x_1,\ldots,x_r)$ and $\bm y=(y_1,\ldots,y_r)$ in $G^r$. The action of the moment operator is given by
	\begin{equation}
		P_{\mathrm Z}^{(r)}|\bm x,\bm y\rangle
		=
		\begin{cases}
			|\bm x,\bm y\rangle,
			& \bm y=\bm x_\sigma\text{ for some }\sigma\in S_r,\\
			0,
			& \text{otherwise},
		\end{cases}
		\label{eq:app-diagonal-projector-action}
	\end{equation}
	where $\bm x_\sigma:=(x_{\sigma(1)},\ldots,x_{\sigma(r)})$. Thus, $P_{\mathrm Z}^{(r)}$ acts as a projector that retains precisely the subspace spanned by basis vectors where $\bm x$ and $\bm y$ contain the same multiset of labels.
	
	Then, let $H_m:=H^{\otimes m}$ denote the $m$-qubit Hadamard transform, and define
	\begin{equation}
		W_r
		:=
		H_m^{\otimes r}\otimes\overline{H_m}^{\otimes r},
		\qquad
		P_{\mathrm X}^{(r)}
		:=
		W_rP_{\mathrm Z}^{(r)}W_r^\dagger.
		\label{eq:app-X-projector}
	\end{equation}
	Because $H_m$ is real and Hermitian, $W_r=W_r^\dagger$.
	Independence of the three diagonal unitaries in $U_{\mathrm{HD}}=D_3H_mD_2H_mD_1$ then gives
	\begin{align}
		M_r(\nu_{\mathrm{HD}})
		&=
		P_{\mathrm Z}^{(r)}W_rP_{\mathrm Z}^{(r)}
		W_rP_{\mathrm Z}^{(r)}
		\nonumber\\
		&=
		P_{\mathrm Z}^{(r)}P_{\mathrm X}^{(r)}P_{\mathrm Z}^{(r)}.
		\label{eq:app-ideal-HD-factorization}
	\end{align}
	
	While alternating random diagonal unitaries between two Hadamard-related bases is a well-established strategy for generating unitary designs~\cite{Nakata2017TimeIndependentHamiltonian}, the following theorem establishes a substantially tighter bound on the resulting TPE error.
	
	\begin{theorem}[TPE error of Hadamard--diagonal ensembles]
		\label{thm:app-fourier-gap}
		If $d\geq4r^4$, then
		\begin{equation}
			\left\|
			P_{\mathrm Z}^{(r)}P_{\mathrm X}^{(r)}P_{\mathrm Z}^{(r)}
			-P_{\mathrm H}^{(r)}
			\right\|_{\mathrm{op}}
			\leq
			\frac{9r^4}{d}.
			\label{eq:app-fourier-gap}
		\end{equation}
	\end{theorem}
	
	\begin{proof}
		\smallskip
		For $\bm x\in G^r$, define the size of its stabilizer under coordinate
		permutations by
		\begin{equation}
			\varsigma(\bm x)
			:=
			\left|
			\{\pi\in S_r:\bm x_\pi=\bm x\}
			\right|.
		\end{equation}
		
		Let $S_r$ denote the symmetric group on $[r]$.
		Let $\mathcal K$ be the auxiliary Hilbert space with orthonormal basis
		$\{|\bm x,\sigma\rangle:\bm x\in G^r,\ \sigma\in S_r\}$ and define
		\begin{equation}
			Q|\bm x,\sigma\rangle
			:=
			\varsigma(\bm x)^{-1/2}
			|\bm x,\bm x_\sigma\rangle.
			\label{eq:app-lifting-map}
		\end{equation}
		Every basis vector in $\operatorname{ran}P_{\mathrm Z}^{(r)}$ has
		exactly $\varsigma(\bm x)$ preimages under $Q$, and hence
		\begin{equation}
			QQ^\dagger=P_{\mathrm Z}^{(r)}.
			\label{eq:app-lifting-coisometry}
		\end{equation}
		
		Next, we define the operators
		\begin{equation}
			\mathsf A
			:=
			P_{\mathrm Z}^{(r)}W_rP_{\mathrm Z}^{(r)},
			\qquad
			\widehat{\mathsf A}
			:=
			Q^\dagger W_rQ.
			\label{eq:app-lifted-operators}
		\end{equation}
		Eq.~\eqref{eq:app-lifting-coisometry} implies that $Q^\dagger$ restricts to a unitary transformation from $\operatorname{ran}(P_{\mathrm Z}^{(r)})$ onto $\operatorname{ran}(Q^\dagger)$.
		Consequently, the operators are unitarily equivalent on their respective supports:
		\begin{equation}
			Q^\dagger
			\left(
			\mathsf A|_{\operatorname{ran}(P_{\mathrm Z}^{(r)})}
			\right)
			Q
			=
			\widehat{\mathsf A}|_{\operatorname{ran}(Q^\dagger)}.
			\label{eq:app-lifted-unitary-equivalence}
		\end{equation}
		Thus, $\mathsf A$ and $\widehat{\mathsf A}$ share exactly the same nonzero singular values.
		
		\smallskip
		\noindent\emph{Bounding off-diagonal blocks.}---
		View $\widehat{\mathsf A}$ as an $S_r\times S_r$ block matrix, with
		the rows and columns within each block indexed by
		$\bm x,\bm u\in G^r$. We have
		\begin{align}
			\langle\bm x,\sigma|
			\widehat{\mathsf A}
			|\bm u,\tau\rangle
			=
			\frac{d^{-r}}
			{\sqrt{\varsigma(\bm x)\varsigma(\bm u)}}
			(-1)^{
				\sum_{j=1}^r
				x_j\cdot
				[u_j+u_{(\tau\sigma^{-1})(j)}]
			}.
			\label{eq:app-lifted-Walsh-kernel}
		\end{align}
		Put $p:=\tau\sigma^{-1}$, let $c(p)$ be the number of cycles of $p$,
		and define
		\begin{equation}
			\ell(p):=r-c(p).
		\end{equation}
		Let $P_p$ permute the $r$ coordinates according to
		$(P_p\bm u)_j=u_{p(j)}$, and set
		$L_p:=I+P_p$ over $\bF_2$. The kernel of $L_p$ consists of the vectors
		that are constant on every cycle of $p$. Hence
		\begin{equation}
			\operatorname{rank}_{\bF_2}L_p=m \ell(p),
			\qquad
			|\ker L_p^{\mathsf T}|=d^{r-\ell(p)}
		\end{equation}
		when $L_p$ acts on $G^r$.
		
		Define a diagonal contraction $D$ and an unweighted Fourier block
		$K_p$ by
		\begin{equation}
			D|\bm x\rangle
			:=
			\varsigma(\bm x)^{-1/2}|\bm x\rangle,
			\qquad
			(K_p)_{\bm x,\bm u}
			:=
			d^{-r}(-1)^{\langle\bm x,L_p\bm u\rangle},
		\end{equation}
		where
		$\langle\bm x,\bm v\rangle:=\sum_jx_j\cdot v_j$ is the standard inner product over 
		$\bF_2$. Eq.~\eqref{eq:app-lifted-Walsh-kernel} gives
		$\widehat{\mathsf A}_{\sigma,\tau}=DK_pD$. Furthermore,
		\begin{align}
			(K_pK_p^\dagger)_{\bm x,\bm y}
			&=
			d^{-2r}
			\sum_{\bm u\in G^r}
			(-1)^{\langle\bm x-\bm y,L_p\bm u\rangle}
			\nonumber\\
			&=
			d^{-r}\,
			\mathbbm{1}
			\left\{
			L_p^{\mathsf T}(\bm x-\bm y)=0
			\right\}.
			\label{eq:app-Fourier-block-product}
		\end{align}
		This matrix is a direct sum of copies of
		$d^{-r}J_{d^{r-\ell(p)}}$, where $J_s$ denotes the $s\times s$
		all-ones matrix. Since $\|J_s\|_{\mathrm{op}}=s$, we obtain
		\begin{equation}
			\|K_p\|_{\mathrm{op}}^2=d^{-\ell(p)}.
		\end{equation}
		Using $\|D\|_{\mathrm{op}}\leq1$, we conclude that
		\begin{equation}
			\left\|
			\widehat{\mathsf A}_{\sigma,\tau}
			\right\|_{\mathrm{op}}
			\leq
			d^{-\ell(\tau\sigma^{-1})/2}.
			\label{eq:app-permutation-block-bound}
		\end{equation}
		
		\smallskip
		\noindent\emph{Approximation by a rank-$r!$ projector.}---
		Let $\mathsf B$ denote the block-diagonal part of
		$\widehat{\mathsf A}$. For a diagonal block, $\tau=\sigma$, so
		$p$ is the identity and $L_p=0$. Therefore
		\begin{equation}
			\langle\bm x,\sigma|
			\widehat{\mathsf A}
			|\bm u,\sigma\rangle
			=
			\frac{d^{-r}}
			{\sqrt{\varsigma(\bm x)\varsigma(\bm u)}}
			=
			a_{\bm x}a_{\bm u},
		\end{equation}
		where
		\begin{equation}
			a_{\bm x}
			:=
			d^{-r/2}\varsigma(\bm x)^{-1/2},
			\qquad
			\zeta
			:=
			\langle a,a\rangle
			=
			d^{-r}\sum_{\bm x\in G^r}\varsigma(\bm x)^{-1}.
		\end{equation}
		Thus every diagonal block equals $|a\rangle\langle a|$. Clearly
		$\zeta\leq1$. On the other hand, whenever the entries of $\bm x$ are
		pairwise distinct, $\varsigma(\bm x)=1$. There are
		$d(d-1)\cdots(d-r+1)$ such tuples, and therefore
		\begin{align}
			\zeta
			&\geq
			d^{-r}d(d-1)\cdots(d-r+1)
			\nonumber\\
			&=
			\prod_{j=0}^{r-1}
			\left(1-\frac jd\right)
			\geq
			1-\frac{\binom r2}{d}.
		\end{align}
		Let $|\widehat a\rangle:=\zeta^{-1/2}|a\rangle$ and define
		\begin{equation}
			P_0
			:=
			\bigoplus_{\sigma\in S_r}
			|\widehat a\rangle\langle\widehat a|.
		\end{equation}
		Then $P_0$ is an orthogonal projector of rank $r!$, and
		\begin{equation}
			\|\mathsf B-P_0\|_{\mathrm{op}}
			=
			1-\zeta
			\leq
			\frac{\binom r2}{d}.
			\label{eq:app-diagonal-block-approximation}
		\end{equation}
		
		For the block matrix $T=(T_{\sigma,\tau})_{\sigma,\tau\in S_r}$, define
		\begin{equation}
			R:=\max_{\sigma}\sum_{\tau}
			\|T_{\sigma,\tau}\|_{\mathrm{op}},
			\qquad
			C:=\max_{\tau}\sum_{\sigma}
			\|T_{\sigma,\tau}\|_{\mathrm{op}}.
		\end{equation}
		For any block vector \(v=(v_\tau)_{\tau\in S_r}\), the triangle
		inequality and Cauchy--Schwarz inequality give
		\begin{align}
			\|Tv\|^2
			&=
			\sum_{\sigma}
			\left\|
			\sum_{\tau}T_{\sigma,\tau}v_\tau
			\right\|^2
			\nonumber\\
			&\leq
			\sum_{\sigma}
			\left(
			\sum_{\tau}
			\|T_{\sigma,\tau}\|_{\mathrm{op}}\,
			\|v_\tau\|
			\right)^2
			\nonumber\\
			&\leq
			R\sum_{\sigma,\tau}
			\|T_{\sigma,\tau}\|_{\mathrm{op}}\,
			\|v_\tau\|^2
			\leq
			RC\|v\|^2.
		\end{align}
		Taking the supremum over \(v\neq0\) therefore yields
		\begin{equation}
			\|T\|_{\mathrm{op}}
			\leq
			\sqrt{RC}
			=
			\left(
			\max_{\sigma}\sum_{\tau}
			\|T_{\sigma,\tau}\|_{\mathrm{op}}
			\right)^{1/2}
			\left(
			\max_{\tau}\sum_{\sigma}
			\|T_{\sigma,\tau}\|_{\mathrm{op}}
			\right)^{1/2}.
		\end{equation}

		Applying this to $\widehat{\mathsf A}-\mathsf B$ and using
		Eq.~\eqref{eq:app-permutation-block-bound}, both block sums are bounded by
		\begin{equation}
			\sum_{p\neq e}d^{-\ell(p)/2}.
		\end{equation}
		A permutation satisfying $\ell(p)=\ell$ can be expressed as a product
		of $\ell$ transpositions. Hence the number of such permutations is at
		most $\binom r2^\ell\leq r^{2\ell}$. Because
		$d\geq4r^4$, we have $r^2/\sqrt d\leq1/2$, and therefore
		\begin{align}
			\|\widehat{\mathsf A}-\mathsf B\|_{\mathrm{op}}
			&\leq
			\sum_{p\neq e}d^{-\ell(p)/2}
			\nonumber\\
			&\leq
			\sum_{\ell\geq1}
			\left(\frac{r^2}{\sqrt d}\right)^\ell
			\leq
			\frac{2r^2}{\sqrt d}.
			\label{eq:app-offdiagonal-block-bound}
		\end{align}
		Combining Eqs.~\eqref{eq:app-diagonal-block-approximation} and
		\eqref{eq:app-offdiagonal-block-bound}, we find
		\begin{equation}
			\|\widehat{\mathsf A}-P_0\|_{\mathrm{op}}
			\leq
			\eta
			:=
			\frac{2r^2}{\sqrt d}
			+
			\frac{\binom r2}{d}.
			\label{eq:app-lifted-rank-approximation}
		\end{equation}
		
		\smallskip
		\noindent\emph{Identifying the Haar sector.}---
		For \(d\ge r\), the Haar moment projector has rank
		\begin{equation}
			\operatorname{rank}P_{\mathrm H}^{(r)}=r!
		\end{equation}
		by Schur--Weyl duality~\cite[Theorem~9 and Propositions~11 and~25]{Mele2024HaarTools}.
		If $v\in\operatorname{ran}P_{\mathrm H}^{(r)}$, then
		$R_r(U)v=v$ for every $U$. Averaging over diagonal unitaries gives
		$P_{\mathrm Z}^{(r)}v=v$, while taking $U=H_m$ gives $W_rv=v$.
		Consequently,
		\begin{equation}
			\mathsf Av
			=
			P_{\mathrm Z}^{(r)}W_rP_{\mathrm Z}^{(r)}v
			=v.
		\end{equation}
		Moreover, $W_r=W_r^\dagger$ implies
		$\mathsf A=\mathsf A^\dagger$. Thus the orthogonal complement of the
		Haar-invariant subspace is invariant under $\mathsf A$, and
		\begin{equation}
			\mathsf A
			=
			I_{\operatorname{ran}P_{\mathrm H}^{(r)}}
			\oplus
			\mathsf A_\perp.
			\label{eq:app-A-Haar-decomposition}
		\end{equation}
		
		Let $s_j(T)$ denote the $j$-th largest singular value of $T$. The
		variational characterization
		\begin{equation}
			s_{L+1}(T)
			=
			\inf_{\operatorname{rank}R\leq L}
			\|T-R\|_{\mathrm{op}},
		\end{equation}
		together with Eqs.~\eqref{eq:app-lifted-unitary-equivalence} and
		\eqref{eq:app-lifted-rank-approximation}, gives
		\begin{equation}
			\|\mathsf A_\perp\|_{\mathrm{op}}
			=
			s_{r!+1}(\mathsf A)
			=
			s_{r!+1}(\widehat{\mathsf A})
			\leq
			\eta.
			\label{eq:app-A-perp-bound}
		\end{equation}
		Finally, using $W_r^\dagger=W_r$,
		\begin{equation}
			P_{\mathrm Z}^{(r)}P_{\mathrm X}^{(r)}P_{\mathrm Z}^{(r)}
			=
			P_{\mathrm Z}^{(r)}W_rP_{\mathrm Z}^{(r)}W_r
			P_{\mathrm Z}^{(r)}
			=
			\mathsf A^2.
		\end{equation}
		Eq.~\eqref{eq:app-A-Haar-decomposition} therefore yields
		\begin{align}
			\left\|
			P_{\mathrm Z}^{(r)}P_{\mathrm X}^{(r)}P_{\mathrm Z}^{(r)}
			-P_{\mathrm H}^{(r)}
			\right\|_{\mathrm{op}}
			&=
			\|\mathsf A^2-P_{\mathrm H}^{(r)}\|_{\mathrm{op}}
			\nonumber\\
			&=
			\|\mathsf A_\perp^2\|_{\mathrm{op}}
			=
			\|\mathsf A_\perp\|_{\mathrm{op}}^2
			\leq
			\eta^2.
		\end{align}
		Because $\binom r2\leq r^2$ and $d^{-1}\leq d^{-1/2}$, $\eta \leq \frac{3r^2}{\sqrt d}$, which proves Eq.~\eqref{eq:app-fourier-gap}.
	\end{proof}

	\subsection{Stability under diagonal approximation}
	We now examine the stability of the Hadamard-diagonal construction when the ideal diagonal ensemble is replaced by an arbitrary ensemble $\mathcal{D}_m$. Define
	\begin{equation}
		\widetilde P_{\mathrm Z}^{(r)}
		:=
		M_r(\cD_m),
		\qquad
		\epsilon_{\mathrm{diag}}(r)
		:=
		\left\|
		\widetilde P_{\mathrm Z}^{(r)}-P_{\mathrm Z}^{(r)}
		\right\|_{\mathrm{op}},
		\label{eq:app-diagonal-approximation-error}
	\end{equation}
	and
	\begin{equation}
		\widetilde P_{\mathrm X}^{(r)} :=
		W_r\widetilde P_{\mathrm Z}^{(r)}W_r^\dagger.
	\end{equation}
	Let $\nu_{\mathrm{HD}}^{(m)}$ denote the ensemble in which the three
	diagonal layers of the Hadamard--diagonal block are sampled
	independently from $\cD_m$.
	Then, we have
	\begin{equation}
		M_r\left(\nu_{\mathrm{HD}}^{(m)}\right) =
		\widetilde P_{\mathrm Z}^{(r)}
		\widetilde P_{\mathrm X}^{(r)}
		\widetilde P_{\mathrm Z}^{(r)},
		\label{eq:app-implemented-HD-factorization}
	\end{equation}
	The following lemma establishes that the resulting TPE error degrades at most linearly with the diagonal approximation error.
	
	\begin{lemma}[Stability under diagonal approximation]
		\label{lem:app-HD-perturbation}
		For every $r\geq1$ satisfying $d=2^m\geq4r^4$,
		\begin{equation}
			\lambda_r\left(\nu_{\mathrm{HD}}^{(m)}\right)
			\leq
			\frac{9r^4}{2^m}
			+
			3\epsilon_{\mathrm{diag}}(r).
			\label{eq:app-HD-master-gap}
		\end{equation}
	\end{lemma}
	
	\begin{proof}
		By the unitary invariance of the operator norm,
		\begin{equation}
			\left\|
			\widetilde P_{\mathrm X}^{(r)}-P_{\mathrm X}^{(r)}
			\right\|_{\mathrm{op}}
			=
			\epsilon_{\mathrm{diag}}(r).
		\end{equation}
		Moreover, 
		\begin{align}
			\widetilde P_{\mathrm Z}^{(r)}
			\widetilde P_{\mathrm X}^{(r)}
			\widetilde P_{\mathrm Z}^{(r)}
			-
			P_{\mathrm Z}^{(r)}P_{\mathrm X}^{(r)}P_{\mathrm Z}^{(r)} =
			\left(\widetilde P_{\mathrm Z}^{(r)}-P_{\mathrm Z}^{(r)}\right)
			\widetilde P_{\mathrm X}^{(r)}
			\widetilde P_{\mathrm Z}^{(r)}
			+
			P_{\mathrm Z}^{(r)}
			\left(\widetilde P_{\mathrm X}^{(r)}-P_{\mathrm X}^{(r)}\right)
			\widetilde P_{\mathrm Z}^{(r)} +
			P_{\mathrm Z}^{(r)}P_{\mathrm X}^{(r)}
			\left(\widetilde P_{\mathrm Z}^{(r)}-P_{\mathrm Z}^{(r)}\right).
		\end{align}
		Notice that all ideal and approximate moment operators have operator norms bounded by $1$. 
		Therefore,
		\begin{equation}
			\left\|
			M_r\left(\nu_{\mathrm{HD}}^{(m)}\right)
			-
			P_{\mathrm Z}^{(r)}P_{\mathrm X}^{(r)}P_{\mathrm Z}^{(r)}
			\right\|_{\mathrm{op}}
			\leq
			3\epsilon_{\mathrm{diag}}(r).
		\end{equation}
		Combining this estimate with
		Theorem~\ref{thm:app-fourier-gap} proves
		Eq.~\eqref{eq:app-HD-master-gap}.
	\end{proof}
	
	\section{Finite-dimensional implementation}
	\label{app:finite-dimensional}
	This section presents the construction of low-depth designs on finite-dimensional architectures.
	Throughout this section, we assume the design order $k \ge 2$.
	
	We begin by constructing the diagonal ensemble. Recall that $P_{\mathrm Z}^{(r)}=M_r(\mathcal{D}_{\mathrm Z})$ denotes the $r$-th moment operator of the ideal diagonal ensemble, in which every computational-basis state receives an independent random phase. 
	Our goal is to approximate these $2^m$ independent phases using only $\Theta(m)$ random phases generated via a low-depth circuit.
	To achieve this, the construction first mixes the input states using a randomized hash function, subsequently applies independent single-qubit phases, and finally reverses the mixing computation.
	
	Specifically, our protocol operates on an $m$-qubit system where the qubits are organized into prefix windows, defined as follows:

	\begin{definition}[Prefix windows]
		\label{def:prefix-window}
		Fix a dimension $\delta\geq 1$ and a width $L$. 
		Define the width-$L$ tube by
		\begin{equation}
			\mathcal Z_{\delta,L}:=\mathbb N\times[L]^{\delta-1},
		\end{equation}
		Number the sites in lexicographical order of coordinates.
		The prefix window $H_\delta(L,N)$ is the subgraph induced by the first $N$ sites in this order.
		
		The prefix window is called \emph{balanced} if
		\begin{equation}
			L^\delta\leq N\leq \max(4,2^\delta)\cdot L^\delta.
			\label{eq:app-fd-balanced-prefix}
		\end{equation}
	\end{definition}
	
	The performance guarantees of our construction are given in the following theorem.
	
	\begin{theorem}[Low-depth diagonal blocks on finite-dimensional architectures]
		\label{thm:finite-dimensional-ensemble}
		There exist constants $A_\delta, c_\delta, C_\delta > 0$ such that for every $m \geq A_\delta \log k$ and every $m$-qubit balanced prefix window, there exists an ensemble $\nu_{\delta,m}$ of diagonal unitaries acting on the $m$ qubits that satisfies
		\begin{equation}
			\label{eq:finite-dimensional-distance}
			\|M_r(\nu_{\delta,m}) - P_{\mathrm Z}^{(r)}\|_{\mathrm{op}} \leq C_\delta \exp\left(-\frac{c_\delta m}{\log k}\right),
		\end{equation}
		simultaneously for all $1 \leq r \leq 2k$. Furthermore, any unitary sampled from $\nu_{\delta,m}$ can be implemented with a circuit depth of
		\begin{equation}
			\label{eq:finite-dimensional-depth}
			\cO_\delta\left(m^{1/\delta}\log k\right).
		\end{equation}
		These circuits require no ancillary qubit.
	\end{theorem}

	Substituting this ensemble into Eq.~\eqref{eq:app-implemented-HD-factorization} gives an approximate design, whose error can then be suppressed via the repetition bound in Eq.~\eqref{eq:app-gap-amplification}.
	Gluing multiple such designs together produces a global design with a reduced asymptotic depth, as detailed in Sec.~\ref{app:finite-dimensional-repetition}.
	
	The remainder of this section is organized as follows. We first prove Theorem~\ref{thm:finite-dimensional-ensemble} and then promote its diagonal block to a full unitary design. Specifically, Sec.~\ref{app:finite-dimensional-construction} constructs the diagonal ensemble, Sec.~\ref{app:finite-dimensional-distance} establishes the moment estimate bounded in Eq.~\eqref{eq:finite-dimensional-distance}, and Sec.~\ref{app:finite-dimensional-depth} details its ancilla-free $\delta$-dimensional implementation achieving the depth stated in Eq.~\eqref{eq:finite-dimensional-depth}. Finally, Sec.~\ref{app:finite-dimensional-repetition} integrates this diagonal ensemble into the Hadamard--diagonal block, suppresses the resulting TPE error via independent repetitions, and glues the local designs to obtain a low-depth design on a larger system.

	\subsection{Ensemble construction}
	\label{app:finite-dimensional-construction}
	We first provide the construction of the ensemble $\nu_{\delta,m}$ in Theorem~\ref{thm:finite-dimensional-ensemble}.
	Our construction employs several circuit primitives on specific qubit subsets that conventionally require ancillary space. To maintain an entirely ancilla-free architecture, our key insight is to catalytically borrow inactive qubits while operations are performed on the active registers.
	
	Concretely, we first partition the $m$ system qubits into a sufficiently large constant number $r_0 = r_0(\delta)$ of registers, $Z_1, Z_2, \ldots, Z_{r_0}$. We choose this partition such that $|Z_1| = |Z_2| = s$ for some $s = \Theta_\delta(m)$, while every other register has a size of at most $s$. Each register is then treated as an element of the finite field $\mathbb{F}_{2^s}$, where any register shorter than $s$ is implicitly padded with zeros. Accordingly, we express every computational-basis state in the block form
	\begin{equation}
		\xi = (\xi_1, \ldots, \xi_{r_0}), \qquad \xi_a \in \mathbb{F}_{2^s},
		\label{eq:app-fd-block-label}
	\end{equation}
	where $\xi_a$ denotes the zero-padded bit string associated with register $Z_a$. Finally, we designate the first two registers as our primary operational workspaces, defining
	\begin{equation}
		X := Z_1, \quad Y := Z_2.
	\end{equation}
	Thus, $X$ and $Y$ initially contain the field elements $\xi_1$ and $\xi_2$, respectively, and serve as the main registers for our subsequent circuit operations.
	
	A sample from the ensemble is generated by the following procedure:
	\begin{enumerate}
		\item \emph{Concentration.} Apply a random linear map from all other registers to concentrate their weighted sum into $Y$. Specifically, sample the coefficients $\lambda_a$ independently and uniformly from $\mathbb{F}_{2^s}$ for all $a \neq 2$, and perform the update
		\begin{equation}
			Y \gets Y \oplus \bigoplus_{a \neq 2} \lambda_a Z_a.
		\end{equation}
		
		\item \emph{Hashing.} Mix the concentrated register $Y$ into $X$ using a random polynomial. Sample a function
		\begin{equation}
			M(z) = \mu_0 z + \mu_1 z^2 + \mu_2 z^4
		\end{equation}
		with coefficients $\mu_0, \mu_1, \mu_2 \in \mathbb{F}_{2^s}$, together with a shift $\beta\in \mathbb{F}_{2^s}$, chosen independently and uniformly at random, and apply the update
		\begin{equation}
			X \gets X \oplus M(Y) \oplus \beta.
		\end{equation}
		
		\item \emph{Quadratic mixing.} Add a random sparse quadratic function of $X$ into $Y$. Fix $w = \lceil C_w \log k \rceil$ for a sufficiently large constant $C_w$. Sample permutations $\pi_1, \sigma_1, \ldots, \pi_w, \sigma_w \in S_s$ independently and uniformly, alongside independent, uniformly random bits $\gamma_{j,l} \in \mathbb{F}_2$ for $j \in [s]$ and $l \in [w]$. Apply the quadratic mixing
		\begin{equation}
			\label{eq:rand-quad}
			Y_j \gets Y_j \oplus \bigoplus_{l=1}^w \gamma_{j,l} X_{\pi_l(j)} X_{\sigma_l(j)},
		\end{equation}
		where $Y_j$ indicates the $j$-th bit of the register $Y$.

		\item \emph{Twirling.} Sample random phases $\theta_1,\cdots,\theta_s\in[0,2\pi)$ uniformly and independently.
		For each $j\in [s]$, apply the single-qubit phase gate
		\begin{equation}
			\operatorname{diag}(1,e^{\mathrm{i}\theta_j})
		\end{equation}
		to the $j$-th qubit of the register $Y$.
		
		\item \emph{Uncomputation.} Invert the first three steps.
	\end{enumerate}
	
	
	\subsection{Approximation error analysis}
	\label{app:finite-dimensional-distance}
	We now prove Eq.~\eqref{eq:finite-dimensional-distance}. 
	We first characterize the Boolean functions computed by the circuit.  
	Using the block notation $\xi=(\xi_1,\ldots,\xi_{r_0})$ from Eq.~\eqref{eq:app-fd-block-label}, the field values carried by $Y$ and $X$ after the affine mixer are, respectively,
	\begin{equation}\begin{aligned}
			\label{eq:hashxy}
			y(\xi)&:={\xi_2}+\sum_{a\neq 2}\lambda_a\xi_a,\\
			x(\xi)&:={\xi_1}+\mu_0y(\xi)+\mu_1y(\xi)^2+\mu_2y(\xi)^4+\beta.
	\end{aligned}\end{equation}
	
	Fix a binary basis of $\mathbb F_{2^s}$ and write $[z]_p$ for the $p$-th bit of $z$. 
	We use the coordinate notation
	\begin{equation}\begin{aligned}
			\label{eq:bits-feature}
			x_p(\xi)&:=[x(\xi)]_p,\qquad  y_p(\xi):=[y(\xi)]_p,\qquad   p\in[s],\\
			h_j(\xi)&:=y_j(\xi)\oplus\bigoplus_{l=1}^{w}\gamma_{j,l}x_{\pi_l(j)}(\xi)x_{\sigma_l(j)}(\xi), \qquad j\in[s].
	\end{aligned}\end{equation}
	Thus a sample from $\nu_{\delta,m}$ acts as
	\begin{equation}
		\ket{\xi} \longmapsto \exp\!\left(
		\mathrm i\sum_{j=1}^{s}\theta_jh_j(\xi)
		\right)\ket{\xi}.
		\label{eq:app-fd-diagonal-action}
	\end{equation}
	
	Write $\mathsf A
	:=\bigl(\{\lambda_a\}_{a\neq2},\mu_0,\mu_1,\mu_2,\beta\bigr)$, $\Pi:=\{\pi_l,\sigma_l\}_{l=1}^{w}$ and $\Gamma:=\{\gamma_{j,l}\}_{j\in[s],l\in[w]}$
	for the randomness in the procedure. 
	We first express the desired operator-norm distance as a failure probability over these three blocks of randomness.
	
	For $\bm u,\bm v\in(\{0,1\}^m)^r$, define their signed multiplicity profile by
	\begin{equation}
		c_\xi(\bm u,\bm v):=\bigl|\{a:u_a=\xi\}\bigr|-\bigl|\{a:v_a=\xi\}\bigr|.
		\label{eq:app-fd-signed-profile}
	\end{equation}
	Let $c(\bm u,\bm v):=(c_\xi(\bm u,\bm v))_{\xi\in \{0,1\}^m}$, and define $\mathcal C_r$ be the set of all realizable nonzero profiles:
	\begin{equation}
		\mathcal C_r:=\left\{c(\bm u,\bm v):
		\bm u,\bm v\in(\{0,1\}^m)^r,
		c(\bm u,\bm v)\not\equiv0
		\right\}.
		\label{eq:app-fd-profile-class}
	\end{equation}
	
	By Eq.~\eqref{eq:app-diagonal-projector-action}, for a computational basis $\ket{\bm u,\bm v}$, $P_{\mathrm{Z}}^{(r)}$ leaves it invariant if and only if $c(\bm u,\bm v)\equiv 0$; otherwise $P_{\mathrm{Z}}^{(r)}$ eliminates it.
	Our objective is therefore to show that $M_r(\nu_{\delta,m})$ approximately reproduces this exact behavior.
	\begin{lemma}[Approximation error as a worst-case failure probability]
		\label{lem:app-fd-distance-probability}
		For $j\in[s]$, define the integer-valued signed test
		\begin{equation}
			S_j(c):=\sum_{\xi\in\operatorname{supp}(c)}c_\xi(-1)^{h_j(\xi)}.
			\label{eq:app-fd-signed-test}
		\end{equation}
		Then for every $r\geq1$, the distance between the target $P_{\mathrm Z}^{(r)}$ and our construction $M_r(\nu_{\delta,m})$ can be written as
		\begin{equation}
			\left\|M_r(\nu_{\delta,m})-P_{\mathrm Z}^{(r)}\right\|_{\mathrm{op}}
			=
			\max_{c\in\mathcal C_r}
			\Pr_{\mathsf A,\Pi,\Gamma}
			\left[S_j(c)=0\ \text{for every }j\in[s]\right].
			\label{eq:app-fd-distance-as-probability}
		\end{equation}
	\end{lemma}
	\begin{proof}
		Fix a computational moment-basis vector $\ket{\bm u,\bm v}$ and let $c$ be its signed multiplicity profile. 
		For fixed $\mathsf A,\Pi,\Gamma$, averaging the independent angles in Eq.~\eqref{eq:app-fd-diagonal-action} gives the diagonal multiplier
		\begin{equation}
			\prod_{j=1}^{s}
			\mathbb E_{\theta_j}
			\exp\!\left(
			\mathrm i\theta_j
			\sum_\xi c_\xi h_j(\xi)
			\right).
			\label{eq:app-fd-angle-average}
		\end{equation}
		The coefficient of each $\theta_j$ is an integer, so its factor is one if $\sum_\xi c_\xi h_j(\xi)=0$ and zero otherwise. 
		Since $c$ is the signed multiplicity profile of two $r$-tuples, $\sum_\xi c_\xi=0$.
		By Eq.~\eqref{eq:app-fd-signed-test},
		\begin{equation}
			S_j(c)
			=\sum_\xi c_\xi\bigl(1-2h_j(\xi)\bigr)
			=-2\sum_\xi c_\xi h_j(\xi).
			\label{eq:app-fd-signed-test-identity}
		\end{equation}
		Hence Eq.~\eqref{eq:app-fd-angle-average} is the indicator of the event $S_j(c)=0$ for every $j$.
		
		Both moment operators are diagonal in the computational moment basis.
		When $c=0$, their eigenvalues are both one.  
		When $c\neq0$, the eigenvalue of $P_{\mathrm Z}^{(r)}$ is zero, whereas that of $M_r(\nu_{\delta,m})$ is the probability on the right-hand side of Eq.~\eqref{eq:app-fd-distance-as-probability}. 
		The operator norm is therefore the largest of these diagonal entries, which proves the identity.
	\end{proof}
	
	Now we can work for each fixed $c\in \mathcal{C}_r$.

	\begin{lemma}[Second- and fourth-moment estimate]
		\label{lem:app-fd-moment-estimate}
		There are constants $A_\delta,c_0,c_1,C_0>0$ such that the following holds when the constant $C_w$ in $w=\lceil C_w\log k\rceil$ is sufficiently large. 
		For every $1\leq r\leq2k$, every $c\in\mathcal{C}_r$, and every $m\geq A_\delta\log k$, there is an event $\mathcal{E}_{c}$ in the probability space of $(\mathsf A,\Pi)$ with
		\begin{equation}
			\Pr_{\mathsf A,\Pi}[\mathcal E_c]
			\geq 1-C_0e^{-c_0s}-e^{-c_1s/w}
			\label{eq:app-fd-structural-failure}
		\end{equation}
		and the following property holds.
		For every realization $(\mathsf A,\Pi)\in\mathcal E_c$, there is a set $G_c=G_c(\mathsf A,\Pi)\subseteq[s]$ with $|G_c|\geq s/2$ such that, for every $j\in G_c$,
		\begin{equation}\begin{aligned}
				\mathbb E_{\gamma_j}S_j(c)^2&=\sum_\xi c_\xi^2,\\
				\mathbb E_{\gamma_j}S_j(c)^4&=3\left(\sum_\xi c_\xi^2\right)^2-2\sum_\xi c_\xi^4\leq 3\left(\sum_\xi c_\xi^2\right)^2,
				\label{eq:app-fd-second-fourth-moments}
		\end{aligned}\end{equation}
		where $\gamma_j=(\gamma_{j,1},\ldots,\gamma_{j,w})$. 
	\end{lemma}
	
	Before proving Lemma~\ref{lem:app-fd-moment-estimate}, we first show that it
	immediately implies the claimed distance bound. 
	Conditional on $\mathcal E_c$, Cauchy--Schwarz gives, for every $j\in G_c$,
	\begin{equation}
		\Pr_{\gamma_j}[S_j(c)\neq0]
		\geq \frac{\bigl(\mathbb E_{\gamma_j}S_j(c)^2\bigr)^2}
		{\mathbb E_{\gamma_j}S_j(c)^4}
		\geq\frac{1}{3}.
		\label{eq:app-fd-single-target-detection}
	\end{equation}
	The independence of the private vectors $\gamma_j$ therefore yields
	\begin{equation}\begin{aligned}
			\Pr_{\mathsf A,\Pi,\Gamma}
			\left[S_j(c)=0\ \text{for every }j\in[s]\right]
			&\leq
			\Pr_{\mathsf A,\Pi}[\mathcal E_c^{\mathrm c}]
			+\max_{(\mathsf A,\Pi)\in\mathcal E_c}
			\Pr_{\Gamma}
			\left[S_j(c)=0\ \text{for every }j\in G_c(\mathsf A,\Pi)\right]
			\nonumber\\
			&\leq
			C_0e^{-c_0s}+e^{-c_1s/w}+(2/3)^{s/2}.
			\label{eq:app-fd-profile-failure-bound}
	\end{aligned}\end{equation}
	This estimate is uniform over $c\in\mathcal C_r$ and
	$1\leq r\leq2k$.  Since $s=\Theta_\delta(m)$ and
	$w=\Theta(\log k)$, Lemma~\ref{lem:app-fd-distance-probability}
	implies
	\begin{equation}
		\left\|M_r(\nu_{\delta,m})-P_{\mathrm Z}^{(r)}\right\|_{\mathrm{op}}
		\leq
		C_\delta
		\exp\!\left[-\frac{c_\delta m}{\log k}\right],
		\qquad 1\leq r\leq2k,
	\end{equation}
	after adjusting the constants.
	
	We finish the subsection by proving the claimed moment estimate.
	The first lemma shows that for any four labels, with high probability over the randomness in Step 1, the function $x$ restricted on their difference space is a uniformly random affine function.
	
	\begin{lemma}[Uniform affine hashing on four-label spans]
		\label{lem:app-fd-affine-restriction}
		Let $T\subseteq\{0,1\}^m$ be a non-empty set of at most four labels.
		Fix $\xi^{(0)}\in T$, and define the difference space
		\begin{equation}
			V:=\operatorname{span}_{\mathbb F_2}\{\xi-\xi^{(0)}:\xi\in T\},\quad
			W=\xi^{(0)}+V.
		\end{equation}
		With probability at least $1-7\cdot 2^{-s}$ over randomness of $\{\lambda_a\}_{a\neq2}$, the restricted map $y|_W$ is injective on the affine subspace $W$ for the function $y$ defined in Eq.~\eqref{eq:hashxy}.
		
		Conditional on $\{\lambda_a\}_{a\neq2}$ for which this injectivity property holds, the functions $x_1|_W,\ldots,x_s|_W$ defined in Eq.~\eqref{eq:bits-feature} restricted on the affine subspace $W$ are independent uniform affine Boolean functions, where the remaining randomness is over $\mu_0,\mu_1,\mu_2,\beta$.
	\end{lemma}
	\begin{proof}
		Since $T$ contains at most four labels, $V$ is generated by at most three vectors, and hence has dimension $d:=\dim_{\mathbb F_2}V\leq 3$.
		As $y$ is $\mathbb F_2$-linear, its restriction to $W=\xi^{(0)}+V$ is injective if and only if $y(v)\neq 0$ for every $0\neq v\in V$.
		For fixed $v\neq 0$, either $v_a=0$ for all $a\neq 2$, in which case
		$y(v)=v_2\neq 0$, or some $v_a\neq 0$ and the uniform coefficient
		$\lambda_a$ makes $y(v)$ uniform in $\mathbb F_{2^s}$. 
		A union bound over the at most $2^{\dim V}-1\leq 7$ nonzero $v$ proves the first claim.
		
		We now condition on a choice of $\{\lambda_a\}_{a\neq2}$ for which
		$y|_W$ is injective. 
		Thus, $y|_V$ is also injective.
		Choose a basis $v^{(1)},\ldots,v^{(d)}$ of $V$, and define
		\begin{equation}
			\eta_0:=y(\xi^{(0)}),
			\qquad
			\eta_i:=y(v^{(i)}),
			\quad i\in[d].
			\label{eq:app-fd-eta-basis}
		\end{equation}
		The injectivity of $y$ on $V$ implies that $\eta_1,\ldots,\eta_d$ are linearly independent over $\mathbb F_2$.
		
		By Eq.~\eqref{eq:hashxy},
		\begin{equation}
			x(\xi)
			=
			\xi_1+\beta
			+\mu_0y(\xi)
			+\mu_1y(\xi)^2
			+\mu_2y(\xi)^4.
			\label{eq:app-fd-x-expanded}
		\end{equation}
		The maps $z\mapsto z^2$ and $z\mapsto z^4$ are $\mathbb F_2$-linear. 
		It follows that $x|_W$ is an $\mathbb F_2$-affine function on $W$. In particular, each coordinate function $x_p=[x]_p$ is an affine Boolean function on this affine subspace.
		
		It remains to prove that these coordinate functions are jointly uniform.
		An $\mathbb F_{2^s}$-valued affine function on $\xi^{(0)}+V$ is uniquely determined by its value at $\xi^{(0)}$ and its increments along the basis vectors $v^{(1)},\ldots,v^{(d)}$.
		Accordingly, define
		\begin{equation}
			B_0:=x(\xi^{(0)}),
			\qquad
			B_i:=x(\xi^{(0)}+v^{(i)})-x(\xi^{(0)})
			\label{eq:app-fd-affine-data}
		\end{equation}
		for every $i\in[d]$.
		Writing $\xi^{(0)}_1$ and $v^{(i)}_1$ for the first field-valued blocks
		of $\xi^{(0)}$ and $v^{(i)}$, respectively, Eq.~\eqref{eq:app-fd-x-expanded}
		gives
		\begin{equation}
			B_0=
			\xi^{(0)}_1
			+\mu_0\eta_0
			+\mu_1\eta_0^2
			+\mu_2\eta_0^4
			+\beta
			\label{eq:app-fd-affine-base-value}
		\end{equation}
		and
		\begin{equation}
			B_i=
			v^{(i)}_1
			+\mu_0\eta_i
			+\mu_1\eta_i^2
			+\mu_2\eta_i^4,
			\qquad i\in[d].
			\label{eq:app-fd-affine-increments}
		\end{equation}
		Consider the $d\times3$ Moore matrix
		\begin{equation}
			\mathcal M
			:=
			\begin{pmatrix}
				\eta_1 & \eta_1^2 & \eta_1^4\\
				\vdots & \vdots   & \vdots\\
				\eta_d & \eta_d^2 & \eta_d^4
			\end{pmatrix}.
			\label{eq:app-fd-moore-matrix}
		\end{equation}
		Because $d\leq3$ and $\eta_1,\ldots,\eta_d$ are linearly independent
		over $\mathbb F_2$, the Moore determinant criterion implies that the
		leading $d\times d$ submatrix of $\mathcal M$ is nonsingular.
		Therefore,
		\begin{equation}
			\operatorname{rank}_{\mathbb F_{2^s}}\mathcal M=d.
		\end{equation}
		Consequently, the linear map
		\begin{equation}
			(\mu_0,\mu_1,\mu_2)
			\longmapsto
			\left(
			\mu_0\eta_i+\mu_1\eta_i^2+\mu_2\eta_i^4
			\right)_{i=1}^d
			\label{eq:app-fd-random-increment-map}
		\end{equation}
		is surjective from $\mathbb F_{2^s}^3$ onto
		$\mathbb F_{2^s}^d$. 
		Since $\mu_0,\mu_1,\mu_2$ are independent and uniform, Eq.~\eqref{eq:app-fd-affine-increments} gives that
		\begin{equation}
			(B_1,\ldots,B_d)
		\end{equation}
		is uniform in $\mathbb F_{2^s}^d$. The deterministic translations $v^{(i)}_1$ do not affect this uniformity.
		
		Moreover, conditional on any realization of
		$\mu_0,\mu_1,\mu_2$, Eq.~\eqref{eq:app-fd-affine-base-value} shows that
		$B_0$ is uniform in $\mathbb F_{2^s}$ because $\beta$ is independent and uniform. Thus,
		\begin{equation}
			(B_0,B_1,\ldots,B_d)
		\end{equation}
		is uniform in $\mathbb F_{2^s}^{d+1}$. 
		Hence $x|_W$ is a uniformly random $\mathbb F_{2^s}$-valued affine function .
		Therefore, each coordinate $x_1|_W,\dots,x_s|_W$ forms independent, uniformly random Boolean coordinate functions on $W$.
	\end{proof}
	
	We next show that random affine functions give an odd quadratic parity with at least constant probability.
	
	\begin{lemma}[Quadratic parity of random affine functions]
		\label{lem:app-fd-quadratic-parity}
		Let $W$ be an affine space over $\mathbb F_2$, let $T\subseteq W$ be a non-empty set containing at most four distinct points, and let $f,g:W\rightarrow\mathbb F_2$ be independent uniformly random affine functions. 
		Then
		\begin{equation}
			\Pr_{f,g}\left[
			\bigoplus_{\xi\in T}f(\xi)g(\xi)=1
			\right]
			\geq \frac{1}{4}.
			\label{eq:app-fd-random-affine-quadratic-parity}
		\end{equation}
	\end{lemma}
	\begin{proof}
		If the labels in $T$ are affinely independent, the evaluation vectors $(f(\xi))_{\xi\in T}$ and $(g(\xi))_{\xi\in T}$ are independent uniform elements of $\mathbb F_2^{|T|}$. 
		Hence the variables $f(\xi)g(\xi)$ are independent Bernoulli variables with mean $1/4$, and
		\begin{equation}
			\Pr_{f,g}\left[
			\bigoplus_{\xi\in T}f(\xi)g(\xi)=1
			\right]
			=\frac{1-(1-2\cdot\frac14)^{|T|}}{2}
			=\frac{1-2^{-|T|}}{2}
			\geq\frac{1}{4}.
			\label{eq:app-fd-affinely-independent-parity}
		\end{equation}
		Every set of at most three distinct points is affinely independent over $\mathbb F_2$. 
		Thus, the only remaining case is $|T|=4$ and $T$ is an affine plane. 
		After choosing affine coordinates on this plane, write
		\begin{equation}
			f(z_1,z_2)=a_0+a_1z_1+a_2z_2,
			\qquad
			g(z_1,z_2)=b_0+b_1z_1+b_2z_2.
		\end{equation}
		Summing over the four points of the plane gives
		\begin{equation}
			\bigoplus_{z\in\mathbb F_2^2}f(z)g(z)
			=a_1b_2\oplus a_2b_1.
		\end{equation}
		The right-hand side equals one for six of the sixteen ordered pairs of $(a_1,a_2),(b_1,b_2)$, so its probability is $3/8$. 
		This proves the claim.
	\end{proof}

	With the above two lemmas, we next show that for any signed multiplicity profile $c$, there are a constant ratio of pair of coordinates in the quadratic mixing that can detect $c\neq 0$ with high probability.
	
	Fix $c\in\mathcal C_r$, write $\Omega:=\operatorname{supp}(c)$, $R:=|\Omega|\leq2r\leq4k$ and
	\begin{equation}
		\mathcal F_4(\Omega)
		:=\{\varnothing\neq S\subseteq\Omega:|S|\leq4\}.
		\label{eq:app-fd-four-subsets}
	\end{equation}
	For a nonempty set $T\subseteq\Omega$ with $|T|\leq 4$ and two coordinate indices $p,q\in[s]$, define the quadratic parity
	\begin{equation}
		Q_T(p,q):=
		\bigoplus_{\xi\in T}x_p(\xi)x_q(\xi)
		\in\mathbb F_2.
		\label{eq:app-fd-quadratic-parity}
	\end{equation}
	
	\begin{lemma}[Dense detecting pairs]
		\label{lem:app-fd-dense-detecting-pairs}
		There are constants $A_0,c_2,C_2>0$ such that, if $s\geq A_0\log k$, then with probability at least $1-C_2e^{-c_2s}$ over $\mathsf A$, 
		for every $T\in\mathcal F_4(\Omega)$,
		\begin{equation}
			\mathcal D_T
			:=
			\left\{
			(p,q)\in[s]^2:p\neq q,\ Q_T(p,q)=1
			\right\}
			\label{eq:app-fd-detecting-pair-set}
		\end{equation}
		satisfies
		\begin{equation}
			|\mathcal D_T|
			\geq \frac18s(s-1).
			\label{eq:app-fd-detecting-pair-density}
		\end{equation}
	\end{lemma}
	\begin{proof}
		For $T\in\mathcal F_4(\Omega)$, fix $\xi^{(0)}\in T$, define
		\begin{equation}
			V_T:=\operatorname{span}_{\mathbb F_2}
			\{\xi-\xi^{(0)}:\xi\in T\},
			\qquad
			W_T:=\xi^{(0)}+V_T,
		\end{equation}
		and apply Lemma~\ref{lem:app-fd-affine-restriction}. 
		That lemma shows that $y$ is injective on $W_T$ except with probability at most $7\cdot2^{-s}$.
		Since
		\begin{equation}
			|\mathcal F_4(\Omega)|
			\leq
			\sum_{t=1}^4\binom{R}{t}
			\leq (R+1)^4
			\leq(4k+1)^4,
			\label{eq:app-fd-number-four-subsets}
		\end{equation}
		a union bound shows that, with failure probability at most $7(4k+1)^4 2^{-s}$, this injectivity holds simultaneously for every $T\in\mathcal F_4(\Omega)$.
		
		Condition on any realization of $\{\lambda_a\}_{a\neq2}$ with this simultaneous injectivity property. For a fixed $T$, the restrictions
		\begin{equation}
			x_1|_{W_T},\ldots,x_s|_{W_T}
		\end{equation}
		are independent uniformly random affine Boolean functions by Lemma~\ref{lem:app-fd-affine-restriction}. 
		Consequently, Lemma~\ref{lem:app-fd-quadratic-parity} gives, for every $p\neq q$,
		\begin{equation}
			\Pr_{\mu_0,\mu_1,\mu_2,\beta}
			\left[Q_T(p,q)=1\right]
			\geq\frac14.
			\label{eq:app-fd-coordinate-pair-detection}
		\end{equation}
		Define the normalized density
		\begin{equation}
			Z_T
			:=\frac{1}{s(s-1)}
			\sum_{p\neq q}\mathbbm 1[Q_T(p,q)=1].
		\end{equation}
		Then $\mathbb E[Z_T]\geq1/4$. 
		Replacing one of the independent coordinate functions $x_p|_{W_T}$ changes $Z_T$ by at most $2/s$.
		McDiarmid's inequality therefore yields
		\begin{equation}
			\Pr\left[Z_T<\frac18\right]
			\leq
			\exp\left(-\frac{s}{128}\right).
			\label{eq:app-fd-pair-density-concentration}
		\end{equation}
		Another union bound over $T\in\mathcal F_4(\Omega)$ gives total failure probability
		\begin{equation}
			7(4k+1)^4 2^{-s}
			+(4k+1)^4e^{-s/128}.
			\label{eq:app-fd-affine-density-failure}
		\end{equation}
		If $s\geq A_0\log k$ with $A_0$ sufficiently large, the expression in Eq.~\eqref{eq:app-fd-affine-density-failure} is at most $C_2e^{-c_2s}$ after adjusting the constants.
	\end{proof}
	
	We next show that the sampled permutations give at least one detecting pair on most output coordinates with high probability.
	
	\begin{lemma}[Shared-permutation coverage]
		\label{lem:app-fd-shared-permutation-coverage}
		Assume that Eq.~\eqref{eq:app-fd-detecting-pair-density} holds for every $T\in\mathcal F_4(\Omega)$. 
		If the constant $C_w$ in $w=\lceil C_w\log k\rceil$ is sufficiently large, then, with
		probability at least $1-e^{-c_3s/w}$ over $\Pi$, there is a set $G_c\subseteq[s]$ with $|G_c|\geq s/2$ such that for every
		$j\in G_c$ and every $T\in\mathcal F_4(\Omega)$, at least one $l\in[w]$ satisfies
		\begin{equation}
			Q_T\bigl(\pi_l(j),\sigma_l(j)\bigr)=1.
			\label{eq:app-fd-simultaneous-permutation-detection}
		\end{equation}
		Here $c_3>0$ is an absolute constant.
	\end{lemma}
	
	\begin{proof}
		Fix $j\in[s]$ and $T\in\mathcal F_4(\Omega)$. For each $l$, the pair $(\pi_l(j),\sigma_l(j))$ is uniform in $[s]^2$. 
		Thus, for $s\geq2$, Eq.~\eqref{eq:app-fd-detecting-pair-density} implies
		\begin{equation}
			\Pr_\Pi\left[
			Q_T\bigl(\pi_l(j),\sigma_l(j)\bigr)=1
			\right]
			\geq
			\frac{s(s-1)}{8s^2}
			\geq\frac1{16}.
		\end{equation}
		The permutation pairs are independent across $l$, and hence
		\begin{equation}
			\Pr_\Pi\left[
			Q_T\bigl(\pi_l(j),\sigma_l(j)\bigr)=0
			\text{ for every }l\in[w]
			\right]
			\leq e^{-w/16}.
			\label{eq:app-fd-one-set-uncovered}
		\end{equation}
		Call $j$ bad if Eq.~\eqref{eq:app-fd-simultaneous-permutation-detection} fails for at least one $T\in\mathcal F_4(\Omega)$, and let $B$ be the number of bad coordinates. 
		By Eqs.~\eqref{eq:app-fd-number-four-subsets}
		and~\eqref{eq:app-fd-one-set-uncovered},
		\begin{equation}
			\Pr_\Pi[j\text{ is bad}]
			\leq(4k+1)^4e^{-w/16}.
		\end{equation}
		Taking $C_w$ sufficiently large makes the right-hand side at most $1/16$, and therefore $\mathbb E_\Pi B\leq s/16$.
		
		We then apply the bounded-differences inequality for random permutations~\cite{McDiarmid2002Permutations}: if a function of $q$ independent uniform permutations changes by at most $L$ when one permutation is modified by a transposition, then, for a universal constant $C_{\mathrm{perm}}>0$,
		\begin{equation}
			\Pr[F-\mathbb EF\geq t]
			\leq
			\exp\left(-\frac{t^2}{C_{\mathrm{perm}}qsL^2}\right).
			\label{eq:app-fd-permutation-bounded-difference}
		\end{equation}
		Here $q=2w$. A transposition in one $\pi_l$ or $\sigma_l$ changes the sampled pair only at the two transposed input coordinates, so it can change $B$ by at most $L=2$. 
		Applying Eq.~\eqref{eq:app-fd-permutation-bounded-difference}, and using
		$\mathbb EB\leq s/16$, gives
		\begin{equation}
			\Pr_\Pi[B>s/2]
			\leq
			\exp\left(-c_3\frac{s}{w}\right)
			\label{eq:app-fd-bad-coordinate-concentration}
		\end{equation}
		for an absolute constant $c_3>0$. 
		Taking $G_c$ to be the complement of the bad coordinates proves the lemma.
	\end{proof}
	
	We now return to Lemma~\ref{lem:app-fd-moment-estimate}.
	\begin{proof}[Proof of Lemma~\ref{lem:app-fd-moment-estimate}]
		Let $\mathcal E_c$ be the event that the conclusion of Lemma~\ref{lem:app-fd-dense-detecting-pairs} holds and that, conditional on it, the shared permutations satisfy the conclusion of Lemma~\ref{lem:app-fd-shared-permutation-coverage}. 
		Combining the two lemmas gives
		\begin{equation}
			\Pr_{\mathsf A,\Pi}[\mathcal E_c]
			\geq
			1-C_2e^{-c_2s}-e^{-c_3s/w}.
			\label{eq:app-fd-combined-structural-event}
		\end{equation}
		After renaming constants, this is Eq.~\eqref{eq:app-fd-structural-failure}.
		
		Fix $(\mathsf A,\Pi)\in\mathcal E_c$ and $j\in G_c$, and write
		\begin{equation}
			\chi_j(\xi):=(-1)^{h_j(\xi)}.
		\end{equation}
		For every nonempty $T\subseteq\Omega$ with $|T|\leq 4$, Eqs.~\eqref{eq:bits-feature} and \eqref{eq:app-fd-quadratic-parity} give
		\begin{equation}
			\begin{aligned}
				\mathbb E_{\gamma_j}
				\prod_{\xi\in T}\chi_j(\xi)
				&={(-1)}^{\sum_{\xi\in T}y_j(\xi)}
				\prod_{l=1}^{w}
				\mathbb E_{\gamma_{j,l}}
				{(-1)}^{
					\gamma_{j,l}
					Q_T(\pi_l(j),\sigma_l(j))
				}.
				\label{eq:app-fd-character-factorization}
			\end{aligned}
		\end{equation}
		By Eq.~\eqref{eq:app-fd-simultaneous-permutation-detection}, some $l\in[w]$ has $Q_T(\pi_l(j),\sigma_l(j))=1$. 
		The corresponding factor in Eq.~\eqref{eq:app-fd-character-factorization} is $\mathbb E_{\gamma_{j,l}}(-1)^{\gamma_{j,l}}=0$. 
		Consequently,
		\begin{equation}
			\mathbb E_{\gamma_j}
			\prod_{\xi\in T}\chi_j(\xi)=0
			\qquad
			\text{for every }\varnothing\neq T\subseteq\Omega,
			\quad |T|\leq4.
			\label{eq:app-fd-character-cancellation}
		\end{equation}
		
		Since $S_j(c)=\sum_{\xi\in\Omega}c_\xi\chi_j(\xi)$, expanding the second moment and using Eq.~\eqref{eq:app-fd-character-cancellation} for each pair of distinct labels gives
		\begin{equation}
			\begin{aligned}
				\mathbb E_{\gamma_j}S_j(c)^2
				&=
				\sum_{\xi,\eta\in\Omega}
				c_\xi c_\eta
				\mathbb E_{\gamma_j}[\chi_j(\xi)\chi_j(\eta)]
				=\sum_{\xi\in\Omega}c_\xi^2.
				\label{eq:app-fd-second-moment-expansion}
			\end{aligned}
		\end{equation}
		
		For the fourth moment, consider an ordered quadruple $(\xi_1,\xi_2,\xi_3,\xi_4)\in\Omega^4$.  
		Because $\chi_j(\xi)^2=1$, the associated character is supported precisely on the set of labels appearing an odd number of times.  
		If this set is nonempty, it has size at most four, and the term vanishes by Eq.~\eqref{eq:app-fd-character-cancellation}.  
		Hence, a term survives only when every label occurs an even number of times: either one label occurs four times, or two distinct labels occur twice each.  
		There are $\binom{4}{2}=6$ orderings of the latter pattern, and therefore
		\begin{equation}
			\begin{aligned}
				\mathbb E_{\gamma_j}S_j(c)^4
				&=
				\sum_{\xi\in\Omega}c_\xi^4
				+6\sum_{\substack{\xi,\eta\in\Omega\\\xi<\eta}}
				c_\xi^2c_\eta^2
				\\
				&=
				3\left(\sum_{\xi\in\Omega}c_\xi^2\right)^2
				-2\sum_{\xi\in\Omega}c_\xi^4
				\\
				&\leq
				3\left(\sum_{\xi\in\Omega}c_\xi^2\right)^2.
				\label{eq:app-fd-fourth-moment-expansion}
			\end{aligned}
		\end{equation}
		These are exactly the two identities in Eq.~\eqref{eq:app-fd-second-fourth-moments}.
	\end{proof}

	\subsection{Circuit implementation}
	\label{app:finite-dimensional-depth}
	This subsection implements the ensemble and establishes the depth bound in Eq.~\eqref{eq:finite-dimensional-depth} using only the $m$ occupied system qubits.
	The idea for eliminating initialized auxiliary qubits is catalytic computation, in which an inactive portion of the data may be borrowed as workspace in an arbitrary state, provided that it is returned exactly after the operation.
	
	For the convenience of data distribution, recursive subdivision, and routing, we work in regular boxes defined as follows.
	\begin{definition}[Regular box]
		\label{def:regular-box}
		A set $S$ of sites is a regular box, if it can be written as $S=[l_1,r_1]\times[l_2,r_2]\times\cdots\times[l_\delta,r_\delta]$, where $r_i-l_i+1=2^{t_i}$ for some integer $t_i\in\mathbb{N}$ for each dimension, and $|t_i-t_j|\leq 1$ for every $i\neq j$.
		In other words, a regular box is a rectangular region whose side lengths are powers of two and differ by at most a factor of two.
		Its volume is denoted by $V:=|S|$.
	\end{definition}
	
	\begin{definition}[Routing numbers~\cite{Alon1993RoutingPermutation}]
		For a graph $G=(V,E)$ and a permutation $\pi$, define $\mathrm{rt}(G,\pi)$ as the minimum depth of swap gates to achieve the permutation $\pi$ on the architecture $G$.
		Define the routing number of $G$ as $\mathrm{rt}(G)=\max_\pi\mathrm{rt}(G,\pi)$.
	\end{definition}
	
	The following lemma shows that the cost of such a rearrangement is of the same order as the side lengths of the regular rectangular box, and there is a Hamiltonian path on which disjoint local interactions can be scheduled.
	\begin{lemma}[Routing in a regular box]
		\label{lem:app-fd-regular-routing}
		Let $B$ be a $\delta$-dimensional regular box of volume $V$.
		The routing number $\mathrm{rt}(B)$ is at most 
		$\cO_\delta(V^{1/\delta})$. Moreover, $B$ also contains a nearest-neighbor Hamiltonian path. 
	\end{lemma}
	\begin{proof}
		After translating coordinates, write
		$B=[L_1]\times\cdots\times[L_\delta]$, where each $L_i$ is a power of two and the side lengths differ by at most a factor of two. 
		We prove the statement by induction on $\delta$.
		For $\delta=1$, the box is the path $P_V=[V]$, which is itself a nearest-neighbor Hamiltonian path. 
		Consider an arbitrary permutation of the qubits.
		Route the qubits using odd--even transposition sort.  
		In every odd round, consider the disjoint edges $(2i+1,2i+2)$ whereas in every even round, consider $(2i,2i+1)$.
		On each active edge, swap its two qubits if and only if their destination labels appear in decreasing order.
		Each round consists of a single layer of disjoint nearest-neighbor swap gates.  
		Odd--even transposition sort terminates after at most $V$ rounds, so every qubit reaches its target position within $V$ nearest-neighbor layers~\cite[Proposition~2.3]{Casanova2008ParallelAlgorithms}.

		Suppose now that $\delta\geq 2$, and partition $B$ into the $L_\delta$ slices obtained by fixing the last coordinate.
		Each slice is a $(\delta-1)$-dimensional regular box of volume $V/L_\delta$.
		Selecting the same Hamiltonian path in each slice and connecting them end-to-end in alternating directions obtains a Hamiltonian path for the entire box.

		Suppose we want to implement a permutation $\tau$.
		Construct a bipartite multigraph $G$ whose left and right vertex sets are two copies of $[L_\delta]$.
		For each qubit, add one edge from its source slice to its destination slice.
		Let $C=[L_1]\times\cdots\times[L_{\delta-1}]$.
		The multigraph is $|C|=V/L_\delta$-regular, and hence admits a proper edge coloring with $|C|$ colors by K\H{o}nig's line-coloring theorem. 
		Every color appears exactly once at every vertex, so each color class is a perfect matching between the source and destination slices.
		Fix a bijection between the $|C|$ colors and the positions in $C$.
		The routing is carried out in three stages.
		\begin{enumerate}
			\item Inside every source slice, permute its $|C|$ qubits so that the qubit whose edge has color $c$ occupies the position of $C$ assigned to $c$. By the induction hypothesis, all source slices can perform these permutations in parallel in depth $\cO_{\delta}\!\left(|C|^{1/(\delta-1)}\right)$.
			\item For every color $c$, consider the sites assigned to $c$ in all $L_\delta$ slices.
			Because the color-$c$ edges form a perfect matching, the desired movements of the qubits on this path constitute a permutation of its $L_\delta$ sites.
			We route these permutations in parallel in depth at most $L_\delta$ by the one-dimensional case.
			\item Within each destination slice, route the qubits from their color positions to their prescribed destination sites. This also takes $\cO_{\delta}\!\left(|C|^{1/(\delta-1)}\right)$ depth.
		\end{enumerate}
		It follows that the routing depth satisfies
		\begin{equation}
			\begin{aligned}
				\operatorname{rt}(B)
				&\leq
				2O_{\delta}\!\left(|C|^{1/(\delta-1)}\right)+L_\delta\\
				&=\cO_\delta\!\left(V^{1/\delta}\right).
			\end{aligned}
		\end{equation}
		Here the last equality uses the regular-box condition that all $L_i$ differ by at most a factor of two.
		This completes the induction.
	\end{proof}
	
	For every integer $0\leq t\leq\log_2 V$, a regular box can also be partitioned into $2^t$ congruent regular subboxes.
	To obtain such a partition, repeat the following operation for $t$ rounds: bisect every current subbox along the same longest coordinate direction.
	
	The following lemma also establishes the routing property of balanced prefix windows.
	\begin{lemma}[Routing in a prefix window]
		\label{lem:app-fd-prefix-routing}
		For every $\delta\geq 1$, every constant $a\geq 1$ and prefix window $H_\delta(L,N)$ with $N\leq aL^\delta$,
		\begin{equation}
			\mathrm{rt}\bigl(H_\delta(L,N)\bigr)=\cO_{\delta,a}(L).
			\label{eq:app-fd-prefix-routing}
		\end{equation}
		In particular, $\mathrm{rt}\bigl(H_\delta(L,N)\bigr)=\cO_{\delta}(L)$ for balanced prefix windows.
	\end{lemma}
	\begin{proof}
		The proof is similar to that of Lemma~\ref{lem:app-fd-regular-routing}.
		We prove the claim by induction on $\delta$. 
		For $\delta=1$, the window is a line and the proof follows Lemma~\ref{lem:app-fd-regular-routing}.
		
		Suppose $\delta\geq 2$.
		Write $N=qL^{\delta-1}+r$ for $0\leq r<L^{\delta-1}$,
		then $H_\delta(L,N)$ can be written as 
		\begin{equation}
			H_\delta(L,N)=([q]\times C)\cup(\{q+1\}\times F),
		\end{equation}
		where $C=[L]^{\delta-1}$ is a $(\delta-1)$-dimensional hypercube, and $F=H_{\delta-1}(L,r)$ is a $(\delta-1)$-dimensional prefix window.
		
		Fix a permutation $\tau$ of the $N$ qubits.
		Construct a bipartite multigraph whose left and right vertices are copies of the occupied slices and whose edges represent qubits, with each edge connecting its source slice to its destination slice. 
		Every complete-slice vertex has degree $|C|$, while every partial-slice vertex has degree $|F|$. 
		Thus the maximum degree is $|C|$, and K\H{o}nig's line-coloring theorem gives a proper edge coloring with $|C|$ colors.
		
		Every color occurs exactly once at each complete-slice vertex. 
		A color is incident to the partial source slice if and only if it is incident to the partial destination slice, as the number of edges of that color is $q$ plus the corresponding incidence indicator when counted from either side.
		Choose a bijection from the colors to the sites of $C$ that maps the $|F|$ colors incident to the partial slices onto the sites of $F$.
		
		The remaining routing follows the proof of Lemma~\ref{lem:app-fd-regular-routing}.
		The complete-slice $C$ and the partial-slice $F$ routings also have depth $\cO_{\delta-1}(L)$ by the induction hypothesis, since $|C|,|F|\leq L^{\delta-1}$. 
		So the total depth is
		\begin{equation}
			2O_{\delta-1}(L)+(q+1)=\cO_{\delta,a}(L),
		\end{equation}
		since $q\leq aL$.
	\end{proof}
	
	We now introduce the formalization of catalytic computation.
	
	\begin{definition}[Catalytic computation]
		\label{def:app-fd-catalytic-target-add}
		Let $W$ be a workspace containing a source register $A$ and a disjoint target register $T$.
		A circuit catalytically implements the target addition for function $f$ if its action on computational-basis labels is
		\begin{equation}
			\ket{a,t,\omega}
			\longmapsto
			\ket{a,t+f(a),\omega}
			\label{eq:app-fd-catalytic-target-add}
		\end{equation}
		for every $a,t,\omega$.
		Here $\omega$ denotes the joint state of all other qubits inside the workspace $W$.
		Equivalently, the implemented unitary is $U_f\otimes I_{W\setminus A\setminus T}$ with $U_f:\ket{a}_A\ket{t}_T\longmapsto\ket{a}_A\ket{t+f(a)}_T$.
		Note that the catalysts may start in an arbitrary state, possibly entangled with $A$, $T$, or an external reference.
	\end{definition}
	
	One challenge in designing catalytic circuits is that we cannot explicitly store intermediate values, which makes it difficult to compose basic operations.
	The next lemma introduces a compiler for fixed, constant-size straight-line programs.
	Once an ordinary arithmetic straight-line program has been written using constant, linear, and bilinear instructions, the compiler converts it into a catalytic circuit.
	With this compiler, we can reduce our task to constructing suitable straight-line programs and implementing their primitive instructions.
	
	\begin{lemma}[A compiler for catalytic computation]
		\label{lem:catalytic-operations}
		Let $R$ be a commutative ring, and let $\mathcal P$ be a fixed arithmetic straight-line program whose registers range over finite free $R$-modules.
		The instructions of $\mathcal P$ have the following forms:
		\begin{enumerate}
			\item Constant instruction: output $T\gets C$ for some hardwired constant $C$.
			\item Linear instruction: given an input $E$, output $T\gets L(E)$ for some hardwired linear map $L$.
			\item Bilinear instruction: given two inputs $E$ and $F$, output $T\gets B(E,F)$ for some hardwired bilinear map $B$.
		\end{enumerate}
		Denote the output of $\mathcal{P}$ on the source registers $X$ with value $x$ by $\mathcal P(x)$.
		
		Suppose that, for every instruction occurring in $\mathcal P$, a catalytic target addition circuits of the corresponding form:
		\begin{equation}
			\ket{t,\omega}\longmapsto \ket{t+C,\omega},
		\end{equation}
		\begin{equation}
			\ket{e,t,\omega}\longmapsto \ket{e,t+L(e),\omega}
		\end{equation}
		or
		\begin{equation}
			\ket{e,f,t,\omega}\longmapsto \ket{e,f,t+B(e,f),\omega}
		\end{equation}
		is available, where $t$ is the initial value in the output register. 
		Then, using only dirty registers to hold intermediate values, the compiler produces a catalytic implementation
		\begin{equation}
			\label{eq:catalytic-program}
			\ket{x,t,\omega}  \longmapsto
			\ket{x,t+\mathcal P(x),\omega}
		\end{equation}
		of $\mathcal{P}$.
		
		The circuit for every primitive instruction is invoked at most $\cO_{|\mathcal P|}(1)$ times. 
		In particular, if all primitive circuits, including the required routing, have depth at most $D$, then the compiled circuit has depth at most $\cO_{|\mathcal P|}(D)$.
	\end{lemma}
	\begin{proof}
		The proof is by induction on the number of instructions in $\mathcal P$.
		The single-instruction case follows from the assumed primitive circuits.
		For a larger program, consider its final instruction and distinguish the following three cases.
		\begin{enumerate}
			\item A constant instruction $T\gets C$.
			The action
			\begin{equation}
				\ket{e}_E\ket{t}_T\ket{\omega}\longmapsto\ket{e}_E\ket{t+C}_T\ket{\omega}
			\end{equation}
			guaranteed by the condition is exactly what we want.
			\item A linear instruction $T\gets L(E)$. 
			By the induction hypothesis and the condition, we already have
			\begin{equation}
				\Phi:\ket{e}_E\ket{t}_T\ket{\omega}\longmapsto\ket{e}_E\ket{t+L(e)}_T\ket{\omega}
			\end{equation}
			and
			\begin{equation}
				\Psi:\ket{x}_X\ket{e}_E\ket{\omega}\longmapsto\ket{x}_X\ket{e+E(x)}_E\ket{\omega}.
			\end{equation}
			Executing
			\begin{equation}
				\Phi^{-1};\quad \Psi;\quad \Phi;\quad \Psi^{-1};
			\end{equation}
			implements the required target addition.
			
			\item A bilinear instruction $T\gets B(E,F)$. 
			By the induction hypothesis and the condition, we already have
			\begin{equation}
				\Phi:\ket{e}_E\ket{f}_F\ket{t}_T\ket{\omega}\longmapsto\ket{e}_E\ket{f}_F\ket{t+B(e,f)}_T\ket{\omega},
			\end{equation}
			\begin{equation}
				\Psi_e:\ket{x}_X\ket{e}_E\ket{\omega}\longmapsto\ket{x}_X\ket{e+E(x)}_E\ket{\omega}
			\end{equation}
			and
			\begin{equation}
				\Psi_f:\ket{x}_X\ket{f}_F\ket{\omega}\longmapsto\ket{x}_X\ket{f+F(x)}_F\ket{\omega}.
			\end{equation}
			Executing
			\begin{equation}\begin{aligned}
					\Phi;\quad \Psi_e;\quad \Phi^{-1};\quad \Psi_f;\\
					\Phi;\quad \Psi_e^{-1};\quad \Phi^{-1};\quad \Psi_f^{-1}.
			\end{aligned}\end{equation}
			implements the required target addition.
			To verify the construction, suppose that the initial values in $E$ and $F$ are $e$ and $f$, respectively. The net increment of the target is
			\begin{equation}
				B(e,f)-B(e+E(x),f)+B(e+E(x),f+F(x))-B(e,f+F(x))=B(E(x),F(x)).
			\end{equation}
		\end{enumerate}
		Each induction step invoke the preceding circuits a constant number of times.
		So the compiled program invokes each primitive $\cO_{|\mathcal{P}|}(1)$ times.
	\end{proof}

	We first show how to implement sparse linear maps on finite-dimensional architectures. 
	To facilitate qubit routing, we catalytically borrow inactive system qubits to serve as temporary routing space.
	
	\begin{lemma}[Catalytic implementations of sparse linear maps]
		\label{lem:sparse-linear-map}
		Fix $\delta\geq 1$, a finite field $S$, and a sparsity $\Delta\geq 1$. Let $A$ and $T$ be disjoint registers containing $n_A,n_T$ elements of $S$.
		Suppose $A,T$, together with some retained registers, occupy a regular box $B$ of volume $V$.
		Suppose that $V\geq\kappa_{\delta,S}(n_A+n_T)$ for a sufficiently large fixed constant $\kappa_{\delta,S}>0$.
		Let
		\begin{equation}
			L:S^{n_A}\longrightarrow S^{n_T}
		\end{equation}
		be a hardwired linear map over $S$ whose matrix has at most $\Delta$ nonzero entries in each row and each column. 
		Then the target-add
		\begin{equation}
			\ket{a,t,\omega}\longmapsto\ket{a,t+L(a),\omega}
			\label{eq:app-fd-sparse-linear-target-add}
		\end{equation}
		can be implemented catalytically inside $B$ by a nearest-neighbor circuit of depth $\cO_{\delta,S,\Delta}(V^{1/\delta})$. 
	\end{lemma}
	\begin{proof}
		Index the matrix of $L$ by target coordinates in its rows and source coordinates in its columns. $L_{v,u}$ is the coefficient with which the source entry $a_u$ contributes to the target entry $(L(a))_v$. 
		Construct a bipartite graph $G$ with vertex classes $A$ and $T$, placing an edge $(u,v)$ if and only if $L_{v,u}\neq 0$.
		Since $G$ has maximum degree at most $\Delta$, K\H{o}nig's line-coloring theorem partitions its edge set into at most $\Delta$ matchings,  written as $E=M_1\sqcup\cdots\sqcup M_\Delta$.
		
		Consider one matching $M_i$.
		Select a Hamiltonian path $P$ of $B$.
		Route the registers so that, for every edge $(u_j,v_j)\in M_i$, the field elements indexed by $u_j$ and $v_j$ occupy adjacent constant-length segments of $P$.
		Then perform the local updates
		\begin{equation}
			\ket{a_{u_j},t_{v_j},\omega_j}\longmapsto\ket{a_{u_j},t_{v_j}+L_{v_j,u_j}a_{u_j},\omega_j}
		\end{equation}
		for all $j$ in parallel.
		Applying the inverse routing permutation then returns every source, target, and catalytic qubit to its original site.
		
		For each matching $M_i$, the routing step takes $\cO_{\delta}(V^{1/\delta})$ depth by Lemma~\ref{lem:app-fd-regular-routing}, and the arithmetic takes $\cO_{S}(1)$ depth.
		Therefore, all $\Delta$ matchings require a total depth of $\cO_{\delta,S,\Delta}(V^{1/\delta})$.
	\end{proof}
	
	Next, we show how to implement polynomial evaluation over finite fields, a crucial primitive required for the hashing step of our low-depth random unitary construction, via the following two lemmas.
	
	\begin{lemma}[Catalytic polynomial multiplication]
		\label{lem:poly-mul}
		Fix $\delta\geq 1$ and a finite field $S$ of characteristic $2$. 
		Let $A,B\in S[z]_{<s}$ be two polynomials of degree at most $s-1$, written as
		\begin{equation}
			A(z)=\sum_{i=0}^{s-1}a_iz^i,\qquad B(z)=\sum_{i=0}^{s-1}b_iz^i.
		\end{equation}
		Let $C\in S[z]_{<2s}$ be an arbitrary target polynomial with degree at most $2s-1$.
		Suppose the workspace $W$ containing the registers of $A,B,C$ occupies a regular box of volume $V=\Theta_{\delta,S}(s)$, with $V\geq \kappa_{\delta,S}s$ for a sufficiently large constant $\kappa_{\delta,S}>0$.
		Then the operator
		\begin{equation}
			\ket{A,B,C,\omega}\longmapsto\ket{A,B,C+AB,\omega}
		\end{equation}
		polynomial multiplication can be realized by a $\delta$-dimensional nearest-neighbor circuit of depth $\cO_{\delta,S}(s^{1/\delta})$.
	\end{lemma}
	\begin{proof}
		We use a constant-radix Toom--Cook evaluation--interpolation recursion.
		Choose a sufficiently large constant radix $\varrho$ which is a power of two, and set $n_{\mathrm{ev}}=2\varrho-1$. 
		Pad $s$ to a power of $\varrho$.
		Work in a constant-degree extension $S'\supseteq S$ containing distinct points $\alpha_0,\ldots,\alpha_{n_{\mathrm{ev}}-1}$. 
		Fix an $S$-basis of $S'$, and encode each $S'$-coefficient by the corresponding constant number of $S$-symbols.
		For given polynomials $A,B,C$, we transform them to polynomials over $S'$, run the following algorithm, and then transform them back to $S$.
		
		We first describe the algebraic operations. 
		For given polynomials $A,B$ with degree at most $s-1$, put $b=\lceil s/\varrho\rceil$. 
		After hardwired zero padding, split
		\begin{equation}
			A(z)=\sum_{q=0}^{\varrho-1}z^{qb}A^{(q)}(z),
			\qquad
			B(z)=\sum_{q=0}^{\varrho-1}z^{qb}B^{(q)}(z),
		\end{equation}
		where $A^{(q)},B^{(q)}\in S[z]_{<b}$.  
		
		Introduce an auxiliary indeterminate $X$ and define
		\begin{equation}
			\tilde A(X,z)=\sum_{u=0}^{\varrho-1}X^{u}A^{(u)}(z),\quad \tilde B(X,z)=\sum_{u=0}^{\varrho-1}X^{u}B^{(u)}(z).
		\end{equation}
		For $i\in[0,n_{\mathrm{ev}}-1]$, define the point evaluations
		\begin{equation}
			A_i(z):=\tilde A(\alpha_i,z)=\sum_{q=0}^{\varrho-1}\alpha_i^qA^{(q)}(z),
			\qquad
			B_i(z):=\tilde B(\alpha_i,z)=\sum_{q=0}^{\varrho-1}\alpha_i^qB^{(q)}(z).
		\end{equation}
		
		Let $\beta_{l,i}$ be the entries of the inverse of the $n_{\mathrm{ev}}\times n_{\mathrm{ev}}$ Vandermonde matrix $(\alpha_i^l)_{i,l}$.
		Define the fixed additive maps
		\begin{equation}
			L_i(T):=\sum_{l=0}^{2\varrho-2}
			z^{l b}\beta_{l,i}T.
			\label{eq:app-fd-interpolation-map}
		\end{equation}
		Vandermonde interpolation gives
		\begin{equation}\begin{aligned}
				\sum_{i=0}^{n_{\mathrm{ev}}-1}L_i(A_iB_i)
				&=\sum_{l=0}^{2\varrho-2}\sum_{i=0}^{n_{\mathrm{ev}}-1}z^{l b}\beta_{l,i}A_iB_i\\
				&=\sum_{l=0}^{2\varrho-2}z^{l b}\sum_{i=0}^{n_{\mathrm{ev}}-1}\sum_{q=0}^{n_{\mathrm{ev}}-1}\beta_{l,i}\alpha_{i}^q\sum_{j=0}^qA^{(j)}B^{(q-j)}\\
				&=\sum_{l=0}^{2\varrho-2}z^{l b}\sum_{j=0}^lA^{(j)}B^{(l-j)}\\
				&=AB.
				\label{eq:app-fd-evaluation-interpolation}
		\end{aligned}\end{equation}
		Each multiplication $A_iB_i$ can be computed recursively, yielding an algorithm for computing $AB$.
		
		It remains to implement the recursion catalytically on the $\delta$-dimensional architecture. At scale $s$, the source registers $A,B$, the target register $C$, and the remaining workspace occupy a regular box $W$, with all coefficients represented over $S'$. The required update is
		\begin{equation}
			\ket{A,B,C,\omega}\longmapsto\ket{A,B,C+AB,\omega}.
		\end{equation}
		For $s$ below a fixed constant, the update is implemented by a constant-size local unitary.
		For larger $s$, the algorithm first allocates three registers $E_i,F_i,T_i$ in $W$ for each $i\in [0,2\varrho-2]$.
		We define the straight-line program $\mathcal{P}(A,B)$ to compute $AB$ by the following four instructions:
		\begin{enumerate}
			\item Set $E_i\gets A_i(A)$ for every $i$ in parallel.
			\item Set $F_i\gets B_i(B)$ for every $i$ in parallel.
			\item Set $T_i\gets E_iF_i$ for every $i$ in parallel.
			\item Output $\sum_i L_i(T_i)$.
		\end{enumerate}
		Following Eq.~\eqref{eq:app-fd-evaluation-interpolation}, the output will be the product $AB$.
		
		Let 
		\begin{equation}\begin{aligned}
				\Phi&:\ket{A}\otimes(\bigotimes_i\ket{E_i})\otimes\ket{\omega}\longmapsto\ket{A}\otimes(\bigotimes_i\ket{E_i+A_i(A)})\otimes\ket{\omega},\\
				\Psi&:\ket{B}\otimes(\bigotimes_i\ket{F_i})\otimes\ket{\omega}\longmapsto\ket{B}\otimes(\bigotimes_i\ket{F_i+B_i(B)})\otimes\ket{\omega},\\
				\mathsf{M}&:(\bigotimes_i\ket{E_i,F_i,T_i})\otimes \ket{\omega}\longmapsto(\bigotimes_i\ket{E_i,F_i,T_i+E_iF_i})\otimes \ket{\omega},\\
				\mathsf{L}&:(\bigotimes_i\ket{T_i})\otimes\ket{C}\otimes \ket{\omega}\longmapsto(\bigotimes_i\ket{T_i})\otimes\ket{C+\sum_i L_i(T_i)}\otimes \ket{\omega}
		\end{aligned}\end{equation}
		be the corresponding catalytic target addition mapping of these instructions, respectively. 
		Lemma~\ref{lem:catalytic-operations} therefore implements $C\gets C+\mathcal P(A,B)=C+AB$ catalytically, invoking these primitives only a constant number of times.
		
		Each of $\Phi,\Psi,\mathsf{L}$ can be implemented in depth $\cO_{\delta,\varrho}(s^{1/\delta})$ by Lemma~\ref{lem:sparse-linear-map}, as the maps are linear maps with sparsity $\cO_{\varrho}(1)$.
		The remaining task is to implement $\mathsf{M}$.
		Partition the regular box into $\varrho$ congruent regular subboxes $W_0,\ldots,W_{\varrho-1}$.
		Route $E_{2i},F_{2i},T_{2i}$ to $W_i$ in one routing process for each $i\in[0,\varrho-1]$.
		Then we recursively invoke the algorithm for each regular box $W_i$.
		Because each recursive invocation is confined to its own subbox $W_i$, all invocations run in parallel.
		After that, we route $E_{2i+1},F_{2i+1},T_{2i+1}$ to $W_i$ in one routing process for each $i\in[0,\varrho-2]$, and recursively invoke the algorithm for them in parallel.
		Finally, route the registers back to the original positions.
		The resulting depth is $\cO_\delta(s^{1/\delta})$ plus a constant number of recursive depths at scale $s/\varrho$.
		
		The workspace does not accumulate across recursion levels, as the inactive data of previous levels remain inside $W$ and serve as catalysts at the next level.
		
		It remains to analyze the depth and catalyst count. If $D(s)$ denotes the circuit depth at scale $s$, then
		\begin{equation}
			D(s)=C_0s^{1/\delta}+C_1D(s/\varrho).
		\end{equation}
		Here $C_0$ may depend on $\delta,S$ and $\varrho$, whereas $C_1$ is an absolute constant.
		Choosing the constant radix $\varrho$ so that $C_1\varrho^{-1/\delta}<1$ gives $D(s)=\cO_{\delta,S}(s^{1/\delta})$.
		
		The required size $\kappa_{\delta,S}s$ of workspace is proved by induction.
		A recursion step with constant scale $s$ takes constant space for multiplication.
		For a recursion step with larger scale $s$, the workspace required to implement the maps $\Phi,\Psi,\mathsf{L}$ is $\cO_{\delta,S}(s)$ by Lemma~\ref{lem:sparse-linear-map}.
		For the map $\mathsf{M}$, we have $|W|/\varrho$ qubits for each subproblem of scale $s/\varrho$.
		This satisfies $|W|/\varrho\geq \frac{\kappa_{\delta,S}s}{\varrho}=\kappa_{\delta,S}\frac{s}{\varrho}$ and is therefore sufficient by the induction hypothesis.
		As each map requires $\kappa_{\delta,S}s$ space, their finite composition also takes $\kappa_{\delta,S}s$ space.
		
		The algorithm over $S'$ is already implemented.
		For a polynomial over $S$, conversions between its $S$-representation and its $S'$-representation are sparse linear maps over $S$.
		So these conversions have catalytic circuit depth $\cO_{\delta,S}(s^{1/\delta})$ by Lemma~\ref{lem:sparse-linear-map}.
		As polynomial multiplication over $S'$ is bilinear, one more application of Lemma~\ref{lem:catalytic-operations} completes the proof.
	\end{proof}
	
	\begin{lemma}[Catalytic fixed-degree finite-field arithmetic]
		\label{lem:app-fd-catalytic-field-arithmetic}
		Fix $\delta\geq 1$ and a constant degree $d\geq0$. 
		For $s\geq 1$, define the field
		\begin{equation}
			K:=\mathbb F_{2^s}
			\simeq \mathbb F_2[z]/(p(z)),
			\label{eq:app-fd-field-representation}
		\end{equation}
		for some hardwired monic irreducible polynomial $p$.
		Let
		\begin{equation}
			P(u)=\sum_{i=0}^{d}a_iu^i\in K[u]
		\end{equation}
		be a hardwired constant-degree polynomial.
		Suppose the workspace $W$ containing registers $X,Y$ occupies regular box of volume $V=\Theta_{\delta,d}(s)$, with $V\geq\kappa_{\delta,d}s$ for a sufficiently large constant $\kappa_{\delta,d}>0$. Then the operator
		\begin{equation}
			\ket{x,y,\omega}
			\longmapsto
			\ket{x,y+P(x),\omega}
			\label{eq:app-fd-catalytic-polynomial-add}
		\end{equation}
		has a nearest-neighbor implementation of depth $\cO_{\delta,d}(s^{1/\delta})$.
		
		The same statement holds when the source contains only $u\leq s$ physical bits and is embedded into $K$ via hardwired zero-padding.
	\end{lemma}
	\begin{proof}
		Let $\mathrm{Red}:\mathbb F_2[z]_{<2s}\longrightarrow\mathbb F_2[z]_{<s}$ be the reduction map modulo $p$. 
		First, consider an arithmetic algorithm for computing $\mathrm{Red}$.
		For an input $T$, write 
		\begin{equation}
			\label{eq:divmod}
			T=pQ+R,
		\end{equation}
		where $Q=Q(T)$ and $R=R(T)$ has degree less than $s$.
		
		For a polynomial $f$ with degree at most $e$, define its reversal $\mathrm{rev}_e(f)=z^ef(z^{-1})$.
		Reversing Eq.~\eqref{eq:divmod} at length $2s-1$ gives
		\begin{equation}
			\mathrm{rev}_{2s-1}(T)=\mathrm{rev}_{s}(p)\mathrm{rev}_{s-1}(Q)+\mathrm{rev}_{2s-1}(R).
		\end{equation}
		The last term is divisible by $z^s$ as $\operatorname{deg}(R)<s$. So
		\begin{equation}
			\mathrm{rev}_{2s-1}(T)\bmod z^s=\mathrm{rev}_{s}(p)\mathrm{rev}_{s-1}(Q)\bmod z^s.
		\end{equation}
		As $p$ is monic, $\operatorname{rev}_s p$ has constant coefficient $1$ and is therefore invertible modulo $z^s$.
		Define
		\begin{equation}
			h:=(\operatorname{rev}_s p)^{-1}\bmod z^s,
		\end{equation}
		which can be hardwired in the circuit.
		Then
		\begin{equation}\begin{aligned}
				Q(T)&=\operatorname{rev}_{s-1}[(\operatorname{rev}_{2s-1}T)h\bmod z^s],\\
				R(T)&=\!\left[T-pQ(T)\right]\bmod z^s.
				\label{eq:app-fd-reduction}
		\end{aligned}\end{equation}
		
		Therefore, the map $\mathrm{Red}(T)=R(T)$ can be written as a straight-line program $\mathcal{R}(T)$ containing reversal, truncation, subtraction and polynomial multiplication.
		Reversal and truncation are sparse linear maps, which are handled by Lemma~\ref{lem:sparse-linear-map}.
		The multiplications are implemented by Lemma~\ref{lem:poly-mul}.
		Each map requires depth $\cO_{\delta}(s^{1/\delta})$.
		Lemma~\ref{lem:catalytic-operations} then compiles the entire reduction map into a catalytic circuit of depth $\cO_{\delta,d}(s^{1/\delta})$.
		
		Because $\mathrm{Red}$ is linear, it can be treated as an atomic linear instruction of a straight-line program.
		Combining polynomial multiplication with this reduction implements field multiplication in catalytic depth $\cO_{\delta,d}(s^{1/\delta})$ using the compiler.
		
		Now we express $P(x)$ as a straight-line program.
		Start with the constant instuctions $E_0\gets 1$ and $S_0\gets a_0$, set $E_i\gets E_{i-1}X$ for $1\leq i\leq d$, and accumulate $S_i\gets S_{i-1}+a_iE_i$ for $1\leq i\leq d$.
		Then the result $P(x)$ is stored in $S_d$.
		Each power update is bilinear, while each accumulation is linear together with a bilinear multiplication.
		The compiler, the catalytic field-multiplication circuit above, and Lemma~\ref{lem:sparse-linear-map} therefore give the claimed depth.
	\end{proof}
	
	\begin{proof}[Proof of Theorem~\ref{thm:finite-dimensional-ensemble}]
		The previous subsection proves Eq.~\eqref{eq:finite-dimensional-distance}.
		Now we show how to implement the random unitaries described in Sec. \ref{app:finite-dimensional-construction}.
		
		Inside the balanced prefix window, we can find a regular box $B$ of size $\Omega_{\delta}(m)$ by selecting the maximum hypercube with side length a power of two.
		Choose $r_0$ in Eq.~\eqref{eq:app-fd-block-label} to be sufficiently large so that the computations in Lemma~\ref{lem:poly-mul} and Lemma~\ref{lem:app-fd-catalytic-field-arithmetic} can be finished in this regular box.
		By Lemma~\ref{lem:app-fd-prefix-routing}, the system has routing number at most $\cO_\delta(m^{1/\delta})$. 
		Therefore, we can permute the data to be processed into the regular box $B$ with circuit depth $\cO_\delta(m^{1/\delta})$ each time.
		
		The updates in Steps 1 and 2 can each be implemented by sequentially applying Eq.~\eqref{eq:app-fd-catalytic-polynomial-add} a constant number of times with different choices of systems and $P$.
		This can be done in depth $\cO_{\delta}(m^{1/\delta})$ by rearranging the source and target qubits into $B$ and invoking Lemma~\ref{lem:app-fd-catalytic-field-arithmetic} for each of the constantly many polynomials $P$.
		
		Step~3 is performed separately for each $l\in[w]$.
		For a fixed pair of permutation $\pi_l$ and $\sigma_l$, first separate the indices with $\pi_l(j)=\sigma_l(j)$.
		Since $X_p^2=X_p$ on computational-basis bits, all corresponding terms are disjoint CNOT updates $Y_j\gets Y_j\oplus\gamma_{j,l}X_p$ and can be executed in one batch.
		After removing these loops, the multigraph on $[s]$ with edges $\{\pi_l(j),\sigma_l(j)\}$ has maximum degree at most two, and hence can be 
		partitioned into three matchings.
		For each matching $M_i$, we rearrange $X_{\pi_l(j)},X_{\sigma_l(j)}$ and $Y_j$ into adjacent positions on the Hamiltonian path of $B$ for each $j\in M_i$.
		The corresponding disjoint three-qubit unitaries realize Eq.~\eqref{eq:rand-quad} in constant depth locally.
		This step costs $\cO_{\delta}(wm^{1/\delta})=\cO_{\delta}(m^{1/\delta}\log k)$ depth.
		
		Step~4 consists only of single-qubit gates, and Step~5 reverses Steps~1--3. The total circuit depth is therefore $\cO_\delta(m^{1/\delta}\log k)$, as claimed in Eq.~\eqref{eq:finite-dimensional-depth}.
	\end{proof}
	
	\subsection{Repetition and gluing to form designs}
	\label{app:finite-dimensional-repetition}
	Theorem~\ref{thm:finite-dimensional-ensemble} constructs an approximation of a random diagonal unitary.
	To achieve a high-precision approximate design, the subsequent step integrates three independent samples from this ensemble into a Hadamard--diagonal block, and then amplifies the design by repeatedly applying this combined block.
	
	For system sizes $m$ below the threshold $A_\delta\log k$ required by Theorem~\ref{thm:finite-dimensional-ensemble}, we replace the Hadamard--diagonal construction with random Pauli rotations. 
	This alternative approach provides a constant spectral gap, guaranteed as follows:
	
	\begin{fact}[Random Pauli rotation gap~{\cite[Theorem~3.16]{baer2026random}}]
		\label{fact:app-fd-brickwork-gap}
		For any $m \geq 1$, let $\mathcal{P}_m$ denote the set of non-identity $m$-qubit Pauli operators. Let $\nu_{\mathrm{RPR},m}$ be the ensemble generated by uniformly and independently sampling an operator $P \in \mathcal{P}_m$ and an angle $\theta \in [0, 2\pi)$, and applying the rotation
		\begin{equation}
			R(P,\theta) := e^{\mathrm{i}\theta P}.
		\end{equation}
		Then, there exists a universal constant $g > 0$ (which can be taken as $g = 1/16$) such that, simultaneously for all moment orders $r \geq 1$, the spectral gap satisfies
		\begin{equation}
			\label{eq:app-fd-rpr-gap}
			\lambda_r(\nu_{\mathrm{RPR},m}) \leq 1-g.
		\end{equation}
	\end{fact}
	
	\begin{theorem}[$m$-qubit approximate designs without ancillae on finite-dimensional lattice]
		\label{thm:fd-small-design}
		Let $0<\epsilon\leq 1$ and $k\geq 2$, on every $m$-qubit balanced prefix window, there exists a multiplicative $\epsilon$-approximate unitary $k$-design with ancilla-free circuit depth
		\begin{equation}
			\label{eq:app-fd-m-qubit-design-depth}
			\cO_\delta\left(m^{1/\delta}\left(k+\frac{\log(\epsilon^{-1})}{m}\right)\log^2 k\right)
		\end{equation}
		on a $\delta$-dimensional architecture.
	\end{theorem}
	\begin{proof}
		First consider the case that $m\geq A_{\delta}\log k$ where Theorem~\ref{thm:finite-dimensional-ensemble} applies.
		Let $\mu_{\delta,m}$ denote the ensemble of unitaries
		\begin{equation}
			U=D_3H^{\otimes m}D_2H^{\otimes m}D_1,
		\end{equation}
		where $D_1,D_2,D_3$ are independently sampled from the ensemble $\nu_{\delta,m}$ constructed in Theorem~\ref{thm:finite-dimensional-ensemble}.
		The stability estimate in Lemma~\ref{lem:app-HD-perturbation} together with Eq.~\eqref{eq:finite-dimensional-distance} gives
		\begin{equation}\begin{aligned}
				\label{eq:fd-block-gap}
				\lambda_k(\mu_{\delta,m})
				&\leq
				\frac{9k^4}{2^m}
				+3C_{\delta}
				\exp\!\left[-\frac{c_{\delta}m}{\log k}\right]\\
				&\leq \exp\!\left[-\frac{a_{\delta}m}{\log k}\right]
		\end{aligned}\end{equation}
		for some constant $a_\delta>0$.
		Independent repetition and Eq.~\eqref{eq:app-gap-amplification} imply
		\begin{equation}\label{eq:gap_repetition_hadamard-diagonal}
			\lambda_k(\mu_{\delta,m}^{*T})\leq \exp\!\left[-\frac{a_{\delta}mT}{\log k}\right]
		\end{equation}
		for every integer $T\geq 1$.
		Fact~\ref{fact:TPE-multiplicative-comparison} bounds the multiplicative error by
		\begin{equation}
			\epsilon_{\mathrm{mult}}\leq 2^{2mk}\lambda_k(\mu_{\delta,m}^{*T})\leq \exp(2mk\log 2-\frac{a_{\delta}mT}{\log k}).
		\end{equation}
		Selecting $T
		=\left\lceil
		\frac{\log k}{a_{\delta}m}
		\left(2mk\log 2+\log(\epsilon^{-1})\right)\right\rceil$ makes $\epsilon_{\mathrm{mult}} \le \epsilon$.
		One sample of $\mu_{\delta,m}$ contains three diagonal circuits of depth $\cO_\delta\left(m^{1/\delta}\log k\right)$ and two transversal Hadamard layers.
		Repeating it $T$ times proves the depth claimed in Eq.~\eqref{eq:app-fd-m-qubit-design-depth}.
		
		It remains to consider the case that $m<A_{\delta}\log k$.
		Use the random-Pauli-rotation ensemble $\nu_{\mathrm{RPR},m}$ from Fact~\ref{fact:app-fd-brickwork-gap}.
		First note that every $R(P,\theta)$ has an ancilla-free implementation of depth
		\begin{equation}
			\label{eq:app-fd-rpr-single-step-depth}
			\cO_\delta\!\left(m^{1/\delta}\log m\right)
		\end{equation}
		on the given architecture.
		To see this, let $S\subseteq[m]$ be the support of $P$, write $Z_S:=\prod_{j\in S}Z_j$, and choose a tensor product $C_P$ of single-qubit Clifford gates such that $C_PP C_P^\dagger=Z_S$.
		A balanced binary tree computes the parity $\bigoplus_{j\in S}x_j$ into one of the qubits in $S$ using $\lceil\log_2|S|\rceil$ layers of disjoint CNOT gates.
		Applying $e^{\mathrm i\theta Z}$ to the parity qubit and reversing the tree implements $e^{\mathrm i\theta Z_S}$.  
		Conjugating $e^{\mathrm i\theta Z_S}$ by $C_P$ gives
		\begin{equation}
			R(P,\theta)=C_P^\dagger e^{\mathrm i\theta Z_S}C_P.
		\end{equation}
		
		It remains to compile every CNOT layer into the physical geometry. 
		The balanced prefix window contains a regular box of volume $\Omega_\delta(m)$ and has routing number $\cO_\delta(m^{1/\delta})$ by Lemma~\ref{lem:app-fd-prefix-routing}.
		The regular box contains a nearest-neighbor matching of size $\Omega_\delta(m)$. 
		Partition a CNOT layer into $\cO_\delta(1)$ batches, route the endpoints in each batch to distinct edges of this physical matching, apply the CNOT gates in parallel, and undo the routing. 
		Hence each CNOT layer has physical depth $\cO_\delta(m^{1/\delta})$. 
		Thus each of the $\cO(\log m)$ layers has depth $\cO_\delta(m^{1/\delta})$, proving Eq.~\eqref{eq:app-fd-rpr-single-step-depth}.
		The construction changes only the $m$ system qubits and exactly reverses all temporary parity computations, so it introduces no ancillary qubits.
		Fact~\ref{fact:app-fd-brickwork-gap} and Eq.~\eqref{eq:app-gap-amplification} give
		\begin{equation}
			\lambda_k(\nu_{\mathrm{RPR},m}^{*T})
			\leq(1-g)^T
			=e^{-\log(\frac{1}{1-g})T}.
		\end{equation}
		The TPE-to-multiplicative conversion therefore shows that
		\begin{equation}
			\label{eq:app-fd-rpr-repetition-count}
			T=
			\left\lceil
			\frac{2mk\log 2+\log(\epsilon^{-1})}{\log(\frac{1}{1-g})}
			\right\rceil
		\end{equation}
		repetitions suffice.
		Combining Eqs.~\eqref{eq:app-fd-rpr-single-step-depth} and~\eqref{eq:app-fd-rpr-repetition-count}, their total depth is
		\begin{equation}
			\label{eq:app-fd-small-register-depth}
			\cO_\delta\!\left(
			m^{1/\delta}\log m
			\left[mk+\log(\epsilon^{-1})\right]
			\right).
		\end{equation}
		Finally, $m<A_{\delta}\log k$ implies $m\log m=\cO_\delta\!\left(\log^2 k\right)$. 
		Consequently, Eq.~\eqref{eq:app-fd-small-register-depth} becomes
		\begin{equation}
			\cO_\delta\!\left(
			m^{1/\delta}
			\left(k+\frac{\log(\epsilon^{-1})}{m}\right)
			\log^2 k
			\right),
		\end{equation}
		which establishes Eq.~\eqref{eq:app-fd-m-qubit-design-depth}.
		Both regimes use only the $m$ system qubits, completing the proof.
	\end{proof}
	
	The following gluing lemma combines overlapping block designs into a design on their union.
	\begin{fact}[Gluing lemma~{\cite[Lemma 8]{Schuster2024RandomUnitaries}}]
		\label{fact:app-fd-gluing}
		Let $A,B,C$ be disjoint qubit registers with $|B|=\xi$, and suppose $2^{\xi}\geq2k^2$.
		Independently sample $U_{AB}$ and $U_{BC}$ from multiplicative $\epsilon_{AB}$- and multiplicative $\epsilon_{BC}$-approximate unitary $k$-designs on $AB$ and $BC$, respectively.
		Then the ensemble of
		\begin{equation}
			V_{ABC}:=U_{AB}U_{BC}
		\end{equation}
		is a multiplicative $\epsilon_{ABC}$-approximate unitary $k$-design on $ABC$, where
		\begin{equation}
			1+\epsilon_{ABC}\leq
			(1+\epsilon_{AB})(1+\epsilon_{BC})
			\left(1+\cO(k^2 2^{-\xi})\right).
		\end{equation}
	\end{fact}
	
	The result is summarized by the following theorem.
	\begin{theorem}[Low-depth approximate designs without ancillae on finite-dimensional lattice]
		\label{thm:fd-large-design}
		Fix $\delta\geq 1$. For every $n,k\geq 2$ and $2^{-\eta nk}\leq \epsilon\leq 1$ for any fixed constant $\eta>0$, there exists an $n$-qubit multiplicative $\epsilon$-approximate unitary $k$-design with ancilla-free circuit depth
		\begin{equation}
			\label{eq:fd-large-design-depth}
			\cO_{\delta,\eta}\left((\log(n)+\log(k)+\log(\epsilon^{-1}))^{1/\delta}k\log^2 k\right)
		\end{equation}
		on a $\delta$-dimensional architecture.
	\end{theorem}
	\begin{proof}
		Set $M=\log(n)+\log(k)+\log(\epsilon^{-1})$ as the order of the size of a block design.
		Let $l=2^{\lceil\frac{\log_2(\kappa_\delta M)}{\delta}\rceil}$ for some constant $\kappa_\delta>0$ to be fixed.
		Then $l$ has order $\Theta(M^{1/\delta})$.
		
		First, consider the case $n<2l^\delta$. 
		Let $L_n:=2^{\lfloor\delta^{-1}\log_2 n\rfloor}$ and place the qubits on the balanced prefix window $H_\delta(L_n,n)$. 
		Since $l^\delta=\cO_\delta(M)$ and $n<2l^\delta$, one has $n=\cO_\delta(M)$. 
		Applying Theorem~\ref{thm:fd-small-design} to the whole system gives depth
		\begin{equation}
			\cO_\delta\!\left(
			n^{1/\delta}\left(k+\frac{\log(\epsilon^{-1})}{n}\right)\log^2 k
			\right)
			=\cO_{\delta,\eta}\!\left(M^{1/\delta}k\log^2 k\right)
		\end{equation}
		as $\epsilon\geq 2^{-\eta nk}$. Hence, it remains to treat $n\geq2l^\delta$.
		
		Place the $n$ qubits on the prefix window $H_\delta(l,n)$.
		Partition the occupied slices consecutively from left to right into $t\geq 2$ blocks $Q_1,Q_2,\ldots,Q_t$ such that 
		\begin{equation}
			l^\delta\leq |Q_i|\leq 2l^\delta.
		\end{equation}
		All blocks except possibly $Q_t$ are unions of complete slices, and $Q_t$ may additionally contain the final partial slice.
		
		Let $R_i:=Q_i\cup Q_{i+1}$ for $i\in[t-1]$. Then
		\begin{equation}
			\label{eq:R-size}
			2l^\delta\leq |R_i|\leq 4l^\delta.
		\end{equation}
		We will construct an approximate unitary design on each patch $R_i$, and use gluing lemma to glue them together.
		To construct a unitary design on $R_i$, we need to prove that each of the $R_i$ is a balanced prefix window.
		As each $Q_i$ is some consecutive slice inside the tube $H_\delta(l,n)$, so does $R_i$.
		Therefore, each $R_i$ forms a prefix window after translation.
		The balance condition comes from Eq.~\eqref{eq:R-size}.

		Now we show the unitary design construction.
		Let $\epsilon'=\frac{\epsilon}{8t}$.
		For each $i\in[t-1]$, independently sample a multiplicative $\epsilon'$-approximate unitary $k$-design $U_i$ on $R_i$ from Theorem~\ref{thm:fd-small-design}.
		Arrange these local unitaries in the two-layer brickwork product
		\begin{equation}
			\label{eq:app-fd-two-layer-product}
			U_{\mathrm{glue}}:=\left(\prod_{\substack{i\in[t-1]\\ i\ \mathrm{even}}}U_i\right)
			\left(\prod_{\substack{i\in[t-1]\\ i\ \mathrm{odd}}}U_i\right).
		\end{equation}
		Define $W_1:=U_1$ and
		\begin{equation}
			W_i:=\begin{cases}
				U_iW_{i-1},&i\text{ even},\\
				W_{i-1}U_i,&i\text{ odd}
			\end{cases}
		\end{equation}
		for $i\geq 2$.
		So $W_{t-1}=U_{\mathrm{glue}}$.
		
		For $W_{i}$, the old region $Q_1\cup\cdots\cup Q_{i}$ and the new patch $R_{i}=Q_{i}\cup Q_{i+1}$ overlap on $Q_i$.
		Let $\epsilon_i$ be the multiplicative error of $W_i$.
		Then $\epsilon_1=\epsilon'$.
		As $M\geq 1+\log k$, by setting $\kappa_\delta$ sufficiently large, $2^{|Q_i|}\geq 2^{l^\delta}\geq 2^{\kappa_\delta M}\geq 2k^2$.
		Hence by Fact~\ref{fact:app-fd-gluing},
		\begin{equation}
			1+\epsilon_i\leq (1+\epsilon_{i-1})(1+\epsilon')(1+\cO(k^22^{-|Q_i|})).
		\end{equation}
		Iterating this inequality gives
		\begin{equation}\begin{aligned}
				1+\epsilon_{t-1}
				&\leq (1+\epsilon')^{t-1}(1+\cO(k^22^{-l^\delta}))^{t-2}\\
				&\leq \exp\!\left(t\epsilon'+\cO(tk^22^{-l^\delta})\right)\\
				&\leq \exp\!\left(t\epsilon'+\cO(tk^22^{-\kappa_\delta M})\right).
		\end{aligned}\end{equation}
		Since $t\leq n$ and $M=\log(n)+\log(k)+\log(\epsilon^{-1})$, setting $\kappa_\delta$ sufficiently large makes $\cO(tk^22^{-\kappa_\delta M})\leq\epsilon/8$. Together with $t\epsilon'\leq\epsilon/8$, this yields
		\begin{equation}
			\epsilon_{t-1} \leq \exp(\epsilon/4)-1 \leq \epsilon,
		\end{equation}
		where the last inequality holds for $0<\epsilon\leq 1$. Hence $U_{\mathrm{glue}}$ is a multiplicative $\epsilon$-approximate unitary $k$-design on all $n$ qubits.

		It remains to analyze the circuit depth of $U_{\mathrm{glue}}$.
		The factors within either product have disjoint supports and may therefore be executed in parallel.
		By Theorem~\ref{thm:fd-small-design}, each single unitary $U_i$ can be computed in depth 
		\begin{equation}\begin{aligned}
				\cO_\delta(M^{1/\delta}(k+\frac{\log(8t\epsilon^{-1})}{M})\log^2 k)
				&=\cO_\delta(M^{1/\delta}(k+\frac{3+\log(n)+\log(\epsilon^{-1})}{M})\log^2 k)\\
				&=\cO_\delta(M^{1/\delta}k\log^2 k).
		\end{aligned}\end{equation}
		The two brickwork layers therefore have total depth $\cO_\delta(M^{1/\delta}k\log^2 k)$, which is Eq.~\eqref{eq:fd-large-design-depth}.
	\end{proof}
	
	\begin{remark}[Very small errors]\label{rmk:very-small-errors}
		While our low-depth constructions for multiplicative errors $\epsilon \geq 2^{-\eta nk}$ suffice for most practical applications, we can also maintain the depth scaling of Eq.~\eqref{eq:fd-large-design-depth} in the highly precise $\epsilon < 2^{-\eta nk}$ regime. This is achieved via the exactification lemma (Lemma~\ref{lem:exactification}), which converts an approximate design with a sufficiently small error into an exact one. By generating an initial approximate design with error $\epsilon = 2^{-\eta nk}$ for a sufficiently large constant $\eta$, utilizing Eq.~\eqref{eq:gap_repetition_hadamard-diagonal} for both design orders $k$ and $2k$, we satisfy the precondition of Lemma~\ref{lem:exactification}. 
		This yields an exact unitary design, and therefore, a valid multiplicative design for any arbitrarily small $\epsilon > 0$. 
		However, this comes with a computational caveat: sampling from the exactified distribution may be inefficient in practice, as the classical reweighting factors introduced by Lemma~\ref{lem:exactification} are not currently known to be efficiently computable.
	\end{remark}
	
	\section{All-to-all implementation}
	\label{app:a2a}
	We now present the low-depth construction for all-to-all architectures.
	This construction mirrors the finite-dimensional case (Appendix~\ref{app:finite-dimensional}), differing only in the hash functions used to mix the input during the first two steps. 
	To further compress circuit depth, we construct these hash functions using a good binary linear code whose encoder admits a linear-size, logarithmic-depth circuit. 
	We summarize the properties of the resulting diagonal block in the following theorem.
	
	\begin{theorem}[Low-depth diagonal blocks with all-to-all connectivity]
		\label{thm:a2a-ensemble}
		There exist constants $A,c,C>0$ such that, for every $k\geq 2$ and $m\geq A\log k$, there is an ensemble $\nu_{\mathrm{a2a},m}$ of diagonal unitaries on $m$ qubits satisfying
		\begin{equation}
			\label{eq:a2a-distance}
			\left\|M_r(\nu_{\mathrm{a2a},m})-P_{\mathrm Z}^{(r)}\right\|_{\mathrm{op}}
			\leq C\exp\!\left[-\frac{cm}{\log k}\right]
		\end{equation}
		simultaneously for every $1\leq r\leq2k$.
		Every sample has an ancilla-free all-to-all circuit of depth
		\begin{equation}
			\label{eq:a2a-diagonal-depth}
			\cO\!\left(\log m+\log k\right).
		\end{equation}
	\end{theorem}
	
	Substituting this ensemble into Eq.~\eqref{eq:app-implemented-HD-factorization} and repeating the process sequentially suppresses the TPE error. Gluing these local designs together then yields a low-depth global design, as detailed in Sec.~\ref{app:a2a-repetition}.
	
	The section is organized as follows.
	Sec.~\ref{app:a2a-construction} defines the code-based diagonal ensemble, Sec.~\ref{app:a2a-distance} proves its moment estimate, Sec.~\ref{app:a2a-depth} gives its ancilla-free all-to-all implementation, and Sec.~\ref{app:a2a-repetition} performs repetition and gluing to obtain the global approximate design.

	\subsection{Code-based ensemble construction}
	\label{app:a2a-construction}
	
	We first employ a binary linear code from Ref.~\cite{Spielman1996LinearTimeCodes}, whose properties are summarized below.
	\begin{fact}[Good binary linear code,~{\cite{Spielman1996LinearTimeCodes}}]
		\label{fact:app-a2a-code}
		There exist two constants $\tau\geq 1$ and $0<\Delta\leq 1$, such that for every $u\geq 1$, there exists a binary linear code $C_u$ where
		\begin{equation}
			C_u:\mathbb{F}_2^{u}\longrightarrow\mathbb{F}_2^{N_u}
		\end{equation}
		for some $u\leq N_u\leq \tau u$ with code distance at least $\Delta N_u$.
		In other words,
		\begin{equation}
			\operatorname{wt}(C_u(z))\geq \Delta N_u
		\end{equation}
		for every nonzero $z\in\{0,1\}^{u}$.
		Moreover, the encoding algorithm $C_u$ has a linear-size, logarithmic-depth, bounded-fan-in, bounded-fan-out circuit with only classical XOR gates.
	\end{fact}
	
	\begin{remark}
		The construction initially supports a sequence of input lengths $u$ that increase by a factor of two.
		For an intermediate input length, one may use the next supported length and pad the message by hardwired zeros.
		This gives the formulation in Fact~\ref{fact:app-a2a-code} for every $u$.
	\end{remark}
	
	We construct a random hash function $H: \mathbb{F}_2^u \to \mathbb{F}_2^v$ as follows. For each coordinate $j \in [v]$, independently and uniformly sample $I_j \in [N_u]$ and $\epsilon_j \in \mathbb{F}_2$, and set
	\begin{equation}
		\label{eq:a2a-hash}
		[H(z)]_j:=\epsilon_{j}[C_{u}(z)]_{I_{j}}.
	\end{equation}
	Set $\rho:=\Delta/2\leq \frac{1}{2}$, then for every $z\neq 0$,
	\begin{equation}
		\rho \leq \Pr_{I_j,\epsilon_j}\bigl([H(z)]_j=1\bigr)
		=\frac{\operatorname{wt}(C_u(z))}{2N_u}
		\leq \frac{1}{2}.
		\label{eq:app-a2a-one-coordinate-balance}
	\end{equation}
	Thus, each output coordinate takes either bit value with probability at least $\rho$.
	Furthermore, the output coordinates of $H(z)$ are mutually independent.

	As in Appendix~\ref{app:finite-dimensional}, we partition the $m$ qubits into a sufficiently large constant number $r_0$ of nonempty
	registers $Z_1,\ldots,Z_{r_0}$ such that $X:=Z_1$, $Y:=Z_2$, $|X|=|Y|=s=\Theta(m)$ and $|Z_a|\leq s$ for every $a\in[r_0]$.
	A sample from $\nu_{\mathrm{a2a},m}$ is obtained from the following five steps.
	\begin{enumerate}
		\item \emph{Concentration.}
		For each $a\neq2$ and each $j\in[s]$, independently and uniformly sample
		\begin{equation}
			I_{a,j}\in [N_{|Z_a|}],
			\qquad
			\epsilon_{a,j}\in\mathbb{F}_2,
		\end{equation}
		and define the hash functions
		\begin{equation}
			[L_a(z)]_j:=\epsilon_{a,j}[C_{|Z_a|}(z)]_{I_{a,j}}.
			\label{eq:app-a2a-concentration-map}
		\end{equation}
		Apply
		\begin{equation}
			Y\gets Y\oplus\bigoplus_{a\neq2}L_a(Z_a).
		\end{equation}
		
		\item \emph{Hashing.} Independently and uniformly sample $I'_j\in[N_s]$ and
		$\epsilon'_j\in \mathbb F_2$ for $j\in[s]$, and set
		\begin{equation}
			[M(z)]_j:=\epsilon'_j[C_s(z)]_{I'_j}.
			\label{eq:app-a2a-hashing-map}
		\end{equation}
		Sample a uniform shift $\beta\in\mathbb F_2^s$.
		Apply
		\begin{equation}
			X\gets X\oplus M(Y)\oplus\beta.
		\end{equation}
		\item Sample the same randomness $\Pi,\Gamma$ as in Step~3 of Sec.~\ref{app:finite-dimensional-construction}, with $w=\lceil C_w\log k\rceil$, and apply
		\begin{equation}
			Y_j\gets Y_j\oplus
			\bigoplus_{l=1}^w\gamma_{j,l}X_{\pi_l(j)}X_{\sigma_l(j)}.
			\label{eq:app-a2a-quadratic-shear}
		\end{equation}
		\item Apply independent uniform one-qubit phases to the qubits of $Y$.
		\item Reverse the first three steps.
	\end{enumerate}
	
	\subsection{Approximation error analysis}
	\label{app:a2a-distance}
	This subsection proves the distance claimed in Eq.~\eqref{eq:a2a-distance}. 
	
	Use the same notation defined in Sec.~\ref{app:finite-dimensional-distance} to represent the intermediate values, the randomness, and the signed multiplicity profile.
	The definitions of $y$ and $x$ are modified as
	\begin{equation}\begin{aligned}
			\label{eq:a2a-hashxy}
			y(\xi)&:=\xi_2\oplus\bigoplus_{a\neq 2}L_a(\xi_a),\\
			x(\xi)&:=\xi_1\oplus M(y(\xi))\oplus \beta
	\end{aligned}\end{equation}
	according to the new code-based implementation in Sec.~\ref{app:a2a-construction}.
	The action of the ensemble $\nu_{\mathrm{a2a},m}$ is again
	\begin{equation}
		\ket{\xi} \longmapsto \exp\!\left(
		\mathrm i\sum_{j=1}^{s}\theta_jh_j(\xi)
		\right)\ket{\xi}
	\end{equation}
	with
	\begin{equation}
		h_j(\xi):=y_j(\xi)\oplus\bigoplus_{l=1}^{w}\gamma_{j,l}x_{\pi_l(j)}(\xi)x_{\sigma_l(j)}(\xi).
	\end{equation}
	
	The randomness for the first two steps is now 
	\begin{equation}
		\mathsf{A}:=(\{I_{a,j},\epsilon_{a,j}\}_{a\neq 2,j\in [s]},\{(I'_j,\epsilon'_j)_{j\in[s]}\},\beta).
	\end{equation}
	
	The distance reduction in Lemma~\ref{lem:app-fd-distance-probability} uses only the diagonal action and the independence of the phase angles, but not the concrete function $h$, and therefore applies without change.
	This gives
	\begin{equation}
		\left\|M_r(\nu_{\mathrm{a2a},m})-P_{\mathrm Z}^{(r)}\right\|_{\mathrm{op}}
		=\max_{c\in\mathcal C_r}
		\Pr_{\mathsf A,\Pi,\Gamma}
		\left[S_j(c)=0\ \text{for every }j\in[s]\right].
		\label{eq:app-a2a-distance-as-probability}
	\end{equation}
	
	It remains to establish the four-label randomness property used by the rest
	of the argument in Appendix~\ref{app:finite-dimensional}.
	
	We first give the balanced behavior of $x$ on a difference space spanned by four labels, which replaces Lemma~\ref{lem:app-fd-affine-restriction}.
	\begin{lemma}[Balanced affine hashing on four-label spans]
		\label{lem:app-a2a-affine-restriction}
		Let $T\subseteq\{0,1\}^m$ be a non-empty set of at most four labels.
		Fix $\xi^{(0)}\in T$, and define the difference space
		\begin{equation}
			V:=\operatorname{span}_{\mathbb F_2}\{\xi-\xi^{(0)}:\xi\in T\},\quad
			W=\xi^{(0)}+V.
		\end{equation}
		With probability at least $1-7\cdot e^{-\rho s}$ over randomness of the maps $\{L_a\}_{a\neq2}$, the restricted map $y|_W$ is injective on the affine subspace $W$ for the function $y$ defined in Eq.~\eqref{eq:a2a-hashxy}.
		
		Conditional on this event, the functions $x_1|_W,\ldots,x_s|_W$ are independent random affine Boolean functions.
		The base values $x_1(\xi^{(0)}),\ldots,x_s(\xi^{(0)})$ are independent uniform bits conditional on their linear parts.
		For every $j\in [s]$ and $\xi\in W$ with $\xi\neq\xi^{(0)}$,
		\begin{equation}
			\rho\leq \Pr_{M,\beta}[x_j(\xi^{(0)})\neq x_j(\xi)]\leq 1-\rho.
		\end{equation}
	\end{lemma}
	\begin{proof}
		The restriction $y|_W$ is injective if and only if $y(v)\neq 0$ for every nonzero $v\in V$.
		For fixed $v\neq 0$, if $v_a=0$ for every $a\neq 2$, then $y(v)=v_2\neq 0$.
		Otherwise, choose $a\neq2$ with $v_a\neq0$ and condition on all maps $L_b$ with $b\neq a$.
		By Eq.~\eqref{eq:app-a2a-one-coordinate-balance} and the independence,
		\begin{equation}
			\Pr_{L_a}[y(v)=0]\leq (1-\rho)^s\leq e^{-\rho s}.
		\end{equation}
		A union bound over at most $7$ nonzero labels in $V$ proves the injectivity claim.
		
		Condition on a choice of the maps $L_a$ for which $y|_W$ is injective.
		Since $y$ and $M$ are linear, $x|_W$ is affine.
		Moreover, the triples $(I'_j,\epsilon'_j,\beta_j)$ are independent across $j$, so the coordinate functions $x_1|_W,\ldots,x_s|_W$ are independent.
		Conditional on $(I'_j,\epsilon'_j)$, the independent uniform bit $\beta_j$ makes $x_j(\xi^{(0)})$ uniform.
		
		Consider the value difference between $\xi^{(0)}$ and $\xi=\xi^{(0)}+v$,
		\begin{equation}
			[x(\xi)\oplus x(\xi^{(0)})]_j=[v_1]_j\oplus \epsilon'_j[C_s(y(v))]_{I_j'}.
		\end{equation}
		By the injectivity of $y$, $y(v)\neq 0$ for every nonzero $v\in V$.
		Thus the second term has probability between $\rho$ and $\frac{1}{2}$ to be $1$ by Eq.~\eqref{eq:app-a2a-one-coordinate-balance}.
		After the first term $[v_1]_j$ is fixed, the probability of a nonzero difference lies between $\rho$ and $1-\rho$.
	\end{proof}
	
	The proof of Lemma~\ref{lem:app-fd-quadratic-parity} only used uniformity to obtain a constant lower bound.  
	The balance above gives the corresponding bound directly, as the following lemma shows.
	\begin{lemma}[Quadratic parity of balanced affine functions]
		\label{lem:app-a2a-balanced-quadratic-parity}
		Let $W$ be an affine space over $\mathbb F_2$, fix a base point $z^{(0)}\in W$,
		and let $f,g:W\to\mathbb F_2$ be independent random affine functions.
		Suppose that, conditional on their linear parts, the constant terms of $f$ and $g$ are independent uniform bits and that, for every nonzero $v$ with $z=z^{(0)}+v\in W$,
		\begin{equation}
			\label{eq:app-a2a-balanced-affine-assumption}
			\rho\leq\Pr_f[f(z^{(0)})\neq f(z)]\leq 1-\rho,\quad \rho\leq \Pr_g[g(z^{(0)})\neq g(z)]\leq1-\rho
		\end{equation}
		for some $0<\rho\leq \frac{1}{2}$.
		Then for every nonempty set $T\subseteq W$ of at most four labels,
		\begin{equation}
			\Pr_{f,g}\!\left[\bigoplus_{\xi\in T}f(\xi)g(\xi)=1\right]
			\geq \rho^2.
			\label{eq:app-a2a-balanced-quadratic-parity}
		\end{equation}
	\end{lemma}
	\begin{proof}
		Write
		\begin{equation}
			P_T(f,g):=\bigoplus_{\xi\in T}f(\xi)g(\xi),\quad f(\xi)=l_f(\xi)+\alpha,\quad g(\xi)=l_g(\xi)+\beta,
		\end{equation}
		where $l_f$ and $l_g$ are both linear functions.
		Then
		\begin{equation}
			P_T(f,g)=\left(\bigoplus_{\xi\in T}l_f(\xi)l_g(\xi)\right)
			\oplus \left(\beta\bigoplus_{\xi\in T}l_f(\xi)\right)
			\oplus \left(\alpha\bigoplus_{\xi\in T}l_g(\xi)\right)
			\oplus (\alpha\beta\cdot [|T|\bmod 2]).
		\end{equation}
		Denote $A=\bigoplus_{\xi\in T}l_f(\xi)$, $B=\bigoplus_{\xi\in T}l_g(\xi)$ and $C=\bigoplus_{\xi\in T}l_f(\xi)l_g(\xi)$, then
		\begin{equation}
			\label{eq:a2a-quad-expansion}
			P_T(f,g)=C\oplus \beta A\oplus \alpha B \oplus \alpha\beta\cdot[|T|\bmod 2].
		\end{equation}
		
		For the case where $|T|$ is odd, the result $P_T(f,g)$ has the last term $\alpha\beta$.
		Write
		\begin{equation}
			P_T(f,g)=C\oplus (A\oplus \alpha)(B\oplus \beta)\oplus AB.
		\end{equation}
		Fix the linear parts $l_f,l_g$, so that $A,B,C$ are all determined.
		As $\alpha$ and $\beta$ are independent and uniformly random, 
		\begin{equation}
			\Pr_{\alpha,\beta}[(A\oplus \alpha)(B\oplus \beta)=1]=\frac{1}{4}.
		\end{equation}
		Therefore, $P_T(f,g)$ has probability at least $\frac{1}{4}\geq \rho^2$ to be $1$.
		
		It remains to consider sets $T$ with even size, for which the term $\alpha\beta$ in Eq.~\eqref{eq:a2a-quad-expansion} vanishes.
		If $T$ is affinely independent, fix a label $\xi^{(0)}\in T$.
		Then $B$ can be expressed as
		\begin{equation}\begin{aligned}
				B=\bigoplus_{\xi\in T} l_g(\xi)
				&=\bigoplus_{\xi\in T,\xi\neq \xi^{(0)}}(l_g(\xi)-l_g(\xi^{(0)}))\\
				&=\bigoplus_{\xi\in T}(g(\xi)-g(\xi^{(0)}))\\
				&=g\left(\bigoplus_{\xi\in T,\xi\neq \xi^{(0)}}\xi\right)-g(\xi^{(0)}),
		\end{aligned}\end{equation}
		as $|T|$ is even.
		Since $T$ is affinely independent, $\bigoplus_{\xi\in T,\xi\neq \xi^{(0)}}\xi\neq \xi^{(0)}$.
		By Eq.~\eqref{eq:app-a2a-balanced-affine-assumption}, $\Pr_{l_g}[B=1]\geq \rho$.
		Condition on this event, for every fixed $\beta$ and $l_f$, the term $\alpha B$ is uniformly distributed over $\mathbb F_2$ as $\alpha$ varies.
		Therefore, $\Pr_{f,g}[P_T(f,g)=1]\geq \rho\cdot \frac{1}{2}\geq \rho^2$ in this case.
		
		Finally, consider the case where $T$ is affinely dependent.
		In this case, $|T|=4$ and the four labels form an affine plane.
		Suppose $T=\{\xi^{(0)},\xi^{(0)}+u,\xi^{(0)}+v,\xi^{(0)}+u+v\}$.
		By Eq.~\eqref{eq:app-a2a-balanced-affine-assumption} and the affine property of $f$,  the inequalities $f(\xi^{(0)})\neq f(\xi^{(0)}+u)$ and $f(\xi^{(0)}+v)\neq f(\xi^{(0)}+v+u)$ hold with probability at least $\rho$ over the randomness of $f$.
		Under this event,
		\begin{equation}\begin{aligned}
				\label{eq:a2a-quad-affine-dependent}
				P_T(f,g)&=\bigoplus_{\xi\in T}f(\xi)g(\xi)\\
				&=g\left(\xi^{(0)}+f(\xi^{(0)}+u)u\right)+g\left(\xi^{(0)}+v+f(\xi^{(0)}+v+u)u\right).
		\end{aligned}\end{equation}
		The two evaluation points are always distinct as their difference is either $v$ or $u+v$.
		Another Eq.~\eqref{eq:app-a2a-balanced-affine-assumption} on $g$ together with the affine property shows that Eq.~\eqref{eq:a2a-quad-affine-dependent} has probability at least $\rho$ to be $1$, regardless of the values of $f(\xi^{(0)}+u)$ and $f(\xi^{(0)}+v+u)$.
		Therefore, $\Pr_{f,g}[P_T(f,g)=1]\geq \rho^2$ in this case.
	\end{proof}
	
	Fix $1\leq r\leq2k$ and a signed multiplicity profile $c\in\mathcal C_r$ (Eq.~\eqref{eq:app-fd-profile-class}). Let $\Omega:=\operatorname{supp}(c)$, $R:=|\Omega|\leq2r\leq4k$, and 
	\begin{equation}
		\mathcal F_4(\Omega)
		:=\{\varnothing\neq S\subseteq\Omega:|S|\leq4\}.
	\end{equation}
	Using the newly constructed functions $x_p$, we define the quadratic parity
	\begin{equation}
		Q_T(p,q):=
		\bigoplus_{\xi\in T}x_p(\xi)x_q(\xi)
		\in\mathbb F_2
	\end{equation}
	and the detecting pairs
	\begin{equation}
		\mathcal D_T
		:=
		\left\{
		(p,q)\in[s]^2:p\neq q,\ Q_T(p,q)=1
		\right\}
	\end{equation}
	analogously to Eqs.~\eqref{eq:app-fd-four-subsets}, \eqref{eq:app-fd-quadratic-parity}, and \eqref{eq:app-fd-detecting-pair-set}. 
	The following lemma provides the code-based counterpart to Lemma~\ref{lem:app-fd-dense-detecting-pairs}.

	\begin{lemma}[Dense detecting pairs]
		\label{lem:app-a2a-dense-detecting-pairs}
		There are constants $A_0,c_0,C_0,\zeta>0$ such that, if $s\geq A_0\log k$, then with probability at least $1-C_0e^{-c_0s}$ over
		$\mathsf A$,
		\begin{equation}
			|\mathcal D_T|\geq\zeta s(s-1)
			\label{eq:app-a2a-detecting-pair-density}
		\end{equation}
		simultaneously for every $T\in\mathcal F_4(\Omega)$.
	\end{lemma}
	\begin{proof}
		By Eq.~\eqref{eq:app-fd-number-four-subsets}, there are at most
		$(4k+1)^4$ relevant sets $T$.  For each such $T$, fix
		$\xi_T^{(0)}\in T$ and write
		\begin{equation}
			V_T:=\operatorname{span}_{\mathbb F_2}
			\{\xi-\xi_T^{(0)}:\xi\in T\},
			\qquad W_T:=\xi_T^{(0)}+V_T.
		\end{equation}
		Lemma~\ref{lem:app-a2a-affine-restriction}
		and a union bound show that $y$ is injective on every associated affine span $W_T$,
		except with probability at most
		\begin{equation}
			7(4k+1)^4e^{-\rho s}.
			\label{eq:app-a2a-simultaneous-injectivity-failure}
		\end{equation}
		Condition on simultaneous injectivity.  For every fixed $T$ and
		$p\neq q$, Lemmas~\ref{lem:app-a2a-affine-restriction}
		and~\ref{lem:app-a2a-balanced-quadratic-parity} imply
		\begin{equation}
			\Pr[Q_T(p,q)=1]\geq \rho^2.
			\label{eq:app-a2a-one-pair-detection}
		\end{equation}
		Define
		\begin{equation}
			Z_T:=\frac{1}{s(s-1)}
			\sum_{p\neq q}\mathbbm1[Q_T(p,q)=1].
		\end{equation}
		Then $\mathbb E[Z_T]\geq \rho^2$.  
		Replacing one of the independent affine coordinate functions $x_p|_{W_T}$ changes $Z_T$ by at most $2/s$.
		McDiarmid's inequality therefore gives
		\begin{equation}
			\Pr[Z_T<\rho^2/2]\leq
			\exp(-\rho^4s/8).
			\label{eq:app-a2a-pair-density-concentration}
		\end{equation}
		A second union bound over $T\in\mathcal F_4(\Omega)$ together with Eq.~\eqref{eq:app-a2a-simultaneous-injectivity-failure} gives total failure probability
		\begin{equation}
			7(4k+1)^4 e^{-\rho s}
			+(4k+1)^4e^{-\rho^4 s/8},
			\label{eq:app-a2a-affine-density-failure}
		\end{equation}
		which proves the lemma with $\zeta=\rho^2/2$,
		for sufficiently large $A_0$.
	\end{proof}
	
	The remainder of the argument is the same as in
	Sec.~\ref{app:finite-dimensional-distance}.  
	Indeed, the proof of Lemma~\ref{lem:app-fd-shared-permutation-coverage} uses the value $1/8$ in Eq.~\eqref{eq:app-fd-detecting-pair-density} only
	as a positive lower bound on the density of $\mathcal D_T$.  
	Replacing it by $\zeta$ shows that, for sufficiently large $C_w$, with probability at least
	\begin{equation}
		1-C_0e^{-c_0s}-e^{-c_1s/w}
		\label{eq:app-a2a-structural-event}
	\end{equation}
	over $(\mathsf A,\Pi)$, there is a set $G_c\subseteq[s]$ of size at least $s/2$ such that every $j\in G_c$ and every $T\in\mathcal F_4(\Omega)$, at least one $l\in[w]$ satisfies
	\begin{equation}
		Q_T(\pi_l(j),\sigma_l(j))=1.
		\label{eq:app-a2a-simultaneous-permutation-detection}
	\end{equation}
	
	On this event, the character factorization from
	Eq.~\eqref{eq:app-fd-character-factorization} has a zero factor for every nonempty $T\subseteq\Omega$ with $|T|\leq4$. 
	Hence the character cancellation in Eq.~\eqref{eq:app-fd-character-cancellation}, and therefore the second- and fourth-moment identities in Eq.~\eqref{eq:app-fd-second-fourth-moments}, hold without change. 
	The Cauchy--Schwarz estimate in Eq.~\eqref{eq:app-fd-single-target-detection} and the independence of the private vectors $\gamma_j$ now give, uniformly in $c\in\mathcal C_r$,
	\begin{equation}
		\Pr_{\mathsf A,\Pi,\Gamma}
		\left[S_j(c)=0\ \text{for every }j\in[s]\right]
		\leq C_0e^{-c_0s}+e^{-c_1s/w}+(2/3)^{s/2}.
		\label{eq:app-a2a-profile-failure-bound}
	\end{equation}
	Finally, plug in $s=\Theta(m)$ and $w=\Theta(\log k)$. 
	Substituting Eq.~\eqref{eq:app-a2a-profile-failure-bound} into Eq.~\eqref{eq:app-a2a-distance-as-probability} and adjusting constants completes the proof of Eq.~\eqref{eq:a2a-distance} for every $1\leq r\leq2k$.
	
	\subsection{Circuit implementation}
	\label{app:a2a-depth}
	This subsection implements the constructed ensemble with the resources claimed in Theorem~\ref{thm:a2a-ensemble}.
	
	We begin by implementing the code introduced in Fact~\ref{fact:app-a2a-code} catalytically.
	The following lemma converts any bounded-fan-in bounded-fan-out XOR circuit into a catalytic target-add quantum circuit with the same asymptotic size and depth.
	
	\begin{lemma}[Catalytic implementation of classical XOR circuits]
		\label{lem:a2a-catalytic}
		Let $C:\mathbb F_2^n\to\mathbb F_2^q$ be a circuit containing constant fan-in, constant fan-out XOR gates.
		Suppose $C$ has $g$ gates of depth $d$.
		Let $X$ and $Y$ contain $n$ and $q$ bits, respectively.
		Then the target addition
		\begin{equation}
			\ket{x}_X\ket{y}_Y\ket{a}_A\longmapsto \ket{x}_X\ket{y+C(x)}_Y\ket{a}_A
		\end{equation}
		of $C$ can be implemented catalytically with an $\cO(q+g)$-sized catalytic workspace $A$ in $\cO(d)$ depth on an all-to-all architecture.
	\end{lemma}
	\begin{proof}
		First, reduce the fan-out of the XOR circuit to $2$.
		For each XOR gate wire with fan-out larger than $2$, replace the wire with a binary fan-out tree.
		Each tree node is a copy gate with its parent as the input.
		The gates that take this wire as input are connected with the leaves of the tree.
		This process only increases the depth $d$ and the gate count $g$ by a constant multiple as the circuit has constant fan-out.
		
		Allocate an ancilla qubit for each gate wire, including the output wires, on $A$.
		We first construct a quantum circuit $\Phi$ which can output a correct result when $A$ is clean.
		For each gate $T\gets E\oplus F$ in $C$, $\Phi$ applies
		\begin{equation}
			\ket{e}_E\ket{f}_F\ket{t}_T\mapsto \ket{e}_E\ket{f}_F\ket{t+e+f}_T.
		\end{equation}
		Each layer of gates can be implemented in $\cO(1)$ depth using K\H{o}nig's line-coloring theorem.
		Then
		\begin{equation}
			\Phi:\ket{x}_X\ket{y}_Y\ket{0,0}_A\longmapsto \ket{x}_X\ket{y}_Y\ket{C(x),P(x)}_A
		\end{equation}
		for some linear map $P$.
		
		Using the uncomputing trick, construct a circuit $\tilde\Phi$, which performs
		\begin{equation}
			\Phi; \quad \mathrm{COPY}_{A \rightarrow Y}: \ket{x}_X\ket{y}_Y\ket{c,a'}_A\longmapsto \ket{x}_X\ket{y+c}_Y\ket{c,a'}_A; \quad \Phi^{-1}
		\end{equation}
		sequentially. 
		Then
		\begin{equation}
			\tilde\Phi:\ket{x}_X\ket{y}_Y\ket{0}_A\longmapsto \ket{x}_X\ket{y+C(x)}_Y\ket{0}_A
		\end{equation}
		can restore the ancilla qubits.
		
		Moreover, because $\mathrm{COPY}_{A \rightarrow Y}$ preserves the $A$ register and $\tilde\Phi$ consists entirely of CNOT gates (which implement the XOR target additions), the overall action of $\tilde\Phi$ on the computational basis is linear and leaves $A$ invariant. 
		Therefore, given a dirty ancillary space $A$, $\tilde\Phi$ acts as
		\begin{equation}
			\tilde\Phi:\ket{x}_X\ket{y}_Y\ket{a}_A\longmapsto \ket{x}_X\ket{y+C(x)+Q(a)}_Y\ket{a}_A
		\end{equation}
		for some linear map $Q$.
		
		Define $\Psi$ by deleting from $\tilde\Phi$ every CNOT gate whose control lies in $X$.
		Then the action of $\Psi$ is
		\begin{equation}
			\Psi:\ket{x}_X\ket{y}_Y\ket{a}_A\longmapsto \ket{x}_X\ket{y+Q(a)}_Y\ket{a}_A.
		\end{equation}
		Therefore, we can implement the target addition of $C$ catalytically with
		\begin{equation}
			\Psi \circ\tilde\Phi 
		\end{equation}
		which has circuit depth $\cO(d)$.
	\end{proof}
	
	\begin{proof}[Proof of Theorem~\ref{thm:a2a-ensemble}]
		The previous subsection proves Eq.~\eqref{eq:a2a-distance}.
		It remains to show that the ensemble $\nu_{\mathrm{a2a},m}$ can be implemented within the claimed depth using an ancilla-free circuit.
		
		Choose the number $r_0$ of data registers sufficiently large that the $\cO(s)$ catalytic qubits required below can always be borrowed from inactive registers.
		
		We first implement the target addition of hash functions $L_a$ and $M$, which have the form Eq.~\eqref{eq:a2a-hash}, in Steps 1 and 2.
		For a hash function $H:\mathbb{F}_2^u\to \mathbb{F}_2^v$ with $[H(z)]_j:=\epsilon_j[C_u(z)]_{I_j}$, picking a dirty intermediate register $c$ of size $N_u$, the target addition 
		\begin{equation}
			\label{eq:a2a-hash-encoding}
			\ket{z,c,\omega}\longmapsto\ket{z,c+C_u(z),\omega}
		\end{equation}
		of the encoder $C_u$ of Fact~\ref{fact:app-a2a-code} can be implemented by Lemma~\ref{lem:a2a-catalytic} 
		with depth $\cO(d)=\cO(\log u)$ and $\cO(q+g)=\cO(u)$ catalytic qubits.
		Given the codeword $C$, each bit $[C]_i$ is added into sites $j$ with $I_j=i$ and $\epsilon_j=1$ of the output.
		We then implement the target addition
		\begin{equation}
			\label{eq:a2a-hash-distributing}
			\ket{c,(h_j)_{j\in[v]},\omega}\longmapsto\ket{c,(h_j+\epsilon_j[c]_{I_j})_{j\in[v]},\omega}
		\end{equation}
		computing the output of $H$.
		Since there are at most $v$ sites, with a binary fan-out tree, there exists a classical XOR circuit with depth $\cO(\log v)$ and size $\cO(v)$ to compute $(\epsilon_j [C]_{I_j})_{j\in[v]}$, which has fan-in and fan-out two.
		Lemma~\ref{lem:a2a-catalytic} gives a catalytic implementation of Eq.~\eqref{eq:a2a-hash-distributing} with depth $\cO(\log v)$ and catalytic size $\cO(v)$.
		Applying the compiler in Lemma~\ref{lem:catalytic-operations} to combine Eq.~\eqref{eq:a2a-hash-encoding} and Eq.~\eqref{eq:a2a-hash-distributing} gives the target addition for the hash function in $\cO(\log u+\log v)=\cO(\log m)$ depth with $\cO(u+v)=\cO(m)$ catalytic qubits.
		There are only a constant number of such hash functions in Steps 1 and 2.
		One can perform them sequentially.
		The hardwired shift $\beta$ in Step~2 is implemented by parallel Pauli-$X$ gates in depth one.
		
		In Step~3, for a fixed $l\in[w]$, first implement the loop terms with $\pi_l(j)=\sigma_l(j)$ as the CNOT update $Y_j\gets Y_j\oplus\gamma_{j,l}X_{\pi_l(j)}$, which are disjoint because $\pi_l$ is a permutation.
		After removing the loops, the multigraph with edges $\{\pi_l(j),\sigma_l(j)\}_j$ has maximum degree two and hence splits into at most three matchings.
		Within one matching, the corresponding $\mathrm{CC}X$ gates in Eq.~\eqref{eq:app-a2a-quadratic-shear} act on disjoint triples and have constant all-to-all depth.
		Thus the entire quadratic mixing has depth $\cO(w)=\cO(\log k)$.
		
		The phase layer in Step~4 has depth one, and the inversions in Step~5 have the same depths as the corresponding forward operations.
		The total circuit has depth $\cO(\log m+\log k)$, which completes the proof.
	\end{proof}

	\subsection{Repetition and gluing to form designs}
	\label{app:a2a-repetition}
	Theorem~\ref{thm:a2a-ensemble} supplies a shallow approximation to an ideal diagonal unitary.
	As in Sec.~\ref{app:finite-dimensional-repetition}, three independent diagonal samples are first inserted into a Hadamard--diagonal block, and independent repetitions amplify its moment contraction.
	The small-system regime is again handled by the random-Pauli-rotation gap in Fact~\ref{fact:app-fd-brickwork-gap}.
	
	\begin{theorem}[$m$-qubit approximate designs without ancillae with all-to-all connectivity]
		\label{thm:a2a-small-design}
		For every $m\geq1$, $k\geq2$, and $0<\epsilon\leq1$, there exists a multiplicative $\epsilon$-approximate unitary $k$-design on $m$ qubits with an ancilla-free all-to-all circuit of depth
		\begin{equation}
			\label{eq:app-a2a-m-qubit-design-depth}
			\cO\!\left(
			(\log m+\log k)\log k
			\left[k+\frac{\log(\epsilon^{-1})}{m}\right]
			\right).
		\end{equation}
	\end{theorem}
	\begin{proof}
		First suppose that $m\geq A\log k$, where Theorem~\ref{thm:a2a-ensemble} applies.
		Let $\mu_{\mathrm{a2a},m}$ be the ensemble
		\begin{equation}
			U=D_3H^{\otimes m}D_2H^{\otimes m}D_1,
		\end{equation}
		where $D_1,D_2,D_3$ are sampled independently from $\nu_{\mathrm{a2a},m}$.
		Lemma~\ref{lem:app-HD-perturbation} and Eq.~\eqref{eq:a2a-distance} give
		\begin{equation}\begin{aligned}
				\lambda_k(\mu_{\mathrm{a2a},m})
				&\leq \frac{9k^4}{2^m}
				+3C\exp\!\left[-\frac{cm}{\log k}\right]\\
				&\leq \exp\!\left[-\frac{am}{\log k}\right]
				\label{eq:app-a2a-HD-block-gap}
		\end{aligned}\end{equation}
		for a constant $a>0$ when $A$ is sufficiently large.
		By the amplification bound in Eq.~\eqref{eq:app-gap-amplification},
		\begin{equation}
			\lambda_k(\mu_{\mathrm{a2a},m}^{*T})
			\leq
			\exp\!\left[-\frac{amT}{\log k}\right].
		\end{equation}
		Fact~\ref{fact:TPE-multiplicative-comparison} therefore bounds the multiplicative error by
		\begin{equation}
			\epsilon_{\mathrm{mult}}\leq 2^{2mk}\lambda_k(\mu_{\mathrm{a2a},m}^{*T})
			\leq
			\exp\!\left(2mk\log2-\frac{amT}{\log k}\right).
		\end{equation}
		It suffices to take
		\begin{equation}
			\label{eq:app-a2a-HD-repetition-count}
			T=
			\left\lceil
			\frac{\log k}{am}
			\left(2mk\log2+\log(\epsilon^{-1})\right)
			\right\rceil
			=\cO\!\left(
			\log k\left[k+\frac{\log(\epsilon^{-1})}{m}\right]
			\right)
		\end{equation}
		so that the $\epsilon_{\mathrm{mult}} \le \epsilon$.
		One sample of $\mu_{\mathrm{a2a},m}$ contains three diagonal circuits of depth $\cO(\log m+\log k)$ and two transversal Hadamard layers.
		Combining this depth with Eq.~\eqref{eq:app-a2a-HD-repetition-count} proves Eq.~\eqref{eq:app-a2a-m-qubit-design-depth}.
		
		Now suppose that $m<A\log k$.
		Use the random-Pauli-rotation ensemble $\nu_{\mathrm{RPR},m}$ from Fact~\ref{fact:app-fd-brickwork-gap}.
		Every sample $R(P,\theta)=e^{\mathrm i\theta P}$ has an ancilla-free all-to-all implementation of depth $\cO(\log m)$. 
		Indeed, conjugate $P$ by a tensor product of single-qubit Clifford gates so that it becomes $Z_S:=\prod_{j\in S}Z_j$ on its support $S$.
		A balanced binary tree of CNOTs computes the parity of the bits in $S$ into one of those qubits in $\cO(\log m)$ all-to-all depth.
		Applying the corresponding one-qubit $Z$ rotation and reversing the tree implements $e^{\mathrm i\theta Z_S}$.
		Conjugating the single-qubit Clifford gates gives $R(P,\theta)$.
		
		Let $g$ be the constant in Fact~\ref{fact:app-fd-brickwork-gap}.
		The same TPE-to-multiplicative conversion shows that
		\begin{equation}
			T=
			\left\lceil
			\frac{2mk\log2+\log(\epsilon^{-1})}{\log(\frac{1}{1-g})}
			\right\rceil
			\label{eq:app-a2a-rpr-repetition-count}
		\end{equation}
		repetitions suffice to make $\lambda_k(\nu_{\mathrm{RPR},m}^{*T}) \leq e^{-\log(\frac{1}{1-g})T}\leq 2^{-2mk}\epsilon$.
		The resulting all-to-all depth is
		\begin{equation}
			\cO\!\left(
			\log m[mk+\log(\epsilon^{-1})]
			\right).
			\label{eq:app-a2a-small-register-depth}
		\end{equation}
		Since $m=\cO(\log k)$, one has
		\begin{equation}
			\begin{aligned}
				m\log m
				&=\cO\!\left((\log m+\log k)\log k\right),\\
				\log m
				&=\cO\!\left(
				\frac{(\log m+\log k)\log k}{m}
				\right).
			\end{aligned}
		\end{equation}
		Substituting this bound into Eq.~\eqref{eq:app-a2a-small-register-depth} gives Eq.~\eqref{eq:app-a2a-m-qubit-design-depth}.
		Both regimes act only on the $m$ system qubits.
	\end{proof}
	
	We once again apply Fact~\ref{fact:app-fd-gluing} to glue these local designs into a global unitary design with improved depth scaling.
	
	\begin{theorem}[Low-depth approximate designs without ancillae with all-to-all connectivity]
		\label{thm:a2a-large-design}
		Fix a constant $\eta>0$.
		For every $n\geq1$, $k\geq2$, and $2^{-\eta nk}\leq\epsilon\leq 1$,
		there exists an $n$-qubit multiplicative $\epsilon$-approximate unitary $k$-design with an ancilla-free all-to-all circuit of depth
		\begin{equation}
			\label{eq:a2a-large-design-depth}
			\cO_\eta\!\left(k\log k
			\left[\log k+\log\!\left(\log(n)+\log( k)+\log(\epsilon^{-1})\right)\right]
			\right).
		\end{equation}
	\end{theorem}
	\begin{proof}
		Set $M=\log(n)+\log(k)+\log(\epsilon^{-1})$
		and choose a block scale $b:=\lceil \kappa M\rceil$ for some sufficiently large constant $\kappa>0$.
		
		First consider the case $n<2b$. Applying Theorem~\ref{thm:a2a-small-design} to the full $n$-qubit system gives a circuit depth of
		\begin{equation}
			\cO\left((\log n+\log k)\log k\left[k+\frac{\log(\epsilon^{-1})}{n}\right]\right).
		\end{equation}
		The assumption $\epsilon\geq2^{-\eta nk}$ implies $\frac{\log(\epsilon^{-1})}{n}=\cO_\eta(k)$ while $n<2b$ implies $\log n=\cO(\log M)$.
		This gives the depth bound in Eq.~\eqref{eq:a2a-large-design-depth}.
		
		It remains to treat $n\geq 2b$.
		Partition the qubits into $t:=\lfloor n/b\rfloor\geq2$ blocks $Q_1,\ldots,Q_t$. Take $|Q_i|=b$ for $i<t$ and place all remaining qubits in $Q_t$.
		Then $b\leq|Q_i|<2b$ for every $i\in[t]$.
		For $i\in[t-1]$, define the overlapping patches $R_i:=Q_i\cup Q_{i+1}$ so that $2b\leq|R_i|<4b$.
		
		Set $\epsilon':=\epsilon/(8t)$.
		Independently sample a multiplicative $\epsilon'$-approximate unitary $k$-design $U_i$ on each $R_i$ using Theorem~\ref{thm:a2a-small-design}, and arrange them in the two-layer brickwork product
		\begin{equation}
			\label{eq:app-a2a-two-layer-product}
			U_{\mathrm{glue}}
			:=
			\left(\prod_{\substack{i\in[t-1]\\i\ \mathrm{even}}}U_i\right)
			\left(\prod_{\substack{i\in[t-1]\\i\ \mathrm{odd}}}U_i\right).
		\end{equation}
		To apply Fact~\ref{fact:app-fd-gluing} iteratively, define $W_1:=U_1$ and, for $i\geq2$,
		\begin{equation}
			W_i:=
			\begin{cases}
				U_iW_{i-1},&i\text{ even},\\
				W_{i-1}U_i,&i\text{ odd}.
			\end{cases}
		\end{equation}
		Then $W_{t-1}=U_{\mathrm{glue}}$, and at the $i$-th gluing step the overlap of two unitaries is $Q_i$.
		
		Choose $\kappa$ sufficiently large that $2^b\geq 2^{\kappa M}\geq 2k^2$.
		Then Fact~\ref{fact:app-fd-gluing} applies at every gluing step.
		If $\epsilon_i$ denotes the multiplicative error of $W_i$, Fact~\ref{fact:app-fd-gluing} gives
		\begin{equation}
			1+\epsilon_i
			\leq
			(1+\epsilon_{i-1})(1+\epsilon')
			\left(1+\cO(k^22^{-|Q_i|})\right).
		\end{equation}
		Iterating this inequality gives,
		\begin{equation}\begin{aligned}
				1+\epsilon_{t-1}
				&\leq
				(1+\epsilon')^{t-1}
				\left(1+\cO(k^22^{-b})\right)^{t-2}\\
				&\leq
				\exp\!\left(t\epsilon'+\cO(tk^22^{-b})\right).
		\end{aligned}\end{equation}
		Because $t\leq n$ and $M=\log(n)+\log (k)+\log(\epsilon^{-1})$, a sufficiently large $\kappa$ makes
		\begin{equation}
			\cO(tk^22^{-b})=\cO(tk^22^{-\kappa M})\leq\epsilon/8.
		\end{equation}
		Together with $t\epsilon'\leq\epsilon/8$, this yields
		\begin{equation}
			\epsilon_{t-1}
			\leq e^{\epsilon/4}-1
			\leq\epsilon.
		\end{equation}
		Thus $U_{\mathrm{glue}}$ is a multiplicative $\epsilon$-approximate unitary $k$-design on all $n$ qubits.
		
		Finally, the patches with odd indices are pairwise disjoint, as are those with even indices.
		Each patch circuit borrows workspace only from qubits within that patch.
		So each product in Eq.~\eqref{eq:app-a2a-two-layer-product} is executed in parallel.
		For every $i$, the patch size  $|R_i|=\Theta(M)$ and the logarithmic inverse-error parameter $\log(\epsilon'^{-1})=\log(8t\epsilon^{-1})=\cO(M)$.
		Theorem~\ref{thm:a2a-small-design} therefore gives the depth of one patch as
		\begin{equation}\begin{aligned}
				\cO\!\left(
				[\log(|R_i|)+\log k]\log k
				\left[k+\frac{\log(\epsilon'^{-1})}{|R_i|}\right]
				\right)=\cO\!\left(
				k\log k[\log k+\log M]
				\right).
		\end{aligned}\end{equation}
		Since the global construction consists of only two such layers, the overall depth scaling remains unchanged, establishing Eq.~\eqref{eq:a2a-large-design-depth}.
	\end{proof}
	
	We remark that in the highly precise regime $\epsilon < 2^{-\eta nk}$, this depth scaling is preserved via the exactification lemma (see Remark~\ref{rmk:very-small-errors}), albeit at the expense of inefficient classical sampling.

	\section{Exactification of an approximate design}
	\label{app:exactification}
	In this section, we demonstrate that any sufficiently accurate approximate unitary design can be lifted to an exact design by appropriately reweighting its probability distribution. 
	Leveraging this technique, we then construct low-depth exact designs by applying it to our previously established constructions of approximate designs.
	
	\subsection{Design exactification}
	
	We now formalize and prove Lemma~\ref{lem:exactification-informal}, which establishes that any unitary ensemble achieving a sufficiently small TPE error for both orders $k$ and $2k$ can be exactified into a perfect $k$-design. The rigorous statement is as follows:
	\begin{lemma}[Design exactification, formal version of Lemma~\ref{lem:exactification-informal}]
		\label{lem:exactification}
		Let $\nu$ be a unitary ensemble on a $d$-dimensional Hilbert space.
		If
		\begin{equation}
			\lambda_k(\nu),
			\lambda_{2k}(\nu)
			\leq
			\frac{1}{32d^{4k}},
			\label{eq:app-exactification-threshold}
		\end{equation}
		then there exists a weight function
		$\omega:\operatorname{supp}\nu\rightarrow\mathbb R$ satisfying
		\begin{equation}
			\frac{3}{4}
			\leq
			\omega(U)
			\leq
			\frac{5}{4},
			\qquad
			\mathbb E_{U\sim\nu}[\omega(U)]=1,
			\label{eq:exactification-weights}
		\end{equation}
		such that
		\begin{equation}
			\mathbb E_{U\sim\nu}
			\left[
			\omega(U)
			U^{\otimes k}\otimes\overline U^{\otimes k}
			\right]
			=
			P_{\mathrm H}^{(k)}.
			\label{eq:exactified-moment}
		\end{equation}
		Consequently, the reweighted measure
		$d\nu_{\mathrm{ex}}(U):=\omega(U)d\nu(U)$ forms an exact unitary
		$k$-design supported on precisely the same unitaries as $\nu$.
	\end{lemma}
	The underlying intuition is straightforward: the order-$k$ moment controls the mean of the unitary representation, whereas the order-$2k$ moment bounds its covariance. Once both moments are sufficiently close to their ideal Haar values, a perturbation to the probability weights is sufficient to perfectly correct the residual error in the mean.
	
	\begin{proof}
		For convenience, set
		\begin{equation}
			R(U):=R_k(U)
			=
			U^{\otimes k}\otimes\overline U^{\otimes k},
			\qquad
			P:=P_{\mathrm H}^{(k)},
			\qquad
			N:=d^{2k}.
			\label{eq:app-exactification-notation}
		\end{equation}
		Thus, $R$ is an $N$-dimensional unitary representation. Fixing a column vectorization, we define the associated real feature vector as
		\begin{equation}
			x(U):=\operatorname{vec}(R(U))\in\mathbb C^{N^2},
			\qquad
			F(U):=
			\begin{pmatrix}
				\operatorname{Re}x(U)\\
				\operatorname{Im}x(U)
			\end{pmatrix}
			\in\mathbb R^{2N^2}.
			\label{eq:app-real-feature}
		\end{equation}
		Next, we let $\overline F_{\mathrm H} := \mathbb E_{U\sim \mathrm H}[F(U)]$ denote the Haar expectation and center the feature vector as $z(U):=F(U)-\overline F_{\mathrm H}$. This allows us to introduce the relevant real feature space:
		\begin{equation}
			E
			:=
			\operatorname{span}_{\mathbb R}
			\{z(U):U\in\mathrm U(d)\},
			\qquad
			q:=\dim E\leq2N^2.
		\end{equation}
		The Haar covariance matrix on this space is given by
		\begin{equation}
			C_{\mathrm H}
			:=
			\mathbb E_{U\sim \mathrm H}
			\left[z(U)z(U)^{\mathsf T}\right].
			\label{eq:app-haar-feature-covariance}
		\end{equation}
		In what follows, all covariance inverses are implicitly restricted to the subspace $E$. 
		Furthermore, because $F(U)$ is not constant, we may safely assume $q \ge 1$.
		
		\smallskip
		\noindent\emph{Haar covariance.}---We first establish two key properties of $C_{\mathrm H}$:
		\begin{equation}
			C_{\mathrm H}\succeq\frac{1}{N} I_E,
			\label{eq:app-haar-covariance-lower}
		\end{equation}
		and
		\begin{equation}
			z(U)^{\mathsf T}C_{\mathrm H}^{-1}z(U)=q, 
			\qquad
			\forall U\in\mathrm U(d).
			\label{eq:app-constant-leverage}
		\end{equation}

		To verify Eq.~\eqref{eq:app-haar-covariance-lower}, we decompose the moment representation into its irreducible components:
		\begin{equation}
			R(U)
			\simeq
			I_{m_0}
			\oplus
			\bigoplus_{\lambda\neq0}
			I_{m_\lambda}\otimes\pi_\lambda(U),
			\qquad
			d_\lambda:=\dim\pi_\lambda.
		\end{equation}
		Notice that the character of the moment representation is real-valued:
		\begin{equation}
			\operatorname{tr}[R(U)]=|\operatorname{tr}(U)|^{2k}\in\mathbb R.
			\label{eq:app-real-character}
		\end{equation}
		Consequently, for every matrix $Y$ whose vectorized realification lies in $E$, the overlap $\operatorname{vec}_{\mathbb R}(Y)^{\mathsf T} z(U) = \operatorname{tr}[Y^\dagger(R(U)-P)]$ is real. By the definition of the Haar covariance matrix, 
		\begin{equation}
			\operatorname{vec}_{\mathbb R}(Y)^{\mathsf T} C_{\mathrm H} \operatorname{vec}_{\mathbb R}(Y)
			=
			\mathbb{E}_{\mathrm H} \left[ \left| \operatorname{tr}[Y^\dagger(R(U)-P)] \right|^2 \right].
		\end{equation}
		We can evaluate this expectation by expanding $Y = \sum_{\lambda \neq 0, ij} y_{\lambda,ij} B_{\lambda,ij}$ in the Hilbert--Schmidt orthonormal coefficient basis for the non-trivial blocks,
		\begin{equation}
			B_{\lambda,ij}
			:=
			\frac{1}{\sqrt{m_\lambda}}\,
			I_{m_\lambda}\otimes|i\rangle\langle j|,
		\end{equation}
		which yields $\operatorname{tr}[B_{\lambda,ij}^\dagger (R(U)-P)] = \sqrt{m_\lambda} [\pi_\lambda(U)]_{ij}$. Substituting this expansion into the expectation and applying Schur orthogonality for the irreducible matrix elements gives
		\begin{align}
			\operatorname{vec}_{\mathbb R}(Y)^{\mathsf T}
			C_{\mathrm H}
			\operatorname{vec}_{\mathbb R}(Y)
			&=
			\sum_{\lambda\neq0}
			\frac{m_\lambda}{d_\lambda}
			\sum_{i,j}|y_{\lambda,ij}|^2 \nonumber\\
			&\geq
			\frac{1}{N}
			\sum_{\lambda\neq0, i,j}|y_{\lambda,ij}|^2
			=
			\frac{1}{N}\|Y\|_{\mathrm F}^2,
			\label{eq:app-schur-covariance}
		\end{align}
		where the inequality follows from $m_\lambda\geq1$ and $d_\lambda\leq N$. This completes the proof of Eq.~\eqref{eq:app-haar-covariance-lower}.

		For Eq.~\eqref{eq:app-constant-leverage}, left multiplication by
		$R(V)$ induces an orthogonal map on $E$ that sends $z(U)$ to $z(VU)$
		and preserves $C_{\mathrm H}$. Hence
		$z(U)^{\mathsf T}C_{\mathrm H}^{-1}z(U)$ is independent of $U$.
		Averaging it over Haar measure gives
		\begin{equation}
			\mathbb E_{U\sim \mathrm H}
			\left[
			z(U)^{\mathsf T}C_{\mathrm H}^{-1}z(U)
			\right]
			=
			\operatorname{tr}
			\left(C_{\mathrm H}^{-1}C_{\mathrm H}\right)
			=
			q,
		\end{equation}
		which proves Eq.~\eqref{eq:app-constant-leverage}.
		
		\smallskip
		\noindent\emph{Mean and covariance under $\nu$.}--Define the feature mean and its deviation from the Haar mean as
		\begin{equation}
			\overline F_\nu
			:=
			\mathbb E_{U\sim\nu}[F(U)],
			\qquad
			\Delta
			:=
			\overline F_{\mathrm H}-\overline F_\nu,
			\label{eq:app-feature-mean-error}
		\end{equation}
		and the corresponding covariance matrix under $\nu$ as
		\begin{equation}
			C_\nu
			:=
			\mathbb E_{U\sim\nu}
			\left[
			(F(U)-\overline F_\nu)
			(F(U)-\overline F_\nu)^{\mathsf T}
			\right].
			\label{eq:app-feature-covariance}
		\end{equation}
		Let $\alpha:=\lambda_k(\nu)$ and $\beta:=\lambda_{2k}(\nu)$. 
		Because $\Delta$ represents the vectorized realification of $P-M_k(\nu)$, the standard inequality between the Frobenius and operator norms yields
		\begin{equation}
			\|\Delta\|_2
			\leq
			\sqrt{N}\,\alpha.
			\label{eq:app-feature-mean-bound}
		\end{equation}
		Combining this with Eq.~\eqref{eq:app-haar-covariance-lower}, we obtain
		\begin{equation}
			a
			:=
			\sqrt{
				\Delta^{\mathsf T}C_{\mathrm H}^{-1}\Delta
			}
			\leq
			N\alpha.
			\label{eq:app-whitened-mean-bound}
		\end{equation}
		
		We next compare the covariance matrices. Up to fixed permutations of the tensor factors, we have
		\begin{equation}
			R(U)\otimes R(U)\simeq R_{2k}(U),
			\qquad
			R(U)\otimes\overline{R(U)}
			\simeq R_{2k}(U).
			\label{eq:app-second-moment-representations}
		\end{equation}
		Consequently, the order-$2k$ TPE error $\beta$ controls both complex second moments of the vectorized representation $x(U)=\operatorname{vec}(R(U))$.
		More explicitly, for $\eta\in\{\nu,\mathrm H\}$, define
		\begin{equation}
			A_\eta
			:=
			\mathbb E_\eta[x(U)x(U)^\dagger],
			\qquad
			B_\eta
			:=
			\mathbb E_\eta[x(U)x(U)^{\mathsf T}],
			\qquad
			S_\eta
			:=
			\mathbb E_\eta[F(U)F(U)^{\mathsf T}].
		\end{equation}
		The matrices $A_\eta$ and $B_\eta$ are obtained from the averages of
		$R\otimes\overline R$ and $R\otimes R$, respectively, by a fixed
		realignment of matrix indices. If $\mathcal Q$ denotes this
		realignment, then
		\begin{equation}
			\|\mathcal Q(T)\|_{\mathrm{op}}
			\leq
			\|\mathcal Q(T)\|_{\mathrm F}
			=
			\|T\|_{\mathrm F}
			\leq
			N\|T\|_{\mathrm{op}}.
		\end{equation}
		It follows from Eq.~\eqref{eq:app-second-moment-representations} that
		\begin{equation}
			\|A_\nu-A_{\mathrm H}\|_{\mathrm{op}},
			\,
			\|B_\nu-B_{\mathrm H}\|_{\mathrm{op}}
			\leq
			N\beta.
			\label{eq:app-complex-second-moment-bound}
		\end{equation}
		Moreover, writing
		$A:=A_\nu-A_{\mathrm H}$ and $B:=B_\nu-B_{\mathrm H}$, we have
		\begin{equation}
			S_\nu-S_{\mathrm H}
			=
			\frac12
			\begin{pmatrix}
				\operatorname{Re}(A+B)
				&
				\operatorname{Im}(B-A)
				\\
				\operatorname{Im}(A+B)
				&
				\operatorname{Re}(A-B)
			\end{pmatrix}.
		\end{equation}
		Consequently,
		\begin{equation}
			\|S_\nu-S_{\mathrm H}\|_{\mathrm{op}}
			\leq
			N\beta.
			\label{eq:app-real-second-moment-bound}
		\end{equation}
		
		Let $\Pi_E$ be the real orthogonal projector onto the subspace $E$. Recall that $F(U)=z(U)+\overline F_{\mathrm H}$, and note that $\overline F_{\mathrm H}$ is orthogonal to $E$ because the Haar moment $P$ is an orthogonal projector ($P(R(U)-P)=0$). Thus, applying $\Pi_E$ to the uncentered second moment $S_\eta=\mathbb E_\eta[F(U)F(U)^{\mathsf T}]$ annihilates all terms involving $\overline F_{\mathrm H}$, leaving
		\begin{equation}
			\Pi_E S_\eta \Pi_E
			=
			\mathbb E_\eta\left[z(U)z(U)^{\mathsf T}\right].
		\end{equation}
		Using the covariance relations $C_{\mathrm H} = \mathbb E_{U\sim \mathrm H}[z(U)z(U)^{\mathsf T}]$ and $C_\nu = \mathbb E_{U\sim\nu}[z(U)z(U)^{\mathsf T}] - \Delta\Delta^{\mathsf T}$, we have
		\begin{equation}
			C_\nu-C_{\mathrm H}
			=
			\Pi_E(S_\nu-S_{\mathrm H})\Pi_E
			-
			\Delta\Delta^{\mathsf T}.
		\end{equation}
		Therefore,
		\begin{equation}
			\|C_\nu-C_{\mathrm H}\|_{\mathrm{op}}
			\leq
			N\beta+N\alpha^2 \leq
			\frac{1}{4N}.
			\label{eq:app-covariance-comparison}
		\end{equation}
		Combining this with Eq.~\eqref{eq:app-haar-covariance-lower}, we obtain
		\begin{equation}
			C_\nu\succeq\frac12C_{\mathrm H},
			\qquad
			C_\nu^{-1}\preceq2C_{\mathrm H}^{-1}
			\quad\text{on }E.
			\label{eq:app-covariance-inverse-comparison}
		\end{equation}
		
		\smallskip
		\noindent\emph{Correcting the weights.}---
		We now define the reweighting function
		\begin{equation}
			\omega(U) := 1+ 
			\Delta^{\mathsf T}C_\nu^{-1}
			\bigl(F(U)-\overline F_\nu\bigr).
			\label{eq:app-correcting-weight}
		\end{equation}
		We will verify its positivity, normalization, and exactness in turn.
		
		First, we decompose the deviation as $F(U)-\overline F_\nu = z(U) + \Delta$. Then,
		\begin{align}
			\bigl(F(U)-\overline F_\nu\bigr)^{\mathsf T}
			C_\nu^{-1}
			\bigl(F(U)-\overline F_\nu\bigr)
			\nonumber
			&\leq
			2 \bigl(z(U) + \Delta\bigr)^{\mathsf T} C_{\mathrm H}^{-1} \bigl(z(U) + \Delta\bigr)
			\nonumber\\
			&\leq
			2\left(
			\sqrt{z(U)^{\mathsf T}C_{\mathrm H}^{-1}z(U)} + 
			\sqrt{\Delta^{\mathsf T}C_{\mathrm H}^{-1}\Delta}
			\right)^2.
		\end{align}
		Substituting Eq.~\eqref{eq:app-constant-leverage} and the definition of $a$ from Eq.~\eqref{eq:app-whitened-mean-bound}, this evaluates to $2(\sqrt{q}+a)^2$. 
		Furthermore, because $q \geq 1$ and $a \leq N\alpha \leq \frac{1}{32N} \leq \sqrt{q}$, we have $(\sqrt{q}+a)^2 \leq (2\sqrt{q})^2 = 4q$. This gives
		\begin{equation}
			\bigl(F(U)-\overline F_\nu\bigr)^{\mathsf T}
			C_\nu^{-1}
			\bigl(F(U)-\overline F_\nu\bigr)
			\leq
			8q.
			\label{eq:app-corrected-leverage}
		\end{equation}
		Using Cauchy--Schwarz in the $C_\nu^{-1}$ inner product, together with
		$q\leq2N^2$, gives
		\begin{align}
			|\omega(U)-1|
			&\leq
			\|\Delta\|_{C_\nu^{-1}}
			\|F(U)-\overline F_\nu\|_{C_\nu^{-1}}
			\nonumber\\
			&\leq
			4a\sqrt q
			\leq
			8N^2\alpha
			\leq
			\frac14.
			\label{eq:app-positive-weight-bound}
		\end{align}
		Thus $\frac34\leq\omega(U)\leq\frac54$ for every $U\in\operatorname{supp}\nu$.
		
		Second, because
		$\mathbb E_\nu[F-\overline F_\nu]=0$,
		Eq.~\eqref{eq:app-correcting-weight} immediately gives $\mathbb E_\nu\omega(U)=1$.
		
		Finally, let $w(U):=F(U)-\overline F_\nu$. Then
		$\mathbb E_\nu w=0$ and
		$\mathbb E_\nu[ww^{\mathsf T}]=C_\nu$. Therefore
		\begin{align}
			\mathbb E_\nu[\omega(U)F(U)]
			&=
			\overline F_\nu
			+
			\mathbb E_\nu[ww^{\mathsf T}]
			C_\nu^{-1}\Delta
			\nonumber\\
			&=
			\overline F_\nu+\Delta
			=
			\overline F_{\mathrm H}.
			\label{eq:app-exact-feature-mean}
		\end{align}
		Since realification is injective, this is equivalent to
		\begin{equation}
			\mathbb E_{U\sim\nu}
			\left[
			\omega(U)
			U^{\otimes k}\otimes\overline U^{\otimes k}
			\right]
			=
			P_{\mathrm H}^{(k)}.
		\end{equation}
	\end{proof}

	\subsection{Construction of low-depth exact designs}
	
	
	

	We now apply Lemma~\ref{lem:exactification} to our implementation of the $m$-qubit approximate unitary designs in Theorem~\ref{thm:fd-small-design} and Theorem~\ref{thm:a2a-small-design}, and obtain the following low-depth exact unitary design construction:
	
	\begin{theorem}[Low-depth exact designs without ancillae, formal version of Theorem~\ref{thm:exact-design-informal}]\label{thm:exact-design-formal}
		For every $n\geq1$ and $k\geq2$, there exists an $n$-qubit exact unitary $k$-design with ancilla-free circuit depth
		\begin{equation}
			\label{eq:app-exact-fd-depth}
			\cO_\delta\left(n^{1/\delta}k\log^2 k\right)
		\end{equation}
		on a $\delta$-dimensional architecture, and an $n$-qubit exact unitary $k$-design with ancilla-free circuit depth
		\begin{equation}
			\label{eq:app-exact-a2a-depth}
			\cO(k\log k(\log n+\log k)).
		\end{equation}
	\end{theorem}
	
	\begin{proof}
		Denote $d:=2^n$ and set 
		\begin{equation}
			\epsilon_{\star}:=\frac{1}{64d^{5k}}.
		\end{equation}
		Consider a multiplicative $\epsilon_\star$-approximate unitary $2k$-design $\nu$ on this Hilbert space.
		Such an ensemble is also a multiplicative $\epsilon_\star$-approximate unitary $k$-design.
		To see this, define two channels
		\begin{equation}
			\mathcal J(X):=
			X\otimes\frac{I_{d^k}}{d^k},
			\qquad
			\mathcal T(Y):=
			\operatorname{Tr}_{k+1,\ldots,2k}(Y).
		\end{equation}
		For either $\eta=\nu$ or the Haar measure, their $k$-th and $2k$-th moment operators satisfy
		\begin{equation}
			\Phi_\eta^{(k)}=\mathcal T\circ\Phi_\eta^{(2k)}\circ\mathcal J.
		\end{equation}
		Since $\mathcal J$ and $\mathcal T$ are completely positive, the multiplicative inequalities at order $2k$ imply the same inequalities at order $k$.
		
		Applying Fact~\ref{fact:TPE-multiplicative-comparison} gives
		\begin{equation}
			\lambda_r(\nu)
			\leq 2d^{r/2}\epsilon_\star
			\leq 2d^k\epsilon_\star
			=\frac{1}{32d^{4k}}
		\end{equation}
		at order $r=k$ and $r=2k$.
		Thus every multiplicative $\epsilon_\star$-approximate unitary $2k$-design $\nu$ satisfies the hypotheses of Lemma~\ref{lem:exactification}. 
		Its support can consequently be reweighted to form an exact unitary $k$-design.
		
		For a $\delta$-dimensional architecture, by Theorem~\ref{thm:fd-small-design}, there is an $n$-qubit multiplicative $\epsilon_\star$-approximate unitary $2k$-design with ancilla-free circuit depth
		\begin{equation}
			\cO_\delta\left(n^{1/\delta}\left(2k+\frac{\log(\epsilon_\star^{-1})}{n}\right)\log^2(2k)\right)= \cO_\delta\left(n^{1/\delta}k\log^2 k\right).
		\end{equation}    
		For the all-to-all architecture, by Theorem~\ref{thm:a2a-small-design}, there is an $n$-qubit multiplicative $\epsilon_\star$-approximate unitary $2k$-design with ancilla-free circuit depth
		\begin{equation}
			\cO\left(
			\left(\log n+\log(2k)\right)\log(2k)\left[2k+\frac{\log(\epsilon_\star^{-1})}{n}\right]\right)=\cO(k\log k(\log n+\log k)).
		\end{equation}
		As exactification only reweights the ensemble $\nu$ and hence preserves the depth and the absence of ancillary qubits, this gives the claimed depth bound in Eq.~\eqref{eq:app-exact-fd-depth} and Eq.~\eqref{eq:app-exact-a2a-depth}.
	\end{proof}
	
\end{document}